\documentclass[11pt,onehalfspacing,notitlepage]{ay_article}

\usepackage{adjustbox}
\usepackage{amsbsy}
\usepackage{amsmath}
\usepackage{amssymb}
\usepackage{amsthm}
\usepackage{array}
\usepackage{bigstrut}
\usepackage{booktabs}
\usepackage{makecell}
\usepackage{blkarray}
\usepackage{bm}
\usepackage{cancel}
\usepackage{commath}
\usepackage{chngcntr}
\usepackage{dsfont}
\usepackage{econometrics}
\usepackage{enumerate}
\usepackage{fancyhdr}
\usepackage{float}
\usepackage{gensymb}
\usepackage{graphicx}
\usepackage{IEEEtrantools}
\usepackage{longtable}
\usepackage[utf8]{inputenc}
\usepackage{marginnote, csquotes, empheq}
\usepackage{xurl}

\usepackage{acro}
\usepackage{subcaption}

\usepackage{mathrsfs}
\usepackage{mathtools}
\usepackage{mdframed}
\usepackage{multirow}
\usepackage{nccmath}
\usepackage{natbib}
\usepackage{parskip}
\usepackage{pdfpages}
\usepackage{pgf}
\usepackage{pifont}
\usepackage{setspace}
\usepackage{stackengine}
\usepackage{tabularx}
\usepackage{textcomp}
\usepackage{verbatim}
\usepackage{geometry}
\usepackage{oubraces}
\usepackage{zref-xr}
\usepackage{placeins}
\usepackage{pifont}
\usepackage[colorlinks=true,linkcolor={blue!70!black},urlcolor={blue!70!black},anchorcolor={blue!70!black},citecolor={blue!70!black}]{hyperref}
\usepackage{rotating}
\usepackage{ragged2e}
\usepackage{etoolbox}
\usepackage{xr}
\usepackage{tcolorbox}
\usepackage{colortbl}
\usepackage{multirow}
\usepackage{pdfpages}
\usepackage{float}
\usepackage{longtable}
\usepackage{nicefrac} 
\usepackage[us, 12hr]{datetime}
\usepackage{tikz}
\usetikzlibrary{intersections, arrows.meta,calc}
\usepackage{pgfplots}
\pgfplotsset{width=10cm,compat=1.18}

\renewcommand{\P}{P}

\definecolor{uclablue}{HTML}{2774AE}
\definecolor{ucladarkblue}{RGB}{0, 85, 135}
\definecolor{uclagold}{RGB}{255, 184, 0}
\definecolor{pastelred}{HTML}{d62d0e}
\definecolor{pastelgreen}{RGB}{2,117,36}

\theoremstyle{definition}

\newtheorem{proposition}{Proposition}

\newtheorem{hypothesis}{Hypothesis}

\theoremstyle{remark}

\newcommand{\Mquote}[1]{\enquote{\textit{#1}}}

\makeatletter
\newcommand{\p}{%
  \@ifnextchar0{\p@digit}{%
  \@ifnextchar1{\p@digit}{%
  \@ifnextchar2{\p@digit}{%
  \@ifnextchar3{\p@digit}{%
  \@ifnextchar4{\p@digit}{%
  \@ifnextchar5{\p@digit}{%
  \@ifnextchar6{\p@digit}{%
  \@ifnextchar7{\p@digit}{%
  \@ifnextchar8{\p@digit}{%
  \@ifnextchar9{\p@digit}{\p@err}}}}}}}}}}}
\newcommand{\p@digit}[1]{%
  \@ifnextchar c{\p@digitc{#1}}{%
  \@ifnextchar C{\p@digitc{#1}}{\p@plain{#1}}}}
\newcommand{\p@digitc}[2]{\p@maybe@math{\mathcal{P}_{#1}^{c}}}
\newcommand{\p@plain}[1]{\p@maybe@math{\mathcal{P}_{#1}}}
\newcommand{\p@maybe@math}[1]{%
  \ifmmode
    #1%
  \else
    \ensuremath{#1}%
  \fi
}

\newcommand{\p@err}{\PackageError{mypkg}{Invalid \string\p usage}{Use \string\p<digit> or \string\p<digit>c}}
\makeatother

\DeclareMathAlphabet{\mathitbf}{OML}{cmm}{b}{it}
\makeatletter
\newcommand{\moverlay}{\mathpalette\mov@rlay}
\newcommand{\mov@rlay}[2]{\leavevmode\vtop{\baselineskip\z@skip{}
\lineskiplimit-\maxdimen\ialign{\hfil\ensuremath{#1##}\hfil\cr#2\cr\cr}}}
\makeatother

\DeclareAcronym{us}{short = US, long = United States}
\DeclareAcronym{frc}{short = FRC, long = Federal Radio Commission}
\DeclareAcronym{fcc}{short = FCC, long = Federal Communications Commission}
\DeclareAcronym{rhc}{short = RHC, long = Rural Health Care Program}
\DeclareAcronym{usf}{short = USF, long = Universal Service Fund}
\DeclareAcronym{hcp}{short = HCP, long = health care provider}
\DeclareAcronym{isp}{short = ISP, long = Internet Service Provider}
\DeclareAcronym{usac}{short = USAC, long = Universal Service Administrative Company}
\DeclareAcronym{mpls}{short = MPLS, long = Multiprotocol Label Switching}
\DeclareAcronym{mbps}{short = Mbps, long = megabit per second}
\DeclareAcronym{gam}{short = GAM, long = Generalized Additive Model}
\DeclareAcronym{frn}{short = FRN, long = Funding Request Number}
\DeclareAcronym{frnln}{short = FRNLN, long = FRN Line Number}
\DeclareAcronym{ols}{short = OLS, long = Ordinary Least Squares}
\DeclareAcronym{pols}{short = POLS, long = Pooled Ordinary Least Squares}
\DeclareAcronym{lhs}{short = LHS, long = left-hand side}
\DeclareAcronym{fe}{short = FE, long = fixed effects}
\DeclareAcronym{rfp}{short = RFP, long = Request For Proposal}
\DeclareAcronym{fwl}{short = FWL, long = Frisch–Waugh–Lovell}
\DeclareAcronym{hcf}{short = HCF, long = Healthcare Connect Fund}
\DeclareAcronym{twfe}{short = TWFE, long = Two-Way Fixed Effects}
\DeclareAcronym{dml}{short = DML, long = Double/\allowbreak Debiased Machine Learning}
\DeclareAcronym{plr}{short = PLR, long = partially linear regression}
\DeclareAcronym{did}{short = DiD, long = difference-in-difference}

\begin{document}
\setlength{\abovedisplayskip}{1pt}
\setlength{\belowdisplayskip}{1pt}
\setlength{\abovedisplayshortskip}{1pt}
\setlength{\belowdisplayshortskip}{1pt}
\setlength\parindent{24pt}
\AtBeginEnvironment{table}{\small\renewcommand{\arraystretch}{1}}
\AtBeginEnvironment{table*}{\small\renewcommand{\arraystretch}{1}}


\begin{titlepage}
\title{\large{\center{\textbf{Price Cap vs. Per-Unit Subsidies: Selection, Pricing, and Cross Subsidization}}}\\

\author{
Ram Sewak Dubey \textcircled{r} Maysam Rabbani \textcircled{r} Rodrigo Pinto%
\thanks{Department of Economics, Feliciano School of Business, Montclair State University, Montclair, NJ, USA (Dubey, Rabbani);
Department of Economics, University of South Florida, FL, USA (Pinto).
We are grateful for the interaction with Thomas E Cone, the seminar participants at Montclair State University, University of South Florida, New York State Economics Association, and Southern Economics Association.
Author names have been randomized. The official randomization link \href{https://shorturl.at/HdYKL}{may be found here}.
}\\
}
}

\date{\today}
\end{titlepage}
\maketitle
\begin{abstract}
\singlespacing

We evaluate subsidy mechanisms in the FCC's Rural Health Care program using administrative data covering the full population of participants. The original price-cap mechanism removes cost-containment incentives for health care providers. An ad valorem mechanism introduced in 2014 addresses this flaw by making providers bear 35\% of costs. However, allowing consortium applications creates a new distortion: cross-subsidization from eligible to ineligible members. We develop theoretical models predicting these effects and estimate treatment effects using an extension of the two-way fixed effects framework with continuous treatments. We find that the ad valorem mechanism substantially reduces program spending relative to the price cap, while the consortium option significantly inflates it.
Enforcement records and an inverted U-shaped relationship between cross-subsidization intensity and ineligible member share corroborate the findings.

\noindent \emph{Keywords:} \texttt{Policy Evaluation},\;\texttt{Subsidies}, \:\texttt{Telecommunication}.

\end{abstract}

\clearpage

\newpage

\section{Introduction}
\label{sec:intro}

The \acf{rhc} spends over \$700 million per year subsidizing internet connectivity for rural \acfp{hcp} across the United States, yet its original subsidy mechanism insulates recipients from price increases, eliminating any incentive to contain costs.
Under the program's original price-cap scheme, a rural hospital pays the urban benchmark price regardless of how much the \acf{isp} charges, and the government covers the difference.
Any price reduction benefits only the government, since the hospital's cost is unchanged.
In 2014, policymakers introduced an ad valorem alternative in which hospitals pay 35\% of the negotiated price, restoring the elasticity of demand.
But the reform also permitted \textit{consortium} applications that mix eligible and ineligible providers, inadvertently opening a new channel for subsidy extraction: cross-subsidization.

This paper provides the first comprehensive economic evaluation of the \ac{rhc} program.
We exploit the 2014 introduction of the \acf{hcf} as a policy change that, for the first time, offered every price-cap participant the option of staying on the price cap (\p1) or switching to the ad valorem mechanism (\p2) or its consortium variant (\p2c).
Because all \acp{hcp} were on \p1\ in 2013 and the new mechanisms became available simultaneously in 2014, the only margins of switching are \p1$\to$\p2\ and \p1$\to$\p2c, and identification rests on a single pre/post comparison free of staggered-adoption bias.\footnote{In the balanced 2013--2014 panel of incumbents (970 \acp{hcp}; 1{,}940 HCP-year observations), 419 \acp{hcp} allocate some share of their activity to \p2\ and 43 to \p2c; treatment intensity is the share of each HCP's subsidized bandwidth purchased under each new mechanism.}
We document two main findings.
First, \acp{hcp} that switched from \p1\ to \p2\ experienced an average reduction of 71\% in prices and a 51\% reduction in \ac{hcp} net cost (price minus subsidy), consistent with the prediction that restoring demand elasticity lowers the monopolist's equilibrium markup.
Second, \acp{hcp} that switched from \p1\ to \p2c\ saw their net costs increase by 470\% on average, despite moving to a mechanism that should, absent manipulation, deliver the same ad valorem savings.
The cost increase is difficult to rationalize without cross-subsidization from eligible to ineligible consortium members.

These findings speak to a general design vulnerability: whenever a subsidy program permits joint applications by eligible and ineligible participants, the group can strategically reallocate costs to maximize subsidy extraction.
The mechanism is simple.
By inflating the reported price for eligible members and deflating it for ineligible ones, a consortium extracts higher subsidies without changing the \ac{isp}'s total revenue.
The gains from manipulation are bounded only by the risk of detection, which in this program appears weak.
Consider GCI Communication Corp, which paid a \$42.1 million settlement for price misreporting (Appendix~\ref{sec:cross_sub}). The penalty amounted to just 3.4\% of the subsidies the firm collected over the relevant period, leaving the vast majority of the gains from misreporting intact.

We develop a theoretical model of monopoly pricing under the three mechanisms to derive testable implications.
Under the price cap, the provider's demand is pinned at the benchmark price, so the monopolist inflates the billed price until the marginal enforcement penalty equals the marginal revenue from overcharging.
Under the ad valorem subsidy, the \ac{hcp} bears a fixed percentage of the price, which anchors the monopolist's markup to the demand curve; we show that the ad valorem mechanism dominates the price cap in consumer price, quantity, provider expenditure, and government outlays.
For consortia, we model cross-subsidization as a multiplicative distortion applied to the eligible member's price, with detection risk increasing and convex in the distortion.
The optimal distortion is inverted U-shaped in the ratio of ineligible to eligible revenue within the consortium: when the ratio is small, there is insufficient ineligible revenue to shift; when it is large, the marginal penalty exceeds the marginal saving.

Our empirical strategy compares outcomes in 2013 (when only \p1\ existed) with outcomes in 2014 (when \p2\ and \p2c\ first became available), using the full population of \ac{rhc} subsidy requests recorded in \ac{fcc}-mandated administrative data.
We extend the conventional \ac{twfe} framework to accommodate two simultaneous continuous treatment intensities, namely the shares of \ac{hcp} activity under \p2\ and \p2c, rather than a single binary indicator. We verify the parallel trends assumption following \citet{aryal2025benefits} and \citet{manski2018right}.

We supplement these estimates with double/debiased machine learning (\acs{dml}, \citealp{chernozhukov2018double}) to allow for nonlinear covariate relationships and report robustness checks across \ac{hcp} types, service types, and sample restrictions.
Switching remains a voluntary decision, so our estimates, which target the average treatment effect on the treated, may reflect selection into mechanisms.
We address this concern through several complementary approaches: an institutional argument that the three programs differ only in their reimbursement formula, a logit analysis showing that a single mechanical cost signal dominates the switching decision, a sensitivity analysis that allows proportional departures from parallel trends \citep{aryal2025benefits}, and a theory-derived test of cross-subsidization that does not rely on parallel trends.

The evidence for cross-subsidization rests on four complementary sources.
First, the theoretical model predicts that consortia inflate eligible members' prices whenever detection is imperfect.
Second, the empirical estimates show that switching to \p2c\ increased \ac{hcp} net cost, a pattern that is difficult to rationalize since the ad valorem rate is identical under \p2\ and \p2c.
Third, a local linear regression of cost distortion on consortium composition reveals an inverted U-shape, exactly as the model predicts.
Fourth, \ac{fcc} enforcement records, including settlements for price misreporting in consortium applications, provide direct institutional corroboration of the mechanism.
No single piece of evidence is conclusive on its own; together, they form a coherent case suggesting that cross-subsidization is economically significant and consistent with the theoretical predictions.

The cross-subsidization channel we identify requires three conditions: subsidies tied to a reported price that participants can manipulate, a group structure that permits internal reallocation of costs across eligible and ineligible members, and enforcement too weak to deter distortion at the margin.
These conditions are not unique to telecommunications.
The \ac{usf}'s own E-Rate program satisfies all three.
E-Rate averages discount rates across consortium members ranging from 20\% to 90\% based on school demographics; the averaged discount is mechanically manipulable through consortium composition, costs are allocated across eligible and ineligible entities using self-reported usage measures, and enforcement has historically been weak enough to sustain vendor overcharges of 400--500\% before detection \citep{goolsbee2006impact, hazlett2019educational}.
The cross-subsidization channel our framework predicts has not been empirically tested in E-Rate, making it a direct out-of-sample implication.
More broadly, similar conditions may arise in health care, housing, and energy programs that permit joint applications by heterogeneous participants.

Our paper contributes to two literatures.
First, we advance the empirical literature on subsidy design and program effectiveness \citep{eriksson1998targeted, goolsbee2006impact, gruber2005subsidies, fairlie2013experimental, mendez2021impacts, caliendo2011start} by providing causal evidence from a large-scale public program on how price-cap versus ad valorem subsidies affect negotiated prices and participant costs.
Second, we identify a novel form of regulatory leakage, namely cross-subsidization from eligible to ineligible consortium members, that complements the literature on unintended consequences of government interventions \citep{vernon1979unintended, cicala2019regulating, fremeth2018spillovers}.
On the methodological side, we also contribute by extending \ac{twfe} to handle multiple simultaneous continuous treatments and combining it with \ac{dml}, complementing recent advances in causal estimation under complex treatment structures \citep{chernozhukov2018double, aryal2025benefits, manski2018right}.
Our cross-subsidization channel is the mirror image of the mechanism in \citet{rotemberg2019equilibrium}: whereas subsidies to eligible firms crowd out ineligible rivals in his framework, ineligible entities benefit \textit{from within} the subsidized group in ours.

The paper proceeds as follows.
Section~\ref{sec:background} describes the institutional environment.
Section~\ref{sec:model} presents the theoretical model.
Section~\ref{sec:treatment} develops the econometric framework.
Section~\ref{sec:apply} reports empirical results.
Section~\ref{sec:discussion} discusses policy implications.
Section~\ref{sec:conclusions} concludes.
\section{Institutional environment\protect\footnote{Appendix \protect\ref{Apx:institutional} provides a more detailed version of this section.}}
\label{sec:background}

The \ac{fcc} was established in 1934 with the mandate of making rapid, efficient, nationwide communications available to all people in the United States at reasonable charges (\href{https://uscode.house.gov/view.xhtml?req=granuleid:USC-prelim-title47-section151&num=0&edition=prelim}{47 U.S.C. \S~151}).
The Telecommunications Act of 1996 created an umbrella program, the \ac{usf}, comprising four sister programs that provide targeted subsidies to distinct beneficiary groups (Figure \ref{fig:schematic}): 
the High Cost Program reimburses telecommunication firms serving high-cost areas; 
E-Rate subsidizes internet access for schools and libraries; 
the Lifeline Program supports low-income consumers with internet and mobile service bills; 
and the \ac{rhc} subsidizes internet access and related equipment for eligible \acp{hcp}.

\begin{figure}[htbp]
    \centering
    \begin{minipage}{\linewidth}
    \centering
    \includegraphics[width=0.65\linewidth]{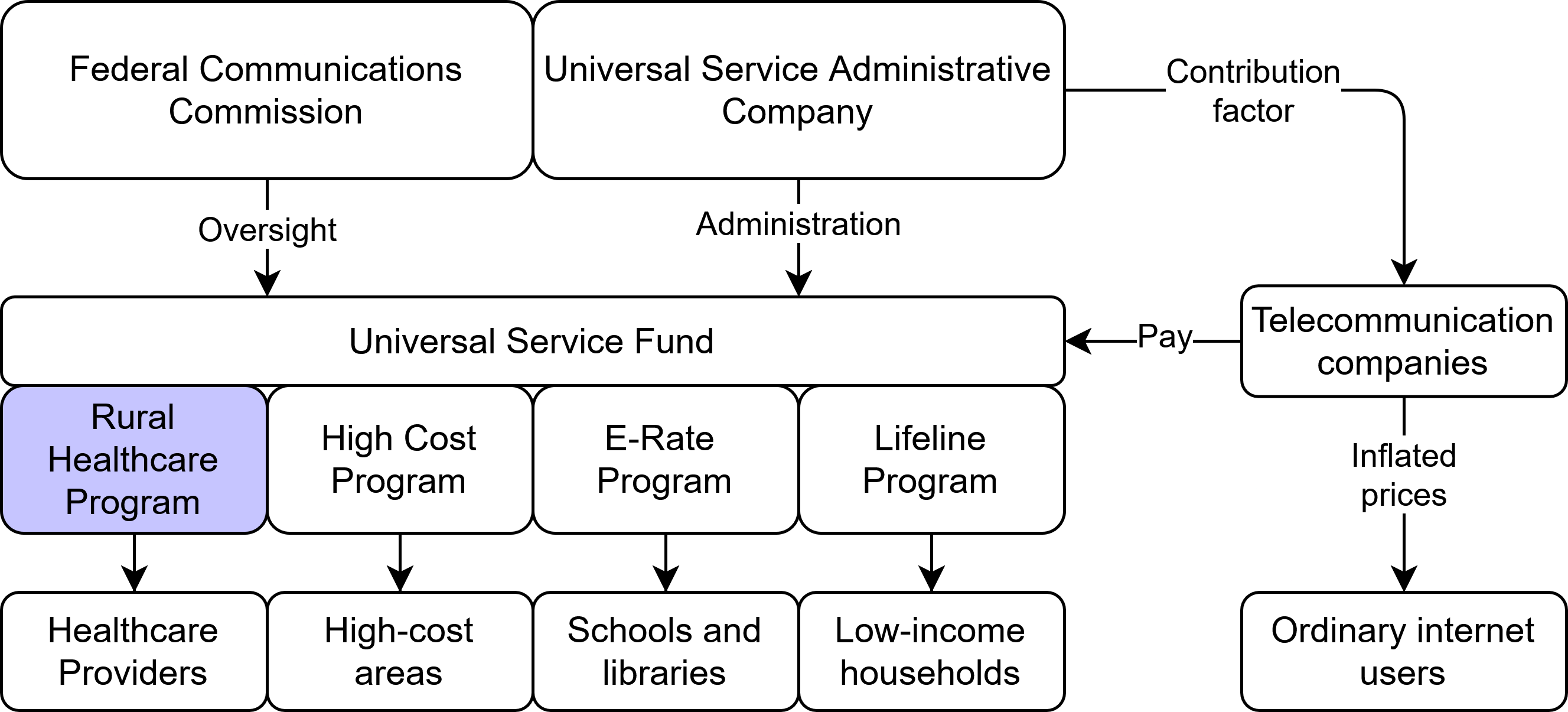}
    \caption{The regulatory structure and key players.}
    \label{fig:schematic}
    \vspace{2pt}
    \parbox{\linewidth}{\scriptsize\justifying\textit{Notes:} The Rural Health Care Program is one of the four targeted subsidy programs under the Universal Service Fund. The programs are administered by the Universal Service Administrative Company, under the oversight of the Federal Communications Commission, and funded by taxing telecommunication companies who pass it on to ordinary internet users.}
    \end{minipage}
\end{figure}

The \ac{fcc} appointed the \ac{usac}, a private non-profit organization, to administer the \ac{usf}. 
Instead of drawing on the federal budget, \ac{usac} imposes a levy on telecommunication companies, known as the ``contribution factor.'' 
Acting as the de facto setter and enforcer of this levy, \ac{usac} updates the contribution factor each quarter. 
Firms are allowed to itemize and shift the contribution factor to end users, effectively making it a tax on households and businesses. 
The contribution factor has risen persistently from 3.14\% in 1998 to 37.6\% in 2026 (Figure \ref{fig:be_contribution_factor}). 
Independent audits and policymakers have expressed concern that program costs are inflated by waste, fraud, and abuse, a concern corroborated by our analysis.
The four \ac{usf} programs together cost \$8.6 billion in 2024 \citep{usf_annual_report_2024}, which is borne by internet users. 
This financing structure may generate a positive feedback loop: rising program costs increase the contribution factor, raising internet prices, exacerbating affordability, expanding reliance on subsidies and further inflating costs \citep{hazlett2019educational}. 
A back-of-the-envelope calculation suggests that the contribution factor has raised the average internet prices in the \ac{us} by 6.7\%, forcing 1.2–3.1\% of users to depend on subsidies. 
This dynamic may undermine the program’s original objective.

There are several studies of the Lifeline program \citep{mendez2021impacts, ackerberg2014estimating, conkling2018crowd, wallsten2016learning, lyons2023assessing, ward2010effect}, E-Rate program \citep{hazlett2019educational, goolsbee2006impact, greenstein2020basic}, and High Cost program \citep{wallsten2011universal, savage2025impacts, berg2011incentives, berg2011universal, boik2017economics}. However, rigorous economic analyses of the \ac{rhc} program are missing. To our knowledge, there are only two correlation studies of the \ac{rhc} \citep{rabbani_bb1, rabbani_bb2} showing that internet plans in \p1\ cost approximately 155\% more than those in \p2/\p2c.

Our analysis focuses on the \ac{rhc} program, whose institutional timeline is shown in Figure \ref{fig:timeline}. 
Established in 1997, the program initially used a single price-cap subsidy mechanism known as the Telecommunications Program (\p1). 
If an \ac{hcp} paid price $p$ for internet service and could demonstrate that the same plan cost \( p_u \) in a nearby urban area, the \ac{hcp} would pay $p_u$ while the program covered $p - p_u$. 
This design is flawed economically because, at a fixed $p_u$, \p1\ insulates \acp{hcp} from marginal changes in $p$.
Demand becomes perfectly inelastic, eliminating the incentives for cost minimization \citep{rabbani_bb2}.

\begin{figure}[htbp]
    \centering
    \begin{minipage}{\linewidth}
    \centering
    \includegraphics[width=\linewidth]{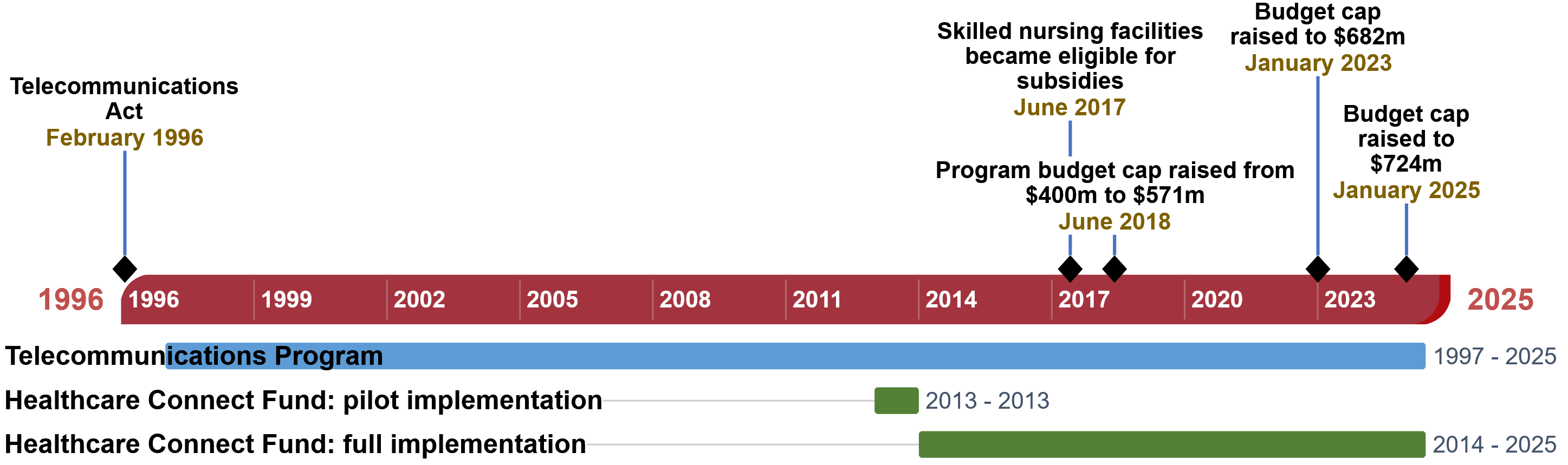}
	\caption{A timeline of program implementation.}
	\label{fig:timeline}
    \vspace{2pt}
    \parbox{\linewidth}{\scriptsize\justifying\textit{Notes:} The figure illustrates major events during the lifetime of the Rural Healthcare Program. It also shows the timeline of implementation of its two mechanisms, namely, the Telecommunications Program and the Healthcare Connect Fund.}
    \end{minipage}
\end{figure}

Piloted in 2013 and fully implemented in 2014, the \acf{hcf} began to operate alongside \p1. 
Rather than benchmarking against urban prices, \ac{hcf} uses an ad valorem subsidy that covers $0.65p$, requiring the \ac{hcp} to pay $0.35p$ as copayment.
This proportional cost-sharing introduces efficiency incentives: 
\acp{hcp} gain \$0.35 for every \$1.00 saved and lose \$0.35 for every \$1.00 wasted. 
Under the \ac{hcf}, participants may seek subsidies individually (\p2) or form a consortium of \acp{hcp} that requests subsidies on behalf of its members (\p2c).
The consortium option extends benefits to urban \acp{hcp} if the consortium is majority-rural, i.e., if more than half of its members are located in rural areas.%
\footnote{For example, if a large urban general hospital and two small rural clinics form a consortium, it is considered majority rural because two out of three members are rural, and all three members would be subsidized.}

An \ac{hcp} is deemed ``eligible'' for subsidies if 
(1) it is a non-profit or public entity, 
(2) it is rural or part of a majority-rural consortium, and 
(3) it belongs to a qualifying entity type (listed in Appendix \ref{sec:rhc_program}). 
Once eligibility is established, the \ac{hcp} may choose \p1, \p2, or \p2c. 
The only exception is that eligible urban \acp{hcp} can only seek subsidies via \p2c. 
\acp{hcp} can switch programs upon subsidy renewal. 
\acp{hcp} who do not meet all the three criteria are deemed ineligible.
Ineligible \acp{hcp} do not receive subsidies and do not count in the determination of the majority-rural status.
But they may join a consortium to attempt to secure lower internet prices as part of a collective bargaining.

When a consortium submits a subsidy request, it enters a competitive bidding process, in which \acp{isp} place bids for providing internet services to all eligible and ineligible members of a consortium.
The auction winner is not the lowest bid. 
Instead, the \ac{hcp} effectively chooses the bid that it deems most desirable based on subjective and arbitrary criteria (see Table \ref{tab:rfp_weights}).
Therefore, \acp{isp} compete, not to offer the most competitive bid, but to make the bid most appealing to the \ac{hcp}.
We hypothesize that this enables a cross-subsidization scheme, where the winning bid offers inflated prices to eligible members and discounted prices to ineligible members, resulting in elevated subsidy outlays.
For every \$1 that is cross-subsidized, the consortium would collect an additional \$0.65 in subsidies.
This could be done in a revenue-neutral manner to the \ac{isp}, or the \ac{isp} may collect some of the proceeds.
Appendix \ref{sec:cross_sub} documents the institutional conditions enabling this behavior and presents supportive anecdotal evidence.

\begin{figure}[htbp]
    \centering
    \begin{minipage}{\linewidth}
    \centering
    \includegraphics[width=\linewidth]{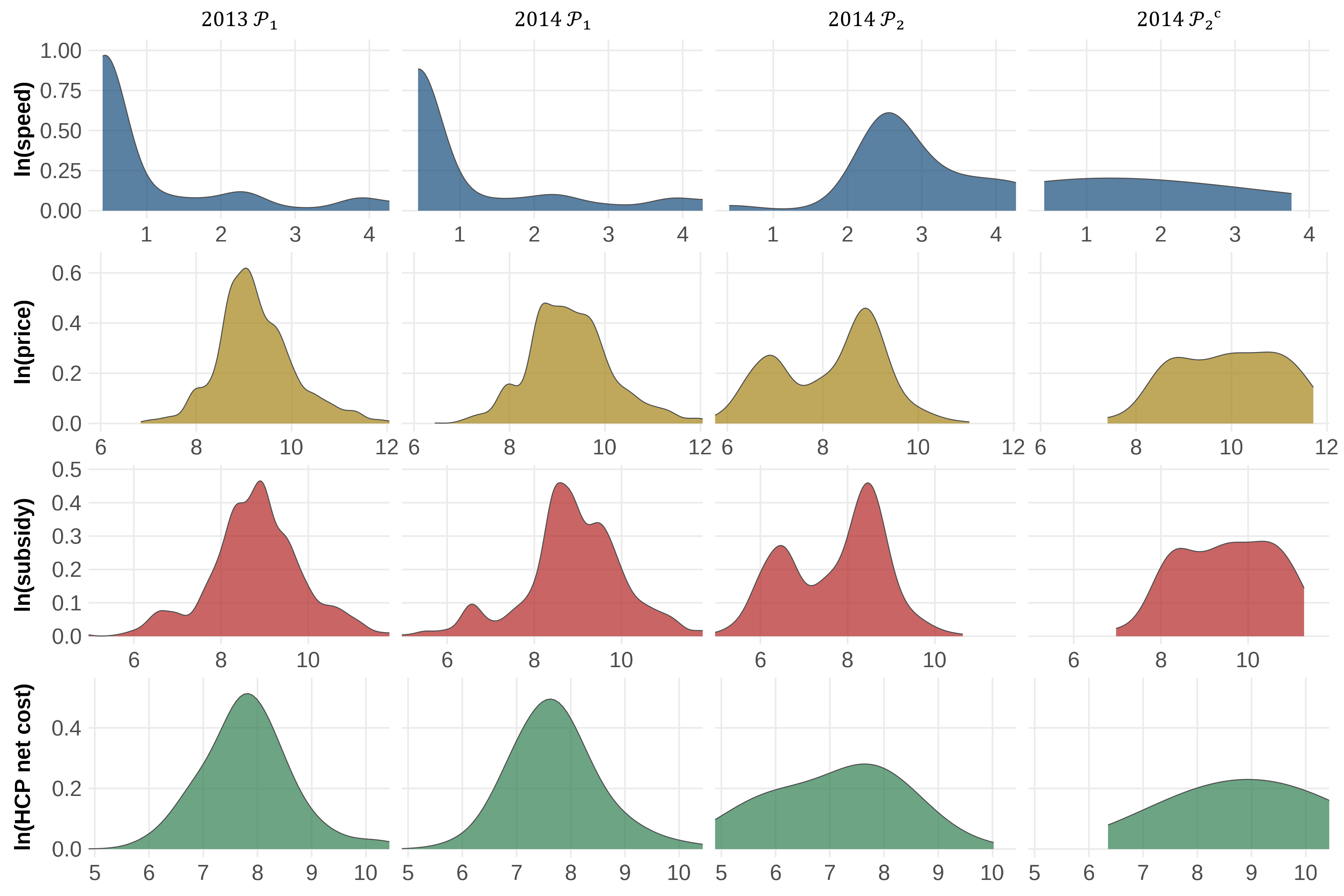}
	\caption{Kernel density of request-level variables by program, 2013--2014.}
	\label{fig:bd_density_4x4}
    \vspace{2pt}
    \parbox{\linewidth}{\scriptsize\justifying\textit{Notes:} Each panel shows a kernel density estimate. Rows correspond to the natural logarithm of speed, price, subsidy, and \ac{hcp} net cost. Columns correspond to programs: 2013 \p1, 2014 \p1, 2014 \p2, and 2014 \p2c. Within each row, all panels share a common horizontal and vertical axis to facilitate visual comparison across programs.}
    \end{minipage}
\end{figure}

Figure \ref{fig:bd_density_4x4} uses publicly available program data (described in Section \ref{sec:data}) to document initial patterns.
When \p2\ and \p2c\ became available in 2014, more than half of the \acp{hcp} selected \p2, while the adoption of \p2c\ was minimal.
The density plots reveal that \p1\ plans span a wide range of speeds and prices, whereas \p2\ plans concentrate at higher speeds.
\p2c\ plans, though few in 2014, already exhibit higher prices and subsidies relative to \p2.
Widespread participation in \p2c\ did not occur until 2017, when regulatory changes expanded eligibility and attracted new \acp{hcp} who predominantly chose \p2c.%
\footnote{The first event was in June 2017, when skilled nursing facilities became eligible for subsidies. This opened the program to many new \acp{hcp}. The second event was a budget adjustment. When the program was implemented in 1997, \ac{usf}'s annual budget was capped at \$400 million. Two decades had passed and this cap was never reached. The budget cap was exceeded in 2016. This created a financial strain on the program. In June 2018, for the first time in the program's lifetime, the \ac{fcc} began indexing the \ac{rhc}'s annual cap to inflation. The budget kept rising steadily, reaching \$724 million in 2025. \href{https://www.fcc.gov/general/rural-health-care-program}{FCC: Rural Health Care Program}, \href{https://docs.fcc.gov/public/attachments/DOC-358434A1.pdf}{FCC: Report and Order – WC Docket No. 17-310}}

The first switching opportunity between programs was in 2014. 
\acp{hcp} selecting \p2\ that year opted for higher speeds relative to those who stayed in \p1. 
Nonetheless, the minority that stayed in \p1\ and accounted for a small share of total purchased speed generated more than half of the \ac{rhc}'s total cost burden in 2014. 
This pattern could be economically explained. 
In 2013, all \acp{hcp} were under \p1. 
Given that the \ac{hcp} net cost ratio for \p2/\p2c\ was 35\%, \acp{hcp} facing a ratio above 35\% had an incentive to switch to \p2/\p2c\ in 2014, while those with a ratio below 35\% would remain under \p1\ to preserve their lower \ac{hcp} net cost.
A low ratio in \p1\ implies that the rural rate is substantially higher than the urban benchmark. 
For instance, a ratio of 5\% implies that the rural rate is 20 times the urban rate. 
Thus, a low \ac{hcp} net cost ratio may signify an overpriced or cost-inefficient plan \citep{rabbani_bb2}. 
Therefore, the introduction of \p2\ may have caused a selective migration in which cost-conscious \acp{hcp} switched to \p2/\p2c\ while the rest remained in \p1.
This dynamic explains why the small residual group in \p1\ generated a disproportionately large share of the total expenditure.

\begin{figure}[htbp]
    \centering
    \begin{minipage}{\linewidth}
    \centering
    \includegraphics[width=\linewidth]{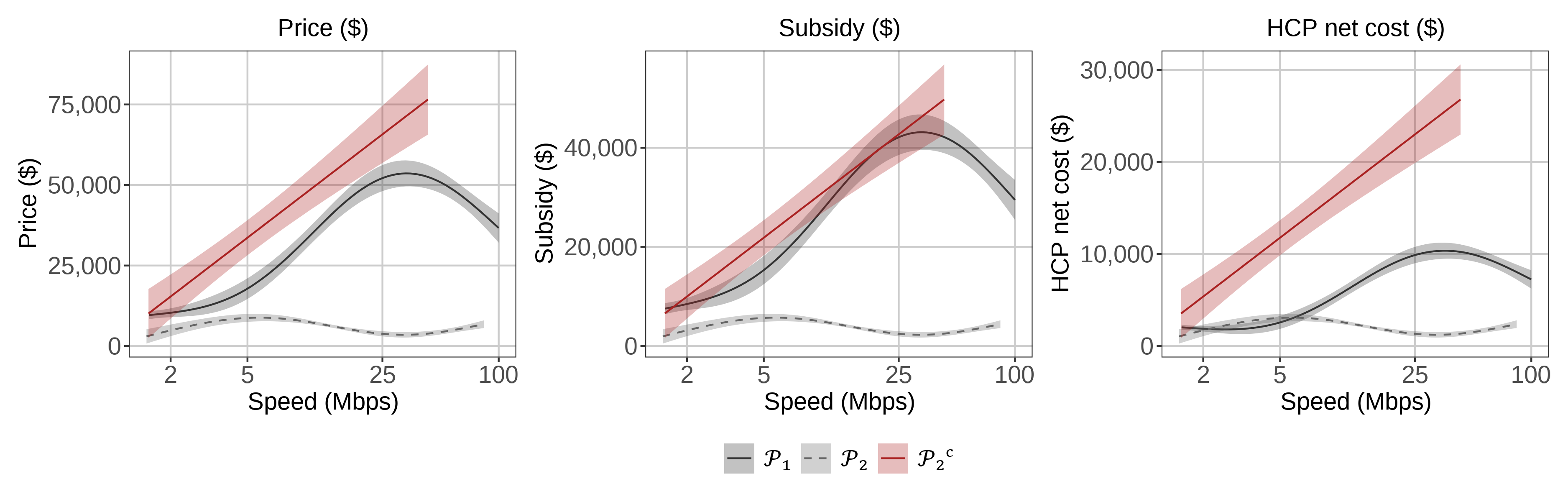}
	\caption{A program comparison of price measures in 2014.}
	\label{fig:bh_GAM_combined2014}
    \vspace{2pt}
    \parbox{\linewidth}{\scriptsize\justifying\textit{Notes:} The figure illustrates Generalized Additive Models fitted at the request level using the 2014 data. Programs \protect\p1, \protect\p2, and \protect\p2c\ are respectively shown in dark gray, light gray, and red. Each dependent variable (price, subsidy, or HCP net cost) is fitted as a function of ln(speed). The y-axis shows dollar values on a linear scale.}
    \end{minipage}
\end{figure}

Figure \ref{fig:bh_GAM_combined2014} compares internet price, subsidy, and \ac{hcp} net cost across different speed levels using a \ac{gam}. 
\ac{gam} is a semi-parametric curve-fitting method that accommodates nonlinear relationships between variables at distinct speed levels.
Each fitted curve is shown with a 95\% confidence interval. The left, middle, and right panels correspond to price, subsidy, and \ac{hcp} net cost, respectively.
Consistent with the cost-effectiveness hypothesis, internet plans that are subsidized via \p2\ are cheaper than those in \p1. Meanwhile, plans in \p2c\ are generally more expensive than those in \p2, which is consistent with the cross-subsidization hypothesis. 
We formalize both hypotheses in Section \ref{sec:model}.

\begin{figure}[htbp]
    \centering
    \begin{minipage}{\linewidth}
    \centering
    \includegraphics[width=.5\linewidth]{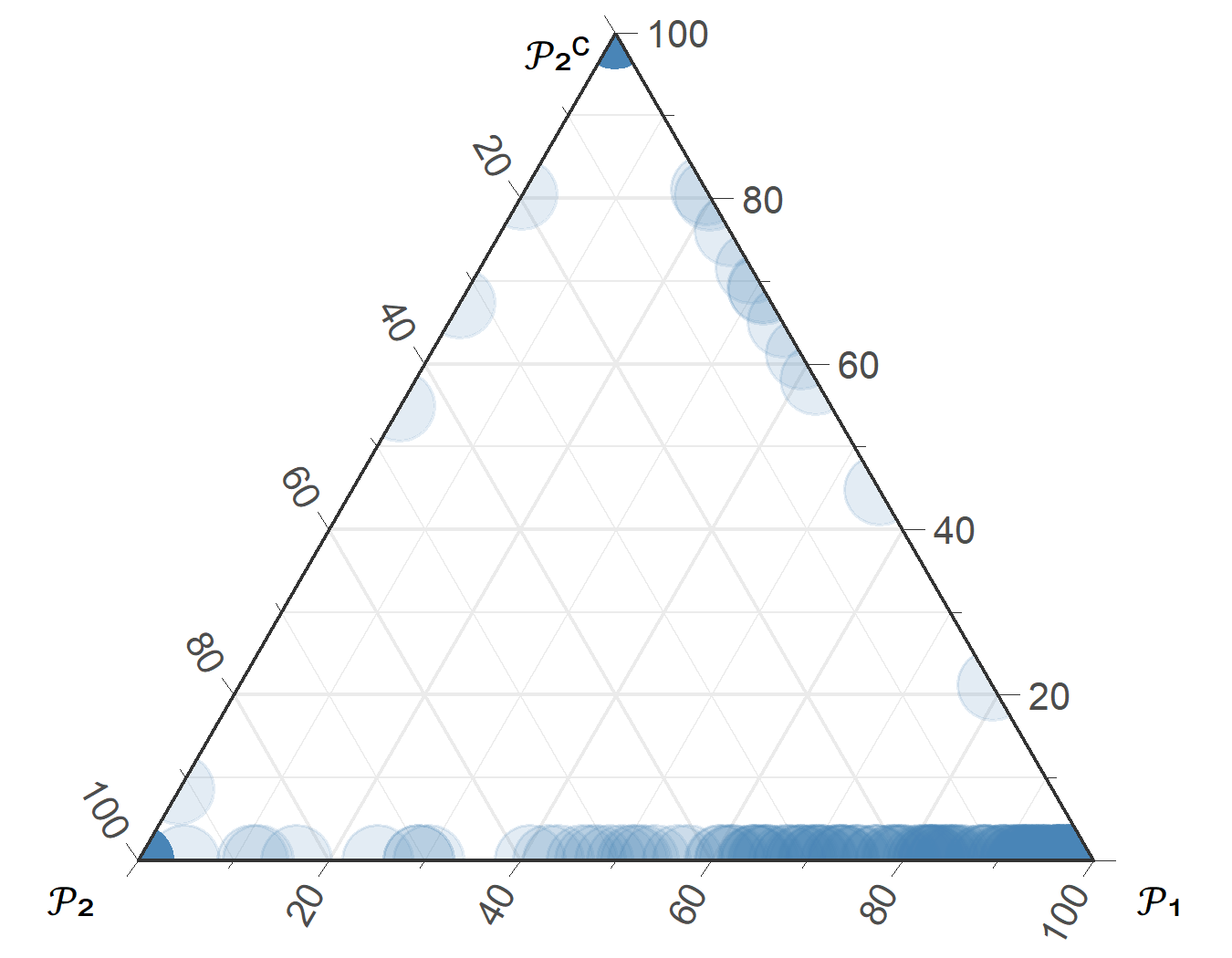}
	\caption{Kernel density distribution of the programs.}
	\label{fig:bf_density2014}
    \vspace{2pt}
    \parbox{\linewidth}{\scriptsize\justifying\textit{Notes:} Each point represents one \ac{hcp}. The axes show the fraction of the \ac{hcp}'s subsidized requests reimbursed under each program: \p1\ (price-cap), \p2\ (individual ad valorem), and \p2c\ (consortium ad valorem). Coordinates sum to one. Proximity to a vertex indicates a higher concentration under that program. Data are from the \ac{fcc}'s \ac{rhc} program for funding year 2014, the first year in which \p2\ and \p2c\ were available. The concentration along the \p1-\p2\ edge reflects the initial transition from the price-cap to the ad valorem mechanism, with few providers yet participating in consortia.}
    \end{minipage}
\end{figure}

An \ac{hcp} may have multiple active lines of subsidy at the same time, especially if it has multiple buildings or locations.
Many \acp{hcp} receive support across tens or hundreds of lines simultaneously, often using a combination of \p1, \p2, and \p2c.
Figure \ref{fig:bf_density2014} illustrates this composition in 2014, the first year in which \p2\ and \p2c\ were available.
The dense band along the \p1--\p2\ edge confirms that most \acp{hcp} divided their requests between \p1\ and \p2, while very few participated in \p2c. 

\section{A monopoly model of subsidy design}
\label{sec:model}

This section develops a monopoly model that isolates three mechanisms linking subsidy design to equilibrium prices, quantities, and government outlays. The \ac{hcp} acts as
the consumer of a monopolistic \ac{isp}. Under a price-cap regime (\p1), consumer demand becomes locally inelastic once the cap binds, severing the link between billed
prices and quantities and permitting the monopolist to inflate bills against the government budget. Under an ad valorem subsidy (\p2), the consumer retains proportional
exposure to the posted price and therefore remains sensitive to price changes, constraining the monopolist's markup. When a consortium of eligible and ineligible members operates under the ad valorem scheme (\p2c), it can inflate prices for eligible members and use the additional subsidies to reduce costs for ineligible ones. This unintended consequence erodes the very price discipline that the ad valorem mechanism was designed to restore.

\subsection{Environment}
\label{sec:model_env}

A monopolist \ac{isp} supplies a homogeneous good at constant marginal cost $c>0$.
Demand $D:\mathbb{R}_+\to\mathbb{R}_+$ satisfies $D^{\prime}(p)<0$ wherever $D(p)>0$ and determines the quantity $Q = D(p_c)$ for the \ac{hcp} net cost~$p_c$. 
Revenue $p\ D(p)$ is strictly concave.
Without subsidies, the monopoly price~$p^{no}$ is the unique solution to the first-order condition:
\begin{equation}
D(p^{no})+(p^{no}-c)\,D^{\prime}(p^{no})=0.
\label{eq:foc_no}
\end{equation}
The monopoly price can be equivalently expressed via the Lerner index:
\[
\frac{p^{no}-c}{p^{no}}=\frac{1}{\varepsilon_D(p^{no})},
\]
where $\varepsilon_D(p)\equiv -\frac{pD^{\prime}(p)}{D(p)}$ is the price elasticity of demand.
Under perfect competition, $p_c=c$ and $Q=D(c)$.
\footnote{Although we present results for a general demand function, Appendix~\ref{sec:appx_linear} provides closed-form expressions under the linear specification $D(p_c)=a-bp_c$ ($a, b>0$, $a>bc$).}

\subsection{Price-cap subsidy (\texorpdfstring{$\mathcal{P}_1$}{p1})}
\label{sec:model_cap}

Under \p1, the \ac{isp} posts a billed price $p\ge 0$.
The consumer pays at most a regulated cap $\bar{p}>0$, so the \ac{hcp} net cost is $p_c^{cap}(p)=\min\{p,\bar{p}\}$.
When $p>\bar{p}$, the government reimburses the gap per unit:
\begin{equation}
G^{cap}(p)=(p-\bar{p})_+\,D\!\bigl(p_c^{cap}(p)\bigr),
\label{eq:gov_cap}
\end{equation}
where $(x)_+\equiv\max\{x,0\}.$
Excess billing is subject to an expected penalty $\alpha\,\Phi(\delta)$ on the price deviation $\delta\equiv(p-\bar{p})_+$, where $\alpha\in(0,1]$ is the probability of being
audited and $\Phi$ is a penalty function satisfying $\Phi(0)=0$, $\Phi'(0)=0$, $\Phi'(\delta)>0$ for $\delta>0$, and $\Phi''>0$. The condition $\Phi'(0)=0$ ensures that small deviations are not deterred, while the strict convexity of $\Phi$ ensures that marginal penalties escalate with the degree of overcharging. For tractability, we adopt the quadratic specification $\Phi(\delta)=\tfrac{\gamma}{2}\delta^2$
$(\gamma>0).$

The economic mechanism is a \emph{discontinuity in the residual demand elasticity}.
When the cap does not bind ($p\le\bar{p}$), the monopolist faces the standard downward-sloping demand curve and
markups are disciplined by the demand response.
Once $p$ crosses $\bar{p}$, the consumer price is pinned at $\bar{p}$, demand freezes at $D(\bar{p})$, and the firm's residual demand becomes perfectly inelastic in the billed price.
The only force restraining further price inflation is the enforcement technology $(\alpha,\Phi)$.

The cap binds whenever $\bar{p}<p^{no}$.\footnote{More generally,
because $\Phi'(0)=0$, exceeding the cap is always locally profitable.
Even when $\bar{p}\ge p^{no}$, the cap binds globally if enforcement
is sufficiently weak: the insulation rent $D(\bar{p})^2/(2\alpha\gamma)$
can exceed the standard monopoly profit loss from operating at
$\bar{p}$ instead of $p^{no}$.
In the \ac{rhc} context, the empirically relevant case is
$\bar{p}<p^{no}$, where the cap binds for all enforcement levels.}
When it does, the \ac{hcp} pays $p_c^{cap}=\bar{p}$ and demand is pinned at $Q^{cap}=D(\bar{p})$.
The \ac{hcp} expenditure is $E^{cap}=\bar{p}\,D(\bar{p})$, and the firm's problem reduces to:

\begin{equation}
\max_{p\ge\bar{p}}\;(p-c)\,D(\bar{p})-\alpha \Phi(p-\bar{p}).
\label{eq:cap_problem}
\end{equation}

\begin{proposition}[Billed-price under \p1]
\label{prop:cap}
Suppose the cap binds. Then the equilibrium billed price is
\begin{equation}
p^{cap}=\bar{p}+(\Phi')^{-1}\!\!\left(\frac{D(\bar{p})}{\alpha}\right),
\label{eq:pcap}
\end{equation}
which is increasing in $D(\bar{p})$ and decreasing in $\alpha$.
Under the quadratic penalty,
\begin{equation}
p^{cap}=\bar{p}+\frac{D(\bar{p})}{\alpha\gamma},
\qquad
G^{cap}=\frac{D(\bar{p})^2}{\alpha\gamma}.
\label{eq:pcap_quad}
\end{equation}
\end{proposition}

\noindent
\textit{Proof.}
See Appendix~\ref{sec:appx_cap}.
\hfill$\square$

\medskip\noindent
The result exposes the fundamental weakness of price-cap subsidies: the marginal government payment is one-for-one in the billed price above $\bar{p}$, unconstrained by any demand response.
Government outlays $G^{cap}$ grow as enforcement weakens ($\alpha\gamma\to 0$), because the monopolist's inflation incentive is checked only by penalties instead of lost sales.

  \subsection{Ad valorem subsidy (\texorpdfstring{$\mathcal{P}_2$}{p2})}
  \label{sec:model_av}

  Under \p2, the government pays a constant share $\tau\in(0,1)$ of the
  transaction price.
  The consumer pays $(1-\tau)p$, so quantity is
  $D\!\bigl((1-\tau)p\bigr)$ and government outlays are
  $G^{adv}(p)=\tau\,p\,D\!\bigl((1-\tau)p\bigr)$.
  The firm's profit is
  $\pi^{adv}(p;\tau)=(p-c)\,D\!\bigl((1-\tau)p\bigr),$
  which yields the first-order condition
  \begin{equation}
  D\!\bigl((1-\tau)p\bigr)+(p-c)(1-\tau)\,D'\!\bigl((1-\tau)p\bigr)=0.
  \label{eq:foc_av}
  \end{equation}
  Let $p^{adv}(\tau)$ denote the solution to~\eqref{eq:foc_av}.
  The \ac{hcp} net cost, quantity, expenditure, and
  government outlay are
  $p_c^{adv}(\tau)\equiv(1-\tau)\,p^{adv}(\tau)$,
  $Q^{adv}(\tau)\equiv D\!\bigl(p_c^{adv}(\tau)\bigr)$,
  $E^{adv}(\tau)\equiv p_c^{adv}(\tau)\,Q^{adv}(\tau)$, and
  $G^{adv}(\tau)\equiv\tau\,p^{adv}(\tau)\,Q^{adv}(\tau)$,
  respectively.
  Unlike the price-cap regime, the firm faces a downward-sloping
  residual demand in its posted price because consumers bear a share
  $(1-\tau)$ of each price increase.
  \p2\ restores the elasticity of demand that was suppressed under \p1.

Substituting $p_c\equiv(1-\tau)p$ into~\eqref{eq:foc_av} yields:

\begin{equation}
  D(p_c)+\bigl(p_c-c(1-\tau)\bigr)\,D'(p_c)=0.
  \label{eq:foc_av_q}
  \end{equation}

This is the standard monopoly first-order condition~\eqref{eq:foc_no}
with \emph{effective marginal cost} $c(1-\tau)$ replacing~$c$.
The ad valorem subsidy acts \emph{as if} it reduces the monopolist's
cost by the factor $(1-\tau)$ and lets the consumer face the resulting
monopoly price directly:

\begin{equation}
p_c^{adv}(\tau)=p^{no}\!\bigl(c(1-\tau)\bigr),
\label{eq:pcav_emc}
\end{equation}

where $p^{no}(\tilde{c})$ denotes the monopoly price as a function of marginal cost $\tilde{c}$.

The effective-marginal-cost representation~\eqref{eq:pcav_emc}
characterizes the condition that a \ac{hcp} would voluntarily switch
from \p1\ to~\p2.

Since $dp^{no}(\tilde{c})/d\tilde{c}>0$, the consumer price
  $p_c^{adv}(\tau)=p^{no}\!\bigl(c(1-\tau)\bigr)$ is strictly
  decreasing in~$\tau$, with $p_c^{adv}(0)=p^{no}(c)$ and
  $p_c^{adv}(\tau)<p^{no}(c)$ for all $\tau>0$.
  When the price cap binds, $\bar{p}<p^{no}(c)$, the
  continuity of $p_c^{adv}$ guarantees the existence of a unique
  \emph{critical subsidy rate} $\tau^*\in(0,1)$ satisfying:\footnote{Provided that the market power is not so extreme that $p^{no}(0)\ge\bar{p}$,
  in which case no $\tau\in(0,1)$ induces switching and the \ac{hcp}
  rationally remains on~\p1. Under linear demand, $\tau^*=2(p^{no}-\bar{p})/c$.
  The condition $\tau^*<1$ requires $\bar{p}>p^{no}-c/2$; when this
  fails, even full subsidy coverage cannot offset the monopolist's
  markup. See Appendix~\ref{sec:appx_linear}.}

  \begin{equation}
  p_c^{adv}(\tau^*)=\bar{p},
  \label{eq:cond_barp}
  \end{equation}

A cost-minimizing \ac{hcp} switches from \p1\ to \p2\ whenever
$\tau\ge\tau^*$, i.e., whenever the ad valorem consumer price falls
below the cap.

  \begin{proposition}[Ad valorem dominance]
  \label{prop:cap_to_av}
  Suppose the price cap binds and $\tau\ge\tau^*$. Then:
  \begin{enumerate}
  \item[\emph{(i)}] $p_c^{adv}(\tau)\le\bar{p}=p_c^{cap}$, with strict
  inequality when $\tau>\tau^*$;
  \item[\emph{(ii)}] $Q^{adv}(\tau)\ge Q^{cap}$;
  \item[\emph{(iii)}] if $\varepsilon_D(p)\le 1$ on
  $[p_c^{adv}(\tau),\,\bar{p}]$, then $E^{adv}(\tau)\le E^{cap}$;
  \item[\emph{(iv)}] under the quadratic penalty,
  $G^{adv}(\tau)<G^{cap}$ for sufficiently small $\alpha\gamma$.
  \end{enumerate}
  \end{proposition}

  \noindent
  \textit{Proof.}
  See Appendix~\ref{sec:appx_av}.
  \hfill$\square$

  \medskip\noindent
  Proposition~\ref{prop:cap_to_av} establishes that the ad valorem
  mechanism dominates the price cap on every margin relevant to
  policymakers.
  The \ac{hcp} pays less per unit under \p2\ (Part~(i)); a lower
  consumer price raises quantity consumed (Part~(ii)); \ac{hcp}
  expenditure falls when demand is price-inelastic (Part~(iii)); and
  government outlays decline because the consumer absorbs a share
  $(1-\tau)$ of the price, anchoring the monopolist's markup to the
  demand curve rather than to the enforcement technology (Part~(iv)).

\subsection{Consortium under ad valorem (\texorpdfstring{$\mathcal{P}_2^c$}{p2c})}
\label{sec:model_consortium}

Mechanism \p2c\ allows \p2-eligible \acp{hcp} to form a consortium with ineligible \acp{hcp}.
Although subsidies apply only to eligible members, the consortium can
reallocate billed charges across members while holding fixed both
total \ac{isp} revenue and quantities consumed.
Shifting charges toward eligible members increases government
reimbursements, financing lower prices for ineligible members. In
effect, eligible \acp{hcp} become a conduit through which
subsidy benefits are extended to ineligible participants.
We develop a stylized model of this channel.

The monopolist \ac{isp} serves two \acp{hcp}:
an \emph{eligible} member $E$ receiving an ad valorem subsidy at rate $\tau \in(0,1)$, and an \emph{ineligible} member $I$ with no subsidy.
The \ac{isp} determines per-unit prices $(p_E,p_I)$, yielding consumer prices $p_{c,E}=(1-\tau)\,p_E$, $p_{c,I}=p_I$ and quantities
$Q_E=D_E(p_{c,E})$ and $Q_I=D_I(p_{c,I}).$  The firm chooses prices to maximize profit

\begin{equation*}
\pi(p_E,p_I) = (p_E-c)\,D_E\!\bigl( (1-\tau) p_E\bigr)+ (p_I-c)\,D_I(p_I).
\end{equation*}

The equilibrium prices $(p_E^{adv},p_I^{adv})$ satisfy the first-order conditions in each market separately:
\begin{align*}
&D_E\!\bigl((1-\tau)p_E^{adv}\bigr)+(p_{E}^{adv}-c)(1-\tau)\,D_E'\!\bigl((1-\tau)p_{E}^{adv}\bigr)\:=\:0,\\
&D_I(p_{I}^{adv})      +\bigl(p_{I}^{adv}-c        \bigr)\,D_I'(        p_{I}^{adv})\:=\:0.
\end{align*}

The consortium sets an internal price allocation $(\tilde{p}_E,\tilde{p}_I)$ that aims to reduce total \ac{hcp} net cost while holding quantity levels and total \acs{isp} revenue fixed:

\begin{equation*}
p_E^{adv}Q_E^{adv}+p_I^{adv}Q_I^{adv}
=\tilde{p}_E\,Q_E^{adv}+\tilde{p}_I\,Q_I^{adv},
\qquad \tilde{p}_E,\tilde{p}_I\ge 0.
\end{equation*}

The reduction in total \ac{hcp} net cost satisfies the accounting identity

\begin{equation}
\Delta C
=\tau\,(\tilde{p}_E-p_E^{adv})\,Q_E^{adv}
=\tau\,(p_I^{adv}-\tilde{p}_I)\,Q_I^{adv}.
\label{eq:DeltaC}
\end{equation}

Each dollar shifted to the eligible member's bill reduces the combined \ac{hcp} net cost by $\tau$ dollars,
financed by additional government reimbursements.
This is a zero-sum transfer between the consortium and the government budget.
The government pays $\tilde{G} = \tau \tilde{p}_E\,Q_E^{adv}$ and the increase in the government outlay is
$\Delta G = \tau (\tilde{p}_E -p_E^{adv})\,Q_E^{adv} = \Delta C$.

Three reduced-form parameters govern the consortium's problem:
\begin{equation}
  \kappa\equiv\frac{\tilde{p}_E}{p_E^{adv}}\ge 1,
  \qquad
  B\equiv\tau\,p_E^{adv}Q_E^{adv},
  \qquad
  R\equiv\frac{p_I^{adv}Q_I^{adv}}{p_E^{adv}Q_E^{adv}}\ge 0.
  \label{eq:consortium_notation}
  \end{equation}
  The \emph{eligible price-distortion ratio}~$\kappa$ measures how much
  the consortium inflates the eligible member's internal price above the
  equilibrium level: $\kappa=1$ is no distortion; $\kappa>1$ is
  cross-subsidization.
  The \emph{eligible subsidy base}~$B$ is the government reimbursement
  at the undistorted equilibrium, i.e., the baseline subsidy before any
  reallocation.
  The \emph{ineligible-to-eligible revenue ratio}~$R$ measures
  consortium composition: $R=0$ means no ineligible members; large~$R$
  means ineligible revenue dominates.
  Nonnegativity of $\tilde{p}_I$ requires $\kappa\in[1,\,1+R]$: the
  consortium cannot shift more revenue to the eligible bill than the
  ineligible member generates.
  Under this notation, the cost reduction becomes
  $\Delta C=B(\kappa-1)$.

Cross-subsidization is constrained by an expected penalty
  $\alpha\,\Phi(\Delta C;\,R)$, where $\alpha\in(0,1]$ is the
  probability of being audited and
  $\Phi:\mathbb{R}_+\times\mathbb{R}_+\to\mathbb{R}_+$ is a penalty
  function satisfying five conditions:
  (i)~$\Phi(0;\,R)=0$: no distortion incurs no penalty;
  (ii)~$\Phi'(0;\,R)=0$: the marginal penalty at zero distortion is
  zero, ensuring that small deviations are not deterred;
  (iii)~$\Phi'(\delta;\,R)>0$ for $\delta>0$: the penalty is strictly
  increasing in the distortion;
  (iv)~$\Phi''(\delta;\,R)>0$: the penalty is strictly convex, so that
  marginal penalties escalate with the degree of cross-subsidization;
  and
  (v)~$\partial^2\Phi/\partial\delta\,\partial R>0$: the penalty is
  supermodular in $(\delta,R)$, so that consortia with larger
  ineligible-to-eligible revenue ratios face a steeper marginal
  deterrent.
  For tractability, we adopt the quadratic specification
  $\Phi(\delta;\,R)=\tfrac{\gamma R}{2}\,\delta^2$ ($\gamma>0$), which
  satisfies all five conditions and yields closed-form solutions.

  The consortium's price allocation solves
  \begin{align}
  \label{eq:CostMini}
  \min_{\tilde{p}_E,\;\tilde{p}_I\ge 0}\quad
  &\underbrace{(1-\tau)\,\tilde{p}_E\,Q_E^{adv}
    \;+\;\tilde{p}_I\,Q_I^{adv}}_{\text{HCP net cost}}
  \;+\;
  \underbrace{\alpha\,\Phi(\Delta C;\,R)}_{\text{expected penalty}}\\
  \text{s.t.}\quad
  &\tilde{p}_E\,Q_E^{adv}+\tilde{p}_I\,Q_I^{adv}
    \;=\;p_E^{adv}\,Q_E^{adv}+p_I^{adv}\,Q_I^{adv},\nonumber
  \end{align}
  which can be rewritten as a maximization of cost reduction minus the
  expected penalty over the price-distortion ratio~$\kappa$:
  \begin{equation}
  \max_{\kappa\in[1,\,1+R]}\;\Psi(\kappa),
  \qquad
  \Psi(\kappa)\equiv
  B\,\kappa-\alpha\,\Phi\!\bigl(B(\kappa-1);\,R\bigr).
  \label{eq:CostMini2}
  \end{equation}

  \begin{proposition}[Optimal cross-subsidization]
  \label{prop:hump}
  Suppose $\Phi$ satisfies conditions~\emph{(i)--(v)}. Then:
  \begin{enumerate}
  \item[\emph{(i)}]
  Cross-subsidization is always locally profitable:
  $\Psi'(1)=B>0$, which implies $\tilde{p}_E>p_E^{adv}$.

  \item[\emph{(ii)}]
  $\Psi$ is strictly concave on $[1,\,1+R]$.
  Hence the optimal distortion $\kappa^*$ is unique.

  \item[\emph{(iii)}]
  If $\alpha\,\Phi'\!(BR;\,R)>1$, the optimum is interior
  ($1<\kappa^*<1+R$) and satisfies
  \begin{equation}
  \alpha\,\Phi'\!\bigl(B(\kappa^*-1);\,R\bigr)=1.
  \label{eq:foc_kappa}
  \end{equation}
  If $\alpha\,\Phi'\!(BR;\,R)\le 1$, the constraint binds and
  $\kappa^*=1+R$.

  \item[\emph{(iv)}]
  Under the quadratic penalty
  $\Phi(\delta;\,R)=\tfrac{\gamma R}{2}\,\delta^2$,
  \begin{equation}
  \kappa^*(R)
  =\min\!\left\{1+R,\;1+\frac{1}{\alpha\gamma B\,R}\right\}.
  \label{eq:xstar}
  \end{equation}
  The optimal distortion $\kappa^*$ is hump-shaped in~$R$,
  attaining its maximum $\kappa^*=1+R^*$ at
  $R^*=1\big/\!\sqrt{\alpha\gamma B}$,
  and converging to $1$ (no cross-subsidization) as
  $R\to 0$ or $R\to\infty$.
  \end{enumerate}
  \end{proposition}

  \noindent
  \textit{Proof.}
  See Appendix~\ref{sec:appx_hump}.
  \hfill$\square$

  \medskip\noindent
Part~(i) establishes that some cross-subsidization always occurs:
  because $\Phi'(0;\,R)=0$, the marginal penalty at zero distortion is
  zero while the marginal saving is $B>0$.
  Part~(ii) guarantees uniqueness via the convexity of the penalty.
  Part~(iii) characterizes the two regimes.
  When $R$ is small, the feasibility constraint $\kappa\le 1+R$ binds
  before enforcement does: there is insufficient ineligible revenue to
  shift.
  When $R$ is large, enforcement binds first: the marginal penalty at
  full reallocation exceeds the marginal saving, and the consortium
  self-restrains.
  Part~(iv) shows that under the quadratic penalty the transition
  between regimes occurs at $R^*=1/\sqrt{\alpha\gamma B}$, generating
  a hump shape that reflects a Laffer-curve logic applied to subsidy
  extraction: just as the Laffer curve traces the tension between the tax base and the behavioral response, $\kappa^*(R)$ traces the tension between cross-subsidization opportunity and enforcement exposure (see Figure~\ref{fig:hump}).

\begin{figure}[htbp]
\caption{Optimal cross-subsidization distortion as a function of consortium composition.}
\label{fig:hump}
\centering
\begin{tikzpicture}[font=\footnotesize]
\begin{axis}[
    width=0.78\textwidth,
    height=0.45\textwidth,
    xlabel={Revenue ratio, $R$},
    ylabel={Optimal distortion, $\kappa^*(R)$},
    xmin=0, xmax=3,
    ymin=0.5, ymax=3.5,
    xtick={1},
    xticklabels={$R^*$},
    ytick={1, 2},
    yticklabels={$1$, $1{+}R^*$},
    axis lines=left,
    axis line style={->},
    clip mode=individual,
    every axis x label/.style={at={(ticklabel* cs:1)}, anchor=west},
    every axis y label/.style={at={(ticklabel* cs:1)}, anchor=south},
]

\addplot[dashed, black!50, thin, domain=1:2.5, samples=2] {1+x};

\addplot[dashed, black!50, thin, domain=0.4:1, samples=80, smooth] {1+1/x};

\addplot[very thick, black, domain=0.01:1, samples=100] {1+x};
\addplot[very thick, black, domain=1:3, samples=200, smooth] {1+1/x};

\draw[dotted, thin, black!40] (axis cs:0,1) -- (axis cs:3,1);
\node[font=\scriptsize, black!50, anchor=east] at (axis cs:0,1) {$\kappa{=}1$};

\draw[dotted, thin, black!50] (axis cs:1,0.5) -- (axis cs:1,2);
\draw[dotted, thin, black!50] (axis cs:0,2) -- (axis cs:1,2);
\addplot[only marks, mark=*, mark size=2.5pt, black] coordinates {(1,2)};

\node[font=\scriptsize, black!60, anchor=south west] at (axis cs:1.6,2.65)
  {feasibility: $1{+}R$};
\node[font=\scriptsize, black!60, anchor=west] at (axis cs:1.8,1.6)
  {enforcement: $1{+}\frac{1}{\alpha\gamma BR}$};

\end{axis}
\end{tikzpicture}

\vspace{2pt}
\parbox{\linewidth}{\scriptsize\justifying\textit{Notes:} The solid curve traces $\kappa^*(R)=\min\bigl\{1+R,\;1+1/(\alpha\gamma BR)\bigr\}$ (Proposition~\ref{prop:hump}, Part~(iv)). For small~$R$, the feasibility constraint $\kappa\le 1+R$ binds: there is insufficient ineligible revenue to shift. For large~$R$, enforcement binds: the marginal penalty exceeds the marginal saving. The peak at $R^*=1/\sqrt{\alpha\gamma B}$ maximizes the distortion---a Laffer-curve logic applied to subsidy extraction. Dashed lines show each branch beyond its binding region. Illustration uses $\alpha\gamma B=1$.}
\end{figure}
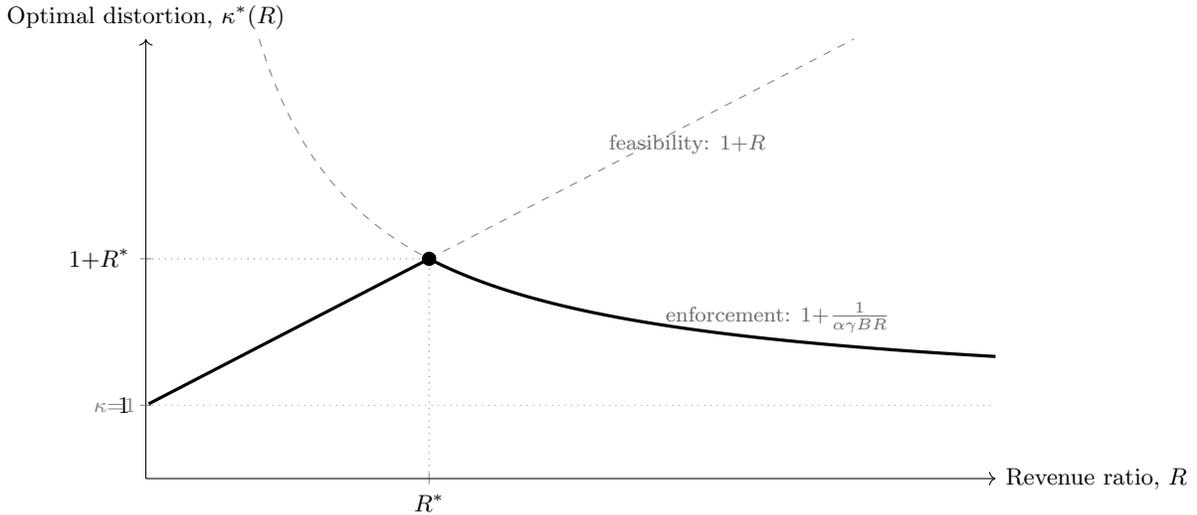

\subsection{Empirical implications}
\label{sec:model_implications}

The three mechanisms analyzed above generate testable predictions.
Our identification exploits the 2014 introduction of the \ac{hcf}, which for the first time offered \p2\ and \p2c\ as alternatives to the incumbent price cap~\p1. Because all \acp{hcp} were on \p1\ prior to 2014, the two observable transitions are \p1\ to \p2\ and \p1\ to~\p2c.

Proposition~\ref{prop:cap_to_av} predicts that switching from the price cap \p1\ to the ad valorem \p2\ lowers the consumer price (Part~(i)), raises quantity (Part~(ii)), reduces \ac{hcp} expenditure when demand is inelastic (Part~(iii)), and lowers government outlays when enforcement is weak (Part~(iv)).
The key economic force is that the ad valorem subsidy restores the link between the billed price and demand, eliminating the insulation rent that arises under the price cap.

An \ac{hcp} switching from \p1\ to \p2c\ is subject to both mechanisms simultaneously: the efficiency gain from ad valorem pricing and the distortion from consortium reallocation (Proposition~\ref{prop:hump}).
Because the consortium inflates the eligible member's billed price above the ad valorem equilibrium ($\tilde{p}_E>p_E^{adv}$), the cost reduction from switching to \p2c\ should be smaller than from switching to \p2.
When the distortion dominates the efficiency gain, the transition can increase costs.
Comparing the two treatment effects therefore provides a direct measure of the cross-subsidization distortion.

Finally, Proposition~\ref{prop:hump}, Part~(iv), predicts that the price distortion $\kappa^*(R)$ first rises and then falls with the ineligible-to-eligible revenue ratio~$R$.
This is a \emph{within-\p2c} prediction: among consortia, the relationship between price inflation and consortium composition should exhibit an inverted~U.

These implications motivate two hypotheses that structure the empirical analysis.

\begin{hypothesis}[Cost effectiveness]
\label{hyp:cost_effectiveness}
The ad valorem mechanism \p2\ reduces prices and government outlays relative to the price cap~\p1.
\end{hypothesis}

\begin{hypothesis}[Cross-subsidization]
\label{hyp:cross_subsidization}
Cross-subsidization under \p2c\ erodes the ad valorem efficiency gain.
The cost reduction from switching to \p2c\ is smaller than from switching to \p2, and when the distortion is sufficiently large it can offset the intended benefits entirely.
\end{hypothesis}

\noindent
Together, these hypotheses capture the paper's central tension: the 2014 reform successfully corrected a design flaw in subsidy delivery, but inadvertently opened a channel for strategic manipulation that partially undoes the correction.
Section~\ref{sec:treatment} develops the econometric framework and Section~\ref{sec:apply} tests both hypotheses.

\section{Econometric framework}
\label{sec:treatment}

\subsection{Setup and identification}
\label{sec:did_setup}

Our empirical setting is a two-period \ac{did} model with two treatment margins.
Let $Y_{it}(d)$ denote the potential outcome for \ac{hcp} $i$ at time $t \in \{0,1\}$ under treatment status $d \in \{1,\, 2,\, 2^c\}$, where $d=1$ corresponds to \p1, $d=2$ to \p2, and $d=2^c$ to \p2c.
Period $t=0$ is 2013 and $t=1$ is 2014.
Because \p2\ and \p2c\ were introduced in 2014, all \acp{hcp} are on \p1\ at $t=0$.
In 2014, each \ac{hcp} either remains on \p1\ or switches to \p2\ or \p2c.

This yields three treatment groups.
Let $g_{d,d'}$ denote the group of \acp{hcp} in status $d$ at $t=0$ and status $d'$ at $t=1$.
Since all \acp{hcp} start on \p1, the set of groups is $\mathcal{G} \equiv \{g_{1,1},\; g_{1,2},\; g_{1,2^c}\}$, where $g_{1,1}$ denotes stayers, $g_{1,2}$ denotes those that switch to \p2, and $g_{1,2^c}$ denotes those that switch to \p2c.
We write $D_{i,2} \equiv \mathbf{1}[G_i = g_{1,2}]$ and $D_{i,2^c} \equiv \mathbf{1}[G_i = g_{1,2^c}]$ for the treatment group indicators, and $T_t \equiv \mathbf{1}[t=1]$ for the post-treatment period.

We target two group-specific average treatment effects on the treated (ATTs), corresponding to the two switching margins.
Table~\ref{tab:att_pt} lists each ATT alongside its identifying parallel trends assumption; the control group in both cases consists of \acp{hcp} that remain on \p1\ in both periods ($g_{1,1}$).%
\footnote{The remaining assumptions (no cross-unit interference (SUTVA), no anticipation, and common support $0<\Pr(G_{i}=g)<1$ for all $g \in \mathcal{G}$) follow the standard \ac{did} framework.}

\begin{table}[htbp]
\centering
\caption{Group-specific treatment effects and parallel trends assumptions.}
\label{tab:att_pt}
\resizebox{\linewidth}{!}{%
\begin{tabular}{ll}
\multicolumn{1}{c}{\small{Treatment Effect (ATT)}} & \multicolumn{1}{c}{\small{Parallel Trends Assumption}} \\
\midrule
$\tau_{1,2} \equiv \mathbb{E}[Y_{i1}(2) - Y_{i1}(1) \mid G_i = g_{1,2}]$ &
$\textbf{PT}_{1,2}\textbf{:}\; \mathbb{E}[Y_{i1}(1) - Y_{i0}(1) \mid G_i = g_{1,2}] = \mathbb{E}[Y_{i1}(1) - Y_{i0}(1) \mid G_i = g_{1,1}]$ \\
\\
$\tau_{1,2^c} \equiv \mathbb{E}[Y_{i1}(2^c) - Y_{i1}(1) \mid G_i = g_{1,2^c}]$ &
$\textbf{PT}_{1,2^c}\textbf{:}\; \mathbb{E}[Y_{i1}(1) - Y_{i0}(1) \mid G_i = g_{1,2^c}] = \mathbb{E}[Y_{i1}(1) - Y_{i0}(1) \mid G_i = g_{1,1}]$ \\
\midrule
\end{tabular}}
\vspace{2pt}
\parbox{\linewidth}{\scriptsize\justifying\textit{Notes:} $Y_{it}(d)$ denotes the potential outcome of \ac{hcp} $i$ at time $t$ under program $d \in \{1, 2, 2^c\}$. $G_i$ is the group assignment: $g_{1,2}$ denotes switchers from \p1{} to \p2{}, $g_{1,2^c}$ from \p1{} to \p2c{}, and $g_{1,1}$ denotes stayers on \p1{}. Each parallel trends assumption requires the counterfactual trend of switchers (had they not switched) to equal the observed trend of stayers.}
\end{table}

Under parallel trends, the ATTs are point identified as:
\begin{align}
\label{eq:tau01}
\tau_{1,2}
&= \mathbb E\bigl[ Y_{i1} - Y_{i0}\mid G_i=g_{1,2}\bigr] - \mathbb E\bigl[ Y_{i1} - Y_{i0}\mid G_i=g_{1,1}\bigr],\\
\label{eq:tau02}
\tau_{1,2^c}
&= \mathbb E\bigl[ Y_{i1} - Y_{i0}\mid G_i=g_{1,2^c}\bigr] - \mathbb E\bigl[ Y_{i1} - Y_{i0}\mid G_i=g_{1,1}\bigr].
\end{align}

The theoretical model in Section~\ref{sec:model} delivers clear predictions for these parameters.
The cost-effectiveness channel implies $\tau_{1,2} < 0$: switching from the price cap to the ad valorem mechanism should lower prices because the \ac{hcp} now bears a share of the cost and resists price inflation.
The parameter $\tau_{1,2^c}$ captures both cost-effectiveness and cross-subsidization simultaneously, since switching from \p1\ to \p2c\ changes both the reimbursement formula and the group structure.
The indirect test $\tau_{1,2^c} - \tau_{1,2} > 0$ therefore isolates the cross-subsidization component: if consortia merely offered the same ad valorem savings as individual filing, the two effects would be equal; a positive difference indicates that consortium membership inflates eligible members' costs beyond what the ad valorem mechanism alone would produce.

\subsection{Estimation}
\label{sec:did_estimation}

Since all \acp{hcp} share the same baseline status (\p1\ at $t=0$), the estimating equation is the standard multi-treatment \ac{twfe}:
\begin{align}
\label{eq:FullDiDi}
Y_{it} \;=\; \alpha \;+\; \gamma_{2}\,D_{i,2} \;+\; \gamma_{2^c}\,D_{i,2^c}
\;+\; \lambda\, T_{t}
\;+\; \tau_{1,2}\bigl(D_{i,2}\times T_{t}\bigr)
\;+\; \tau_{1,2^c}\bigl(D_{i,2^c}\times T_{t}\bigr)
\;+\; \varepsilon_{it},
\end{align}
where $\alpha$ is the intercept, $\gamma_2$ and $\gamma_{2^c}$ are group fixed effects, $\lambda$ is the period effect, and $\varepsilon_{it}$ is a mean-zero error term.
The coefficients $\tau_{1,2}$ and $\tau_{1,2^c}$ recover the ATTs in equations \eqref{eq:tau01}--\eqref{eq:tau02}.%
\footnote{In a balanced panel, equation \eqref{eq:FullDiDi} simplifies to first differences:
$\Delta Y_{it} = \lambda + \tau_{1,2}\, D_{i,2} + \tau_{1,2^c}\, D_{i,2^c} + \Delta \varepsilon_{it},$
where $\Delta Y_{it} \equiv Y_{i1} - Y_{i0}$.}

\subsection{Continuous treatment intensity}
\label{sec:did_continuous}

Each \ac{hcp} may comprise multiple facilities, such as a hospital chain with several campuses, each of which independently chooses among \p1, \p2, and \p2c.
Because the data lack facility identifiers but map every facility to its parent \ac{hcp}, we work at the \ac{hcp} level.
Let $\mathcal{I}_j$ denote the set of facilities in \ac{hcp} $j$, with quantity (Mbps) $Q_{it}$ and price $Y_{it}$ observed at the facility level.
We construct the quantity-weighted average price for \ac{hcp} $j$ in period $t$:
\[
\bar{Y}_{jt} = \frac{\sum_{i \in \mathcal{I}_j} Y_{it}\,Q_{it}}{\sum_{k \in \mathcal{I}_j} Q_{kt}},
\]
and define the continuous treatment shares
\[
S_{jt,2} = \frac{\sum_{i\in \mathcal{I}_j} D_{i,2}\, Q_{it}}{\sum_{k \in \mathcal{I}_j} Q_{kt}}, \qquad
S_{jt,2^c} = \frac{\sum_{i\in \mathcal{I}_j} D_{i,2^c}\, Q_{it}}{\sum_{k \in \mathcal{I}_j} Q_{kt}},
\]
as the fractions of \ac{hcp} $j$'s bandwidth allocated to \p2\ and \p2c\ at time $t$, respectively.
These shares replace the binary indicators in equation \eqref{eq:FullDiDi}:
\begin{align}
\label{eq:LinReg_S}
\bar{Y}_{jt} \;=\; \alpha \;+\; \gamma_{2}\,S_{jt,2} \;+\; \gamma_{2^c}\,S_{jt,2^c}
\;+\; \lambda\, T_{t}
\;+\; \tau_{1,2}\bigl(S_{jt,2}\times T_{t}\bigr)
\;+\; \tau_{1,2^c}\bigl(S_{jt,2^c}\times T_{t}\bigr)
\;+\; \xi_{jt},
\end{align}
where $\xi_{jt}$ is the quantity-weighted average of facility-level errors.
In this specification, $\tau_{1,2}$ measures the effect of shifting all of an \ac{hcp}'s subsidized activity from \p1\ to \p2; intermediate values of $S_{jt,2}$ identify the effect from partial transitions, under linearity.

\subsection{Covariates and inference}
\label{sec:did_covariates}

We augment equation \eqref{eq:LinReg_S} with a vector of observed covariates $\bm{X}_{jt}$ that includes the logarithm of internet speed, \ac{hcp} type indicators, internet service type indicators, state fixed effects, and year fixed effects.
The linear specification is:
\begin{equation}
\label{eq:FullDiDj}
\bar{Y}_{jt}
\;=\;
\bm{S}_{jt}'\,\bm{\theta}_0
\;+\;
\bm{X}_{jt}'\,\bm{\beta}_0
\;+\;
\xi_{jt},
\end{equation}
where $\bm{S}_{jt}$ stacks the treatment share variables and their time interactions, $\bm{\theta}_0$ collects the corresponding parameters including the ATTs of interest, and $\bm{\beta}_0$ captures the covariate effects.
Standard errors are clustered at the \ac{hcp} level throughout.

Linear control for $\bm{X}_{jt}$ may be misspecified if the relationship between covariates and outcomes is nonlinear.
To address this, we supplement the baseline with a \ac{dml} estimator \citep{chernozhukov2018double}, which replaces the linear control function with a flexible, unknown function $g_0(\bm{X}_{jt})$ estimated via random forest.
The method uses Neyman-orthogonal moments and $K$-fold cross-fitting ($K{=}10$) to deliver $\sqrt{n}$-consistent inference on $\bm{\theta}_0$ that is robust to regularization bias in the first-stage nuisance estimation.
Appendix~\ref{sec:machine_learning} details the implementation.
    
\section{An empirical exercise}
\label{sec:apply}

\subsection{Data}
\label{sec:data}

We use the \ac{rhc} Commitments and Disbursements data provided by \ac{usac}. 
The data is made publicly available by an \ac{fcc} mandate.\footnote{\url{https://opendata.usac.org/Rural-Health-Care/RHC-Commitments-and-Disbursements-Tool/sm8n-gg82}}
It comprises the full population of \acp{hcp} who sought subsidies under \p1, \p2, or \p2c, including those who received funding, those who were denied, and those who have requests under review. 
The data spans 2012 to date, with new rows being added on a rolling basis. 
The data has three ``panel'' identifiers: 
funding year serves as the time variable, a unique ID identifies each \ac{hcp}, and another ID identifies each consortium.

The data is at the ``line of subsidy'' level. 
To explain, when a request is submitted, it may demand subsidy for one or multiple lines of internet connection. 
If a request is for one line of internet connection, it would be recorded as one new observation and a unique \ac{frn} identifies the request in the data set. 
Since it demands one line of subsidy, its respective \ac{frnln} would be $1$. 

If a subsidy request comprises multiple lines of internet connection, each line is recorded as a separate row. 
In this case, a unique \ac{frn} would be shared by all lines. 
But each line would be identified by an \ac{frnln} starting at $1$ and rising in integer increments. 
When a consortium application is submitted, all lines of subsidy in it would share the same \ac{frn} and consortium ID, and each request is mapped into its respective \ac{hcp} using the \ac{hcp} ID.

The data underwent an extensive clean-up and exclusion process, detailed in Appendix~\ref{sec:apx_data_cleanup} and summarized in Table~\ref{tab:ab_exclusion}.
Since \p2\ and \p2c\ were implemented in 2014, the baseline analysis focuses on the 2013--2014 year pair.
The 2013--2014 sample starts with 36,576 request-level observations.
We limit to annual contracts and to categories of expense that correspond to internet subscription costs (excluding equipment, maintenance, and infrastructure).
We remove grandfathered requests, pending and withdrawn requests, and cases for which we could not recover internet speed, price, or subsidy amount.
We discard observations located in Alaska because it is an outlier state with exorbitant internet prices, lowest speed, and a unique internet technology profile.
We limit speeds to 1--100\ac{mbps} and drop observations with zero subsidy, yielding 7,780 request-level observations for 2013--2014.

Since the finest identifier is at the \ac{hcp} level, to enable panel analysis, we aggregate the observations at the \ac{hcp} and year level.
We obtain the \ac{hcp}-level speed by summing speeds in a given year that are linked to a common \ac{hcp} ID.
A similar process is used to obtain \ac{hcp}-level values of price and subsidy.
In this process, we keep track of the number of lines of subsidy per \ac{hcp} and use it as a measure of \ac{hcp} size.
After aggregation, the 2013--2014 sample contains 4,231 \ac{hcp}-year observations.
Restricting to \acp{hcp} that appear in both years produces a balanced panel of 1,940 observations (970 \acp{hcp} observed twice).

\begin{table}[htbp]
\centering
\caption{Program utilization, 2013--2014.}
\label{tab:bd_descriptive2_2013_2014}
\begin{minipage}{\columnwidth}
\centering
\small
\begin{tabular*}{\linewidth}{@{\extracolsep{\fill}} lcccc}
\Xhline{2pt}
 & {2013} & \multicolumn{3}{c}{2014} \\
\cline{2-2}\cline{3-5}
 & \p1 & \p1 & \p2 & \p2c \\ \hline
Number of \acp{hcp}                 & 970 & 643 & 419 & 43 \\
Number of requests                  & 1,938 & 1,263 & 850 & 8 \\
Number of lines                     & 1,980 & 1,305 & 973 & 56 \\
Percentage of lines                 & 100.0\% & 55.9\% & 41.7\% & 2.4\% \\
Percentage of speed                 & 100.0\% & 35.3\% & 63.1\% & 1.7\% \\
Percentage of subsidies             & 100.0\% & 77.9\% & 16.5\% & 5.5\% \\
\Xhline{2pt}
\end{tabular*}
\vspace{2pt}
\parbox{\linewidth}{\scriptsize\justifying\textit{Notes:} The table reports program utilization for the 970 baseline \acp{hcp}. A request may comprise multiple lines of internet connection. \p1{} denotes the price-cap program, \p2{} the individual ad valorem program, and \p2c{} the consortium ad valorem program. In 2013, only \p1{} was available; \p2{} and \p2c{} were introduced in 2014. Percentage figures show each program's share of total lines, bandwidth, and subsidy expenditure in a given year.}
\end{minipage}
\end{table}

Table~\ref{tab:bd_descriptive2_2013_2014} shows the number of approved subsidy requests, \acp{hcp}, and lines of internet connection by program for the baseline 2013--2014 sample.
In 2013, all 970 baseline \acp{hcp} were on \p1.
By 2014, 419 had switched to \p2\ and 43 to \p2c, while 643 remained on \p1.
Notably, one \ac{hcp} could be counted under multiple programs if it receives subsidies under more than one program in the same year.
Although \p2c\ accounts for only 2.1\% of lines, it absorbs 4.1\% of subsidies, an early indicator that consortium filing is costlier per line than individual filing.

The three programs use the same eligibility criteria and provide the same coverage via the same bureaucracy. So, rational \acp{hcp} are expected to choose the program that minimizes \ac{hcp} net cost. However, an anomaly occurred in 2014 when \acp{hcp} who switched from \p1\ to \p2c\ accepted unusually high \ac{hcp} net costs.
This pattern, which seems to violate rationality, may be explained by cross-subsidization. 
Cross-subsidization generates inconspicuous gains for the \acp{hcp}.
An \ac{hcp} may switch to a program with a higher \ac{hcp} net cost if the gains from cross-subsidization exceed the associated increase in \ac{hcp} net cost. 
The year 2014 marks the first year in which cross-subsidization was enabled.
Therefore, \acp{hcp} who could gain the most from cross-subsidization had the strongest incentive to switch immediately, creating a visible increase in \ac{hcp} net cost early on. 

\subsection{A descriptive regression analysis}
\label{sec:pooledOLS}

This section discusses the summary statistics and reports a \ac{pols} regression exercise to describe the data and uncover patterns. 
Table \ref{tab:bb_summary_pooledOLS} reports the summary statistics at the level of subsidy line. 
There are large variations in price, subsidy, and \ac{hcp} net cost.
\p1\ was the only option in 2013.
In the 2013--2014 sample, \p1\ makes up 75.7\% of observations (3,205 of 4,234), reflecting its dominance in the pre-treatment year.
\p2\ accounts for 23.0\% and \p2c\ for 1.3\%.
Most of the sample comprises non-profit hospitals, rural health clinics, and community health centers.
\enquote{Ethernet} and \enquote{internet} dominate the service types.
The data spans most \ac{us} states, with higher concentrations in California, Wisconsin, Texas, and Illinois.

\begin{table}[htbp]
\centering
\caption{Summary statistics for the POLS model.}
\label{tab:bb_summary_pooledOLS}
\begin{minipage}{\columnwidth}
\centering
\resizebox{\linewidth}{!}{%
\begin{tabular}{llllllllll}
\Xhline{2pt}
Variable & Obs & Mean & SE & Median & Min & Max & & & \\ \hline
\p1\ & 3,205 & & & & & & & & \\
\p2\ & 973 & & & & & & & & \\
\p2c\ & 56 & & & & & & & & \\
Speed (Mbps) & 4,234 & 13.14 & 0.350 & 1.54 & 1.50 & 100 & & & \\
Price (\$) & 4,234 & 14,044 & 314.70 & 8,328 & 60.00 & 270,024 & & & \\
Subsidy (\$) & 4,234 & 10,706 & 266.38 & 5,715 & 39.00 & 242,974 & & & \\
HCP net cost (\$) & 4,234 & 3,338 & 76.45 & 2,106 & 0.00 & 87,462 & & & \\
\hline
\textbf{HCP type} & \textbf{Obs} & & & & & & & & \\ \hline
Non-profit hospital & 2,352 & Local health department & 266 & Emergency room & 17 &  &  &  &  \\
Rural health clinic & 808 & Mental health center & 265 & Medical school & 6 &  &  &  &  \\
Community health center & 288 & Not available & 232 &  &  &  &  &  &  \\
\hline
\textbf{Service type} & \textbf{Obs} & & & & & & & & \\ \hline
T1 & 2,084 & MPLS & 390 & T3 & 75 & Fiber & 31 &  &  \\
Ethernet & 504 & Internet & 270 & DSL & 65 & VPN & 20 &  &  \\
ISDN & 472 & PRI & 246 & Voice\_Grade & 36 & Other categories & 41 &  &  \\
\hline
\textbf{State} & \textbf{Obs} & & & & & & & & \\ \hline
GA & 585 & MN & 106 & NY & 50 & PA & 20 &  &  \\
WI & 507 & OH & 105 & NM & 49 & SC & 16 &  &  \\
TX & 272 & IA & 100 & OR & 49 & UT & 13 &  &  \\
CA & 261 & MO & 94 & LA & 43 & NH & 11 &  &  \\
VA & 253 & MI & 72 & MT & 36 & ND & 10 &  &  \\
NE & 222 & KS & 67 & OK & 33 & WV & 5 &  &  \\
KY & 216 & SD & 62 & WY & 30 & MA & 4 &  &  \\
MS & 180 & AL & 61 & CO & 28 & FL & 2 &  &  \\
AR & 171 & TN & 57 & NV & 24 &  &  &  &  \\
IL & 152 & NC & 53 & WA & 21 &  &  &  &  \\
AZ & 122 & IN & 52 & ID & 20 &  &  &  &  \\
\Xhline{2pt}
\end{tabular}%
}
\vspace{2pt}
\parbox{\linewidth}{\scriptsize\justifying\textit{Notes:} The table reports summary statistics for the data used in the \ac{pols} model, restricted to the 970 baseline \acp{hcp} over the 2013--2014 period. This includes the number of lines of subsidy under each program, calendar year, internet speed and price measures, \ac{hcp} types, type of internet service, and state where the \ac{hcp} is located. Observations with non-positive \ac{hcp} net cost are excluded.}
\end{minipage}
\end{table}

We use the following specification:

\[
\ln(LHS_{it}) = \beta_1 \mathcal{P}_1 + \beta_2 \mathcal{P}_2 + \beta_3 \mathcal{P}_2^c + \gamma \ln{speed} 
+ \lambda_t + \alpha_s + \psi_i + \theta_z + \epsilon_{it}
\]

where $\ln(LHS_{it})$ is the natural logarithm of the \ac{lhs} variable (price, subsidy, and \ac{hcp} net cost); $\lambda_t$, $\alpha_s$, $\psi_i$, $\theta_z$, respectively, capture the \ac{fe} for year, state, \ac{hcp} type, and internet service type; and $\epsilon_{it}$ is the error term with a standard normal distribution.

\begin{table}[htbp]
\centering
\caption{Pooled OLS regression results.}
\label{tab:ac_pooledOLS1}
\begin{minipage}{\columnwidth}
\centering
\small
\begin{tabular*}{\linewidth}{@{\extracolsep{\fill}} lccc}
\Xhline{2pt}
 & $\ln(\text{price})$ & $\ln(\text{subsidy})$ & $\ln(\text{HCP net cost})$ \\ \hline
$\beta_1$ & 7.1634*** & 7.0166*** & 5.3150*** \\
 & (0.1345) & (0.1759) & (0.2261) \\
 & [0.000] & [0.000] & [0.000] \\
$\beta_2$ & 5.6567*** & 5.5327*** & 4.4648*** \\
 & (0.1363) & (0.1782) & (0.2290) \\
 & [0.000] & [0.000] & [0.000] \\
$\beta_3$ & 7.2355*** & 6.8700*** & 6.3324*** \\
 & (0.1668) & (0.2181) & (0.2802) \\
 & [0.000] & [0.000] & [0.000] \\
$\beta_3 - \beta_2$ & 1.5787*** & 1.3373*** & 1.8675*** \\
 & (0.1092) & (0.1427) & (0.1835) \\
 & [0.000] & [0.000] & [0.000] \\
$\ln(\text{speed})$ & 0.3159*** & 0.2827*** & 0.3059*** \\
 & (0.0168) & (0.0220) & (0.0282) \\
 & [0.000] & [0.000] & [0.000] \\
Year FE, HCP type FE, Service type FE, State FE & Y & Y & Y \\
$N$ & 4,234 & 4,234 & 4,234 \\
Adj.~$R^2$ & 0.9946 & 0.9898 & 0.9789 \\
\Xhline{2pt}
\end{tabular*}
\vspace{2pt}
\parbox{\linewidth}{\scriptsize\justifying\textit{Notes:} The table reports pooled OLS estimates for the 2013--2014 period using the 970 baseline \acp{hcp}. The dependent variables are the natural logarithm of price, subsidy, and \ac{hcp} net cost. $\beta_1$, $\beta_2$, and $\beta_3$ are the coefficients on indicator variables for the price-cap program (\p1{}), the individual ad valorem program (\p2{}), and the consortium ad valorem program (\p2c{}), respectively. The row $\beta_3 - \beta_2$ reports the estimated difference with a two-sided test of $H_0\colon \beta_3 = \beta_2$. The model includes no intercept. Observations with non-positive \ac{hcp} net cost are excluded. Standard errors in parentheses; $p$-values in brackets. Significance: $^{***}\, p < 0.01$, $^{**}\, p < 0.05$, $^{*}\, p < 0.1$.}
\end{minipage}
\end{table}

Table \ref{tab:ac_pooledOLS1} shows the estimation results.
Columns 2-4, respectively, have the natural logarithms of price, subsidy, and \ac{hcp} net cost as the dependent variable. 
A high adjusted $R^2$ suggests that the combination of the covariates explains the majority of the price movements. 
Starting with Column 2, internet plans are cheaper under \p2\ than \p1.
This is in line with the cost efficiency hypothesis, that the better design of \p2\ relative to \p1\ induces \acp{hcp} to be more price sensitive and achieve cost savings.
The same pattern holds for subsidy and \ac{hcp} net cost: \p2\ reduces both relative to \p1, consistent with the cost savings being shared between the government and the provider.

Whereas both \p2\ and \p2c\ incentivize cost-saving, \p2c\ also incentivizes cross-subsidization.
Therefore, if \p2c\ prices are greater than \p2\ prices, this would serve as suggestive evidence of cross-subsidization.
The estimates for $\beta_3-\beta_2$ are large and highly significant across all three outcomes, showing that \p2c\ prices are substantially higher than \p2\ prices.
The difference is largest for \ac{hcp} net cost, indicating that consortium members bear a disproportionate share of the price inflation.

\subsection{Main results}
\label{sec:model_main}

Using the baseline panel described in Section~\ref{sec:data}, we estimate the continuous treatment model (eq.\ref{eq:LinReg_S}), its binary counterpart (eq.\ref{eq:FullDiDi}), and a \ac{dml} specification.
The models estimate the effect of switching from \p1\ to \p2\ ($\tau_{01}$) and from \p1\ to \p2c\ ($\tau_{02}$).
Table~\ref{tab:bj_main_reg_combined} reports the baseline results.
The cost-effectiveness hypothesis is confirmed if $\tau_{01}<0$: switching from \p1\ to \p2\ should reduce prices.
The table strongly supports this across all three outcomes and all three estimation methods.

\begin{sidewaystable}[htbp]
\centering
\caption{Combined regression results, 2014.}
\label{tab:bj_main_reg_combined}
\small
\begin{tabular*}{\linewidth}{@{\extracolsep{\fill}} l ccc ccc ccc}
\Xhline{2pt}
 & \multicolumn{3}{c}{\textbf{Panel A: $\ln(\text{price})$}} & \multicolumn{3}{c}{\textbf{Panel B: $\ln(\text{subsidy})$}} & \multicolumn{3}{c}{\textbf{Panel C: $\ln(\text{HCP net cost})$}} \\
\cmidrule(lr){2-4} \cmidrule(lr){5-7} \cmidrule(lr){8-10}
 & Cont. & Binary & DML & Cont. & Binary & DML & Cont. & Binary & DML \\ \hline
$\tau_{01}$ & -1.249*** & -1.326*** & -1.755*** & -1.262*** & -1.349*** & -1.697*** & -0.704*** & -0.731*** & -1.295*** \\
 & (0.085) & (0.097) & (0.212) & (0.102) & (0.116) & (0.225) & (0.074) & (0.082) & (0.207) \\
 & [0.000] & [0.000] & [0.000] & [0.000] & [0.000] & [0.000] & [0.000] & [0.000] & [0.000] \\
$\tau_{02}$ & 0.556*** & 0.600*** & 0.569** & 0.284* & 0.310 & 0.381 & 1.740*** & 1.863*** & 1.431*** \\
 & (0.136) & (0.158) & (0.264) & (0.164) & (0.190) & (0.269) & (0.119) & (0.135) & (0.262) \\
 & [0.000] & [0.000] & [0.031] & [0.084] & [0.104] & [0.156] & [0.000] & [0.000] & [0.000] \\
$\tau_{02} - \tau_{01}$ & 1.804*** & 1.926*** & 2.324*** & 1.546*** & 1.659*** & 2.078*** & 2.445*** & 2.594*** & 2.726*** \\
 & (0.147) & (0.171) & (0.372) & (0.177) & (0.206) & (0.375) & (0.128) & (0.146) & (0.398) \\
 & [0.000] & [0.000] & [0.000] & [0.000] & [0.000] & [0.000] & [0.000] & [0.000] & [0.000] \\
$N$ & 1,940 & 1,670 & 1,940 & 1,940 & 1,670 & 1,940 & 1,940 & 1,670 & 1,940 \\
$R^2$ & 0.387 & 0.390 &  & 0.312 & 0.311 &  & 0.431 & 0.448 &  \\
\Xhline{2pt}
\end{tabular*}
\vspace{2pt}
\parbox{\linewidth}{\scriptsize\justifying\textit{Notes:} $\tau_{01}$ is the estimated treatment effect of switching from \p1{} to \p2{}, and $\tau_{02}$ is the effect of switching from \p1{} to \p2c{}. The row $\tau_{02} - \tau_{01}$ reports the estimated difference, with the standard error computed from the joint influence-function covariance matrix for the DML column. Cont.~= TWFE with continuous treatment shares; Binary = TWFE with binary treatment indicators; DML = Double Machine Learning with random forests (10-fold cross-validation). A level specification of this table (with price, subsidy, and HCP net cost in thousands of dollars as dependent variables) is reported in Table~\ref{tab:bj_main_reg_levels} in the appendix. Standard errors in parentheses; $p$-values in brackets. Significance: $^{***}\, p < 0.01$, $^{**}\, p < 0.05$, $^{*}\, p < 0.1$.}
\end{sidewaystable}

To test cross-subsidization, we note that $\tau_{12}$ is not directly estimable because \p2\ was not available in 2013, so no \ac{hcp} could switch from \p2\ to \p2c.
We circumvent this by taking advantage of the following rationale:
$\tau_{01}$ captures cost effectiveness, $\tau_{12}$ captures cross-subsidization, and $\tau_{02}$ captures both.
Therefore, a $t$-test of $\tau_{02}-\tau_{01}>0$ indirectly captures cross-subsidization.
Table~\ref{tab:bj_main_reg_combined} strongly supports this hypothesis as well.

The binary treatment specification (Panel~B in Table~\ref{tab:bj_main_reg_combined}) discards \acp{hcp} with mixed program participation within a year, yielding a slightly smaller sample.
\label{sec:ML_results}The \ac{dml} estimates (Panel~C) allow covariates to enter the model nonlinearly using random forests with 10-fold cross-validation (see Appendix~\ref{sec:machine_learning} for details). The estimates are close to the \ac{twfe} baseline and highly significant.
Both alternatives corroborate the continuous treatment findings, confirming that the cross-subsidization result is not an artifact of the linear specification or the choice of estimator.
Table~\ref{tab:bj_twfe_allyears} in the appendix reports results for all year pairs from 2014 through 2021, estimated analogously by limiting each pair $(t-1,t)$ to \acp{hcp} that appear in both years.

Appendix~\ref{sec:robust} reports nine additional robustness checks that vary the sample composition (excluding medical schools, limiting to specific \ac{hcp} types, Ethernet-only service, including Alaska) and data processing assumptions (all speed ranges, all contract durations, Mbps-only speeds). The results are robust throughout.

We also verify that the estimates are not driven by a small number of outliers. Appendix~\ref{sec:apx_influence} applies Cook's Distance diagnostics to the baseline specification and re-estimates the model after dropping all influential \acp{hcp}. The treatment effects remain significant and economically large. A related concern is that \p2c{} participants operate in a narrower speed range than those in \p1{} or \p2{}, so the baseline comparison may not be on common support. Appendix~\ref{sec:apx_common_support} imposes a common-support restriction by limiting all programs to the speed domain of \p2c{}. The coefficients are virtually unchanged.
Finally, Appendix~\ref{sec:apx_oster} applies the coefficient stability test of \citet{oster2019unobservable}. Every $\delta$ is negative, meaning that adding time-varying controls moves the treatment coefficients \textit{further} from zero. Unobservables would need to work in the opposite direction of observables to nullify the effects.

Table~\ref{tab:bk_model_comparison} compares eight price--speed functional forms. The top goodness-of-fit choices are log-log and lin-log.
A Box-Cox test (Table~\ref{tab:bk_boxcox}) rejects both nested specifications, with the maximum-likelihood $\hat{\lambda}$ closer to zero (log-log) than to one (lin-log). We adopt the log-log specification as the baseline.
Since the true transformation lies between the two, we present results under both: Table~\ref{tab:bj_main_reg_levels} reports the lin-log specification alongside the baseline log-log. All key coefficients remain significant and stable, and Table~\ref{tab:bj_elasticity_comparison} confirms that the implied semi-elasticities are comparable in magnitude. Because the two specifications bracket the Box-Cox optimum and yield virtually identical conclusions, the results are not sensitive to the choice of functional form.

\subsection{Testing the parallel-trends assumptions}
\label{sec:parallel}

Central to our analysis is the parallel trends assumption, that absent the effect of switching, the evolution of the outcome variables for \acp{hcp} who switched would be similar to those who did not switch. 
In this section, we briefly describe the approach of \citet{aryal2025benefits}, inspired by \citet{manski2018right}, in testing the sensitivity of the results to a measurable violation of the parallel trends assumption. 
We then proceed to explain how their model, which may not immediately fit our application, could still be applied to our case after a tweak. 

\citet{aryal2025benefits} let the unobserved counterfactual trend for the treatment group ($\beta^{1}$) be proportional to the observed trend of the control group ($\beta^{0}$) by scalar $g \in [0,2]$ such that $\beta^{1} = g \cdot \beta^{0}$. 
The standard parallel trends assumption corresponds to $g=1$. 
A robust treatment effect, $\hat{\beta}^{\text{Robust}}$, which adjusts for this potential violation, is then given by:

\[
\hat{\beta}^{\text{Robust}} = \hat{\beta}^{\text{DiD}} + \left( \hat{\beta}^{0} - \hat{\beta}^{1} \right) = \hat{\beta}^{\text{DiD}} + (1 - g) \cdot \hat{\beta}^{0}
\]

Measuring $\hat{\beta}^{\text{Robust}}$ and testing if its 95\% confidence interval includes zero would provide insights into whether the \ac{did} results could tolerate a known degree and shape of violation of the parallel trends assumption.
This approach is not immediately available to our main specification because it requires $T_t$ to appear as a stand-alone explanatory variable, which is not the case.
We use the following modification in eq.\ref{eq:FullDiDi} to enable \citet{aryal2025benefits}'s technique for $\tau_{01}$, and then separately for $\tau_{02}$. 
In 2013, all subsidy requests were in \p1. 
In 2014, the requests were under \p1, \p2, or \p2c. 
If we discard 2014's \p2c observations, \acp{hcp} who were subsidized in 2013 and 2014 could be placed either in $D_{00}$ or in $D_{01}$, hence, this subset satisfies $D_{02}=D_{11}=D_{12}=D_{22}=0$. 
This simplifies eq.\ref{eq:FullDiDi} to a canonical \ac{did} with a single treatment effect, enabling the test for $\tau_{01}$. 
Likewise, we discard 2014's \p2\ observations (keep \p2c) to enable the test for $\tau_{02}$. 

Implementing the test on a subset of the data may limit inference. 
But if such tests produce estimates that are consistent with the main results, and the results survive the test of \citet{aryal2025benefits}, it would lend support to the robustness of the baseline results to mild violations of the parallel trends assumption. We do this test for $\tau_{01}$ and $\tau_{02}$, separately for price, subsidy, and \ac{hcp} net cost. 
It generates six tests. 

\begin{figure}[htbp]
    \centering
    \begin{minipage}{\linewidth}
    \centering
    \includegraphics[width=\linewidth]{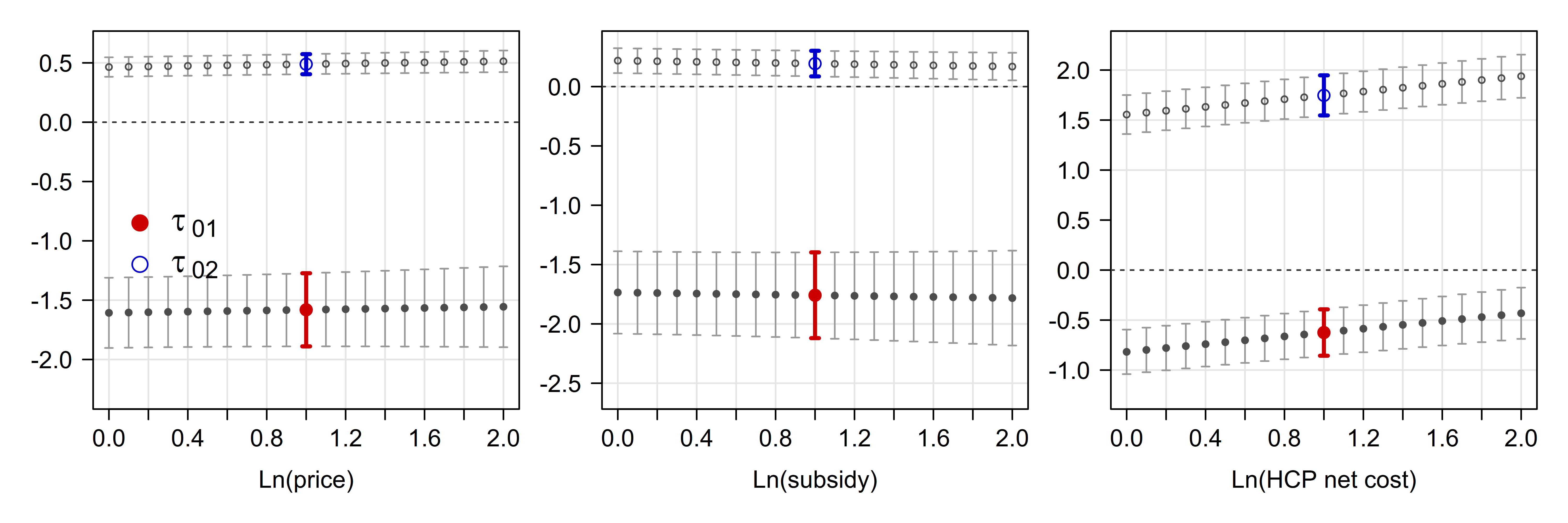}
    \caption{Testing the parallel-trends assumption.}
	\label{fig:ab_manski_combined}
    \vspace{2pt}
    \parbox{\linewidth}{\scriptsize\justifying\textit{Notes:} This figure tests the robustness of the results to a measured linear violation of the parallel-trends assumption \citep{aryal2025benefits, manski2018right}. From left to right, the three panels respectively show the effect of switching between programs on the natural logarithms of price, subsidy, and \ac{hcp} net cost. The red bar is the baseline effect of switching from program \p1\ to \p2. The whiskers create a 95\% confidence interval around the estimate. The gray bars extended to the left and right show what happens if instead of $g=1$, $g$ varies between 0 and 2. Likewise, the blue bar shows the effect of switching from program \p1\ to \p2c. }
    \end{minipage}
\end{figure}

The results are reported in Figure \ref{fig:ab_manski_combined}. 
The two tests for $\tau_{01}$ and $\tau_{02}$ are separately done for price, but put in the same frame (left panel) to facilitate analysis. 
Likewise, the two tests for subsidy in the middle, and the two tests for \ac{hcp} net cost on the right, are placed together. 
The values of $g$ are shown on the horizontal axis, and the treatment effects are measured in the vertical axis. 
The estimates under the parallel trends assumption ($g=1$) are shown in color (solid red for $\tau_{01}$, and hollow blue for $\tau_{02}$). 
The estimates under deviations from the parallel trends assumption are shown in gray, extending to the left ($g<1$) and right ($g>1$) of each central point.

The results with $g=1$ confirm the baseline results, that $\tau_{01}<0$ and $\tau_{01}<\tau_{02}$.
In this regard, this subset of the data behaves similarly to the baseline in supporting the two hypotheses. 
As $g$ varies in $[0,2]$, all the estimates maintain their significance and economic magnitude. 
Another observation that may be cautiously made is that, in each panel, the gap between the blue and red trend, which is sizable and persistent, indirectly corroborates cross-subsidization.

\subsection{An empirical test of cross-subsidization}
\label{sec:share_ineligible}

We found that \p2c\ prices are consistently higher than \p2\ prices, and that this could be explained by cross-subsidization. 
This section presents a test that sheds light on whether cross-subsidization is the underlying mechanism.
Cross-subsidization requires a mix of eligible and ineligible members within a consortium. It is impossible when all members share the same status. So, it must peak at an interior mix. This predicts an inverted U-shaped relationship between the internet price for eligible members and the fraction of members that are ineligible.

To test this hypothesis, we limit the data to consortium applications (\p2c).
So far in the empirical analysis, we strictly used the data of eligible \acp{hcp}. 
In this section, we need to identify ineligible members and measure the fraction of each consortium that they make up.
We identify ineligible members as \acp{hcp} who
(1) are part of a consortium application that has been approved and some members of the consortium have been reimbursed, 
(2) have internet speed and price information of their requested internet on file, and
(3) have received no subsidy. 

We use two methods to calculate the fraction of each consortium that is made up of ineligible \acp{hcp}. 
Our main measure is based on speed, i.e., a consortium's ineligible fraction is equal to the fraction of its total speed that belongs to ineligible members. 
In this method, speed is taken as a measure of the quantity of internet consumed.
Our second approach sets the fraction equal to the percentage count of ineligible members.
For example, when two out of five consortium members are ineligible, the fraction would be 0.40. 

\begin{table}[htbp]
\centering
\caption{Summary statistics of consortium members.}
\label{tab:ae_summary_fwl}
\begin{minipage}{\columnwidth}
\centering
\small
\begin{tabular*}{\linewidth}{@{\extracolsep{\fill}} lcccccc}
\Xhline{2pt}
Variable & Obs & Mean & SE & Median & Min & Max \\ \hline
Mean price for eligible \acp{hcp} & 628 & 11340.5 & 262.6 & 10312.1 & 321.0 & 63216.5 \\
Mean speed for eligible \acp{hcp} & 628 & 1397.0 & 55.8 & 996.4 & 1.5 & 8192.0 \\
Mean number of bidders & 628 & 0.849 & 0.055 & 0.140 & 0.000 & 9.938 \\
Ineligible \acp{hcp}' share by speed & 628 & 0.196 & 0.009 & 0.112 & 0.000 & 1.000 \\
Ineligible \acp{hcp}' share by count & 628 & 0.217 & 0.007 & 0.162 & 0.010 & 0.944 \\
Consortium's total speed & 628 & 77133.0 & 7007.6 & 25274.9 & 12.4 & 1801124.7 \\
\Xhline{2pt}
\end{tabular*}
\vspace{2pt}
\parbox{\linewidth}{\scriptsize\justifying\textit{Notes:} This table reports the summary statistics exclusively for consortium members, pooled across all data years. \acp{hcp} that are members of a consortium could be eligible for subsidies (via program \p2c) or ineligible. The table reports the mean values of internet price and speed for eligible members, the number of \acp{isp} who participate in the competitive bidding process, and the fraction of consortium members that are ineligible in terms of speed and count.}
\end{minipage}
\end{table}

Table \ref{tab:ae_summary_fwl} shows the summary statistics for the data used in this practice. 
All data years that had sufficient observations of consortium applications were combined. 
Mean price and mean speed reflect those of eligible members. 
An average consortium comprises 21.7\% ineligible members who consume 19.6\% of the consortium's total internet speed.

Our estimation implements an \ac{fwl} residualization procedure to isolate the variation in treatment intensity that is orthogonal to high-dimensional controls. 
Specifically, we regress $ln(price)$ and the fractions of ineligible members on the following variables: 
mean number of bidders who participated in the auction, the natural logarithm of the average internet speed requested per eligible member, the natural logarithm of the sum of all speeds requested by the consortium (a measure of consortium size), year \ac{fe}, and consortium \ac{fe}. 
Residuals from these regressions are then adjusted back to the original variable means to facilitate interpretation. 
We then apply local linear regression to the residualized data using a nonparametric locally weighted scatterplot smoothing (LOWESS) estimator.
This approach flexibly estimates the conditional relationship between the fraction of ineligible participants and $ln(price)$ without imposing a parametric functional form. 
This was implemented in R, using \textit{ggplot} with \textit{geom smooth} method set to \textit{loess}.

\begin{figure}[!htb]
    \centering
    \includegraphics[width=.8\linewidth]{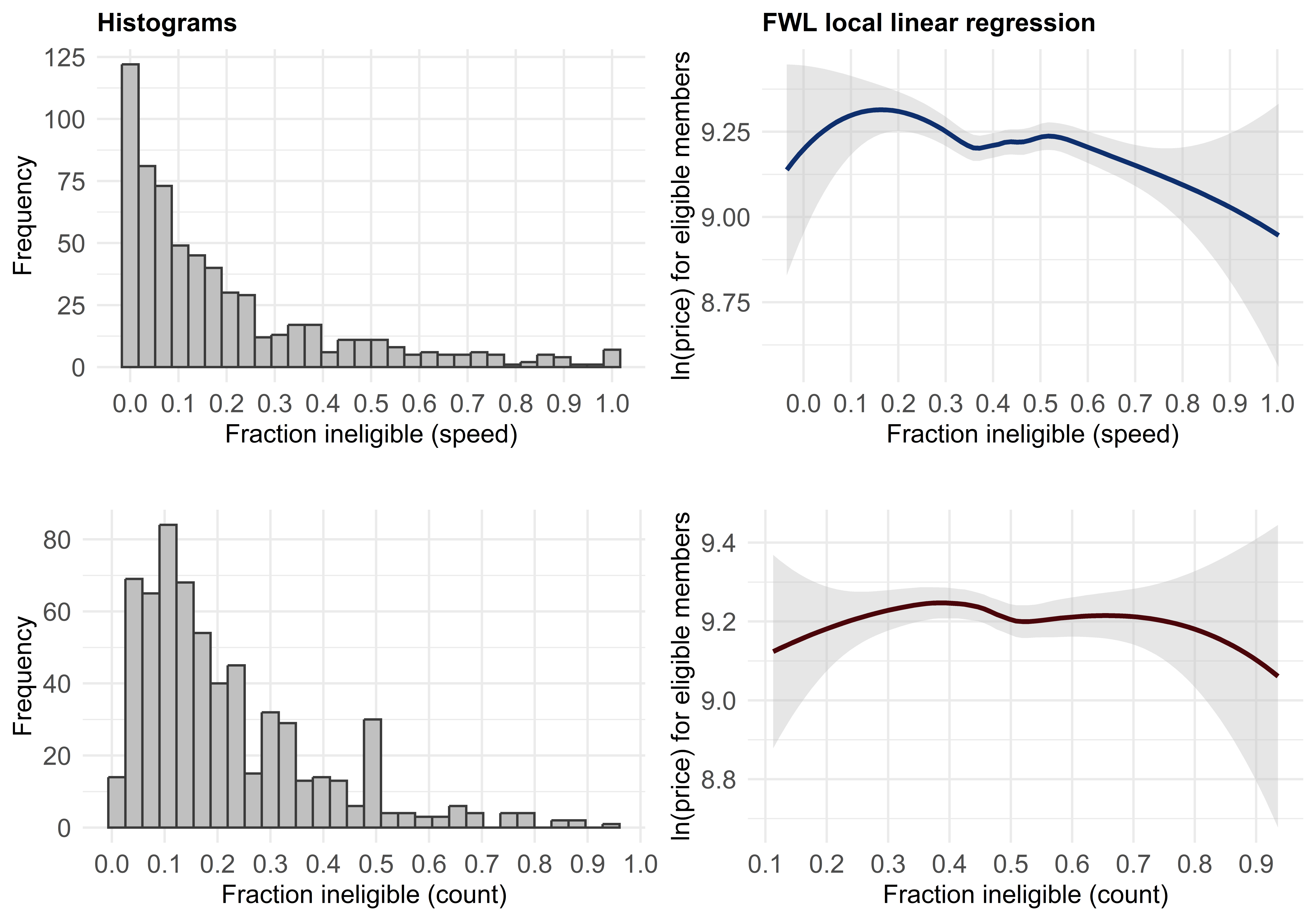}
	\caption{\ac{fwl} local linear regression results.}
	\label{fig:ae_fwl1}
    \vspace{2pt}
    \parbox{\linewidth}{\scriptsize\justifying\textit{Notes:} Left panels show histograms of the fraction of consortium members ineligible for subsidies, measured by share of total consortium speed (top) or by headcount (bottom). Right panels show Frisch--Waugh--Lovell residualized LOESS fits of $\ln(\text{price})$ for eligible members against the ineligible fraction, after partialling out consortium and year fixed effects, $\ln(\text{total consortium speed})$, $\ln(\text{mean speed})$, and mean number of bidders.}
\end{figure}

Figure \ref{fig:ae_fwl1} shows the result of this practice.
Although the sample is small, the results corroborate the presence of an inverted U-shape between the fraction of ineligible members and $ln(price)$, consistent with the underlying role of cross-subsidization in explaining the higher prices in \p2c.

\subsection{Endogeneity in program selection}
\label{sec:Endo}

A natural concern is that the estimated treatment effects reflect selection rather than the causal impact of program design.
\acp{hcp} choose which mechanism to join, and this choice may correlate with unobserved time-varying characteristics that also affect prices, violating the parallel trends assumption.
Our estimates capture the average treatment effect on the treated (ATT) under parallel trends.
We do not claim to eliminate endogeneity. Rather, several features of the setting and the results collectively establish it as a second-order concern.

Programs $\mathcal{P}_1$, $\mathcal{P}_2$, and $\mathcal{P}_2^c$ share identical eligibility criteria and administrative procedures.
The only substantive difference across programs is the reimbursement formula.
Conditional on eligibility, the choice is purely a cost-minimization decision.
Selection therefore operates through the same channel as the treatment effect: an \ac{hcp} switches because the alternative program lowers its cost exposure, and the within-\ac{hcp} price change after switching measures exactly that incentive at work.
Hospital fixed effects absorb all time-invariant heterogeneity.

To see what drives switching, note that a cost-minimizing \ac{hcp} prefers $\mathcal{P}_2$ if and only if $0.35p < p_u$, i.e., $p/p_u < 1/0.35 \approx 2.86$.
We define $H_i = \mathbf{1}[p_i/p_{u,i} > 1/0.35]$ using 2013 values and estimate a logit model of the 2014 switching decision on the balanced panel:
\[
\Pr(\text{switched}_i = 1) = \Lambda\!\left(\beta_0 + \beta_1 \, H_i + \mathbf{X}_i'\boldsymbol{\gamma}\right)
\]
where $\mathbf{X}_i$ includes $\ln(\text{speed})$, $\ln(\text{price})$, and $\ln(\text{requests})$.
Table~\ref{tab:zz4_logit_switching} shows that $H_i$ is negative and highly significant across four specifications with increasingly rich controls. Adding \ac{hcp}-level covariates contributes virtually no explanatory power (pseudo $R^2$: 0.031 to 0.037), confirming that the switching decision is driven predominantly by the cost signal.

\begin{table}[htbp]
\centering
\caption{Logit estimates of the switching decision from $\mathcal{P}_1$ to $\mathcal{P}_2$ in 2014.}
\label{tab:zz4_logit_switching}
\begin{minipage}{\columnwidth}
\centering
\small
\begin{tabular*}{\linewidth}{@{\extracolsep{\fill}} lcccc}
\Xhline{2pt}
 & (1) & (2) & (3) & (4) \\ \hline
$H_i$ & $-$0.966*** & $-$0.990*** & $-$1.094*** & $-$1.079*** \\
 & (0.207) & (0.208) & (0.226) & (0.229) \\
$\ln(\text{speed})$ &  & $-$0.180** & $-$0.280** & $-$0.283** \\
 &  & (0.086) & (0.121) & (0.122) \\
$\ln(\text{price})$ &  &  & 0.194 & 0.176 \\
 &  &  & (0.159) & (0.166) \\
$\ln(\text{requests})$ &  &  &  & 0.108 \\
 &  &  &  & (0.275) \\
Constant & $-$0.336** & $-$0.091 & $-$1.744 & $-$1.595 \\
 & (0.169) & (0.204) & (1.373) & (1.427) \\
\hline
Observations & 565 & 565 & 565 & 565 \\
Pseudo $R^2$ & 0.031 & 0.037 & 0.037 & 0.036 \\
\Xhline{2pt}
\end{tabular*}
\vspace{2pt}
\parbox{\linewidth}{\scriptsize\justifying\textit{Notes:} Logit estimates of the probability of switching from $\mathcal{P}_1$ to $\mathcal{P}_2$ in 2014. All covariates are measured in 2013. $H_i = \mathbf{1}[p_i/p_{u,i} > 1/0.35]$. Standard errors in parentheses. *** $p<0.01$, ** $p<0.05$, * $p<0.1$.}
\end{minipage}
\end{table}

Three additional results support the baseline findings against selection concerns.
First, Section~\ref{sec:parallel} shows that treatment effects remain significant for proportional trend violations across the full range $g \in [0,2]$ (Figure~\ref{fig:ab_manski_combined}).
Second, the \ac{dml} estimates in Section~\ref{sec:ML_results} allow covariates to enter nonlinearly and use cross-fitting. If misspecification of the linear control function drove the baseline, flexible ML should absorb the confounding, yet the \ac{dml} estimates are close to the \ac{twfe} baseline and highly significant.
Third, the inverted U-shaped relationship between ineligible consortium share and eligible members' prices (Section~\ref{sec:share_ineligible}) is a qualitative prediction of the cross-subsidization model that selection bias, which would plausibly generate a monotone pattern, cannot easily produce.
This correspondence between theory and data corroborates cross-subsidization as the underlying channel.

In summary, for endogeneity to overturn the findings, an unobserved time-varying confounder would need to:
(i)~predict switching beyond the cost ratio despite adding no explanatory power over observables;
(ii)~survive hospital fixed effects in a balanced panel;
(iii)~correlate with price changes in a way that reproduces the theory's predicted patterns, including the inverted U-shape;
(iv)~persist across both \ac{twfe} and \ac{dml} estimation;
(v)~withstand proportional violations of parallel trends up to $g = 2$;
(vi)~survive the removal of all statistically influential observations (Appendix~\ref{sec:apx_influence});
(vii)~hold when the sample is restricted to the common speed support of \p2c\ (Appendix~\ref{sec:apx_common_support}); and
(viii)~reverse the \citet{oster2019unobservable} stability bounds, which currently show that unobservables would need to work in the opposite direction of observables to nullify the effects (Appendix~\ref{sec:apx_oster}).
Endogeneity may affect precise magnitudes, but the direction and economic significance of the results are robust to any plausible selection story.

\section{Discussion}
\label{sec:discussion}

Our results strongly support the cost-effectiveness hypothesis. 
In the first year of implementing \ac{hcf}, switching from \p1\ to \p2\ caused a 71.31\% reduction in internet price, 71.68\% reduction in the subsidy, and 50.52\% reduction in the \ac{hcp} net cost. 
Back-of-the-envelope calculations suggest that switching all existing \acp{hcp} from \p1\ to \p2\ would save an additional \$33 million annually.

We also found strong evidence for the cross-subsidization hypothesis. 
In contrast to switching from \p1\ to \p2\ that lowered expenditure, switching from \p1\ to \p2c\ \textit{increased} subsidies by 32.86\%. This reversal is difficult to rationalize without cross-subsidization.
The contrast extends to \ac{hcp} net cost: switchers to \p2\ saw a 50.52\% reduction, while switchers to \p2c\ experienced a 469.73\% increase.
It is hard to justify a rational \ac{hcp}'s voluntary switching to a program that is six times costlier to them, unless such decisions are shaped by latent incentives such as cross-subsidization.

The original intention of the \ac{usf} was to ensure that rural \acp{hcp} do not pay more than their urban counterparts on internet plans. 
So, even in the absence of cross-subsidization, \p2c's extension of subsidies to urban \acp{hcp} seems unjustified. 
The fact that majority-rural status is determined by a head count of \acp{hcp} is immensely problematic, as it implies that a large urban general hospital could receive subsidies to cover 65\% of its internet costs if it forms a consortium with two small rural entities, e.g., health clinics.
By extending subsidies to urban \acp{hcp}, \p2c\ has dramatically increased the number of subsidy recipients, imposing a financial strain on the program's budget.
That said, we recommend discontinuing programs \p1\ and \p2c, and migrating all \acp{hcp} to \p2\ as the only subsidy mechanism.
This will remove the inferior design of \p1\ that results in wasteful spending, prevent cross-subsidization, and end spending on urban \acp{hcp}.

Cross-subsidization adds to the literature of subsidy programs that target one group and exclude another. 
For example, \citet{rotemberg2019equilibrium} shows that offering subsidies to eligible recipients could benefit them, while putting their ineligible rivals at a disadvantage, crowding out the latter group by the former. 
In our study, we find a positive spillover effect of the program on ineligible entities. 
Understanding such dynamics is crucial because, when designing interventions, policymakers must consider the direct effect on the recipient and the indirect effect on third parties who could benefit from or be harmed by interventions.

Cross-subsidization involves an artificial cost inflation for eligible members. Existing literature shows a similar price inflation in the High Cost program, where a subsidy item covers the infrastructure investment cost (the high cost loop support). 
The program provides a higher percentage subsidy to firms with a higher investment cost. \citet{berg2011incentives, berg2011universal} found that firms inflated their investment cost to be eligible for a higher subsidy percentage, and the amount of cost inflation rises with the size of the anticipated subsidy increment. 

Given the complex design of the \ac{rhc}, there are many facets of the program that are worth investigating that are beyond the scope of this study. 
First, the program's \$400 million budget cap was not met for the first two decades of the program, and it was common knowledge that this budget cap existed. 
The mere imposition of an unmet budget cap could signal to recipients the availability of unallocated funds and incentivize a more aggressive pursuit of subsidies, leading to wasteful spending and cost overstatement \citep{berg2011incentives}.

Another avenue for future research, that directly builds on our study, is that cross-subsidization would require cooperation (collusion) between an \ac{isp} and a consortium.
The \ac{isp} may demand a fraction of the proceeds, causing price inflation for eligible \acp{hcp} that exceeds the deflation for ineligible ones \citep{ashenfelter2010effect, deneckere1985incentives, viscusi2018economics}.
Cooperative equilibria are more likely to be sustained in repeated games if fewer parties are involved \citep{bourreau2021market, dal2018determinants, sugaya2022repeated}.
In the context of Medicare reimbursements, \citet{clemens2017shadow} show that subsidies induce higher prices, and this effect is intensified in highly concentrated markets.
We conjecture that a higher degree of market concentration among \acp{isp}, \acp{hcp}, or both, may facilitate cooperation and intensify the cost burden of cross-subsidization.

The impact of the \ac{rhc} on internet connectivity and the subsequent impact on rural healthcare is to be examined in future research. 
Existing studies cast doubt on the effectiveness of the other components of \ac{usf}.
The E-Rate program increased internet connectivity in schools by 68\%, but it did not improve student outcomes \citep{goolsbee2006impact, hazlett2019educational}. 
The High Cost program is unjustified for a large fraction of recipients in North Carolina, in the sense that the benefit to the recipients is less than the subsidy expense \citep{boik2017economics}. 
Likewise, most of the Lifeline program's budget was spent on existing users, with little increase in connectivity \citep{wallsten2016learning, ackerberg2014estimating, lyons2023assessing}. 
Three quarters of low-income households did not take up the subsidy offered by the Lifeline program despite being eligible, whereas the program has subsidized many ineligibles \citep{ward2010effect}.

A subsidy program is unjustified unless it makes the intended impact. 
The \ac{rhc} aims to promote \ac{hcp} internet connectivity.
The combined annual revenue of the subsidized \acp{hcp} exceeds \$1 trillion.%
\footnote{In 2023, \ac{us} hospitals \href{https://www.kff.org/health-costs/key-facts-about-hospitals/?entry=overview-introduction}{generated \$1.5 trillion}. 67.4\% of these hospitals are \href{https://www.aha.org/statistics/fast-facts-us-hospitals}{non-profit or public entities}.} 
So, internet subsidies account for less than 0.1\% of the subsidy recipients' overall operating scale, and it is implausible that the average \ac{hcp} would forgo internet connectivity if the subsidy were withdrawn. 
If the program fails to boost connectivity, then it reduces to a wealth transfer from taxpayers to \acp{hcp}.

\section{Conclusions}
\label{sec:conclusions}

This paper provides a comprehensive economic evaluation of the Federal Communications Commission’s \acf{rhc}, analyzing three subsidy mechanisms: 
price-cap (\p1), individual ad valorem subsidy (\p2), and consortium ad valorem subsidy (\p2c). 
We use administrative data and causal econometric models. 
The findings reveal that transitioning from the price-cap to the individual ad valorem mechanism delivers \$33 million in annual cost savings while preserving the program's effectiveness. 

In contrast, the consortium option enables an unintended cross-subsidization scheme that allows ineligible consortium members to use eligible members as a conduit to bypass their ineligibility and receive subsidies: 
Eligible members bear inflated internet prices, mirrored by deflated prices for ineligible members, driving up the consortium's overall subsidy outlays. 
We empirically confirm the presence of cross-subsidization.
The existing enforcement regime, limited to occasional audits and modest fines, lacks the incentives to contain or eliminate it.

The \ac{rhc} program costs have been steadily rising for almost three decades. 
While the concerns with the program's inefficiencies and loopholes are documented (Section \ref{sec:programs}), little has been done to rectify it.
Instead, policymakers have focused on raising more funds and taxing new entities to keep the program afloat \citep{figliola_2025_usf_overview, fcc_2012_usf_contribution_nprm, kelly_et_al_2025}.
A rigorous economic diagnosis has been missing, and our study is an attempt to address this.
We presented actionable recommendations that policymakers may use to amend the program. 
Yet, many questions are to be answered in future studies including the mediating role of market power on the intensity of cross-subsidization, measuring the program's effect on connectivity and the subsequent effect on health outcomes, and the possibility of sophisticated forms of cross-subsidization if internet service providers partake in the scheme. 

\subsection*{CRediT author statement}

\textbf{R.S.D.}, \textbf{M.R.}, and \textbf{R.P.} have been equally involved in the following stages of the study: conceptualization, methodology, software, validation, visualization, formal analysis, investigation, data curation, writing - original draft, and reviewing and editing the manuscript.


\newpage
{\footnotesize
\begin{singlespacing}
\setlength{\bibsep}{3pt}
\bibliography{0_bibliography}
\bibliographystyle{chicago}
\end{singlespacing}
}

\clearpage
\ \\[3.5in]
\begin{center}
    {\Huge Appendix}
\end{center}
    \setcounter{page}{1}
    \setcounter{table}{0}
    \renewcommand\thetable{A.\arabic{table}}

\addcontentsline{toc}{section}{Appendices}

\newpage
\appendix

\setcounter{table}{0}
\renewcommand{\thetable}{A\arabic{table}}

\setcounter{figure}{0}
\renewcommand{\thefigure}{A\arabic{figure}}
\clearpage 
\section{Detailed description of the institutional environment}
\label{Apx:institutional}

\subsection{Internet subsidy programs in the United States}
\label{sec:programs}

The \ac{frc} was the original government agency that regulated radio communication in the \ac{us} from 1927 to 1934.%
\footnote{\href{https://en.wikipedia.org/wiki/Federal_Radio_Commission}{Wikipedia: Federal Radio Commission}} 
The Communications Act of 1934 abolished \ac{frc} and led to the creation of the \ac{fcc}. 
Among other roles, the \ac{fcc} provided subsidies to ensure that \enquote{\textit{all people in the United States shall have access to rapid, efficient, nationwide communications service with adequate facilities at reasonable charges}}.%
\footnote{\href{https://www.fcc.gov/general/universal-service-fund}{FCC: Universal Service Fund}}
The Telecommunications Act of 1996 substantially expanded the subsidy coverage, creating four sister programs, each targeting one recipient group as shown in Figure \ref{fig:schematic}. 
The High Cost Program partially reimburses telecommunication companies that serve high-cost areas; E-Rate subsidizes internet access for schools and libraries; the Lifeline Program assists low-income consumers with their monthly telephone and telecommunication bills; and the \ac{rhc} (the program that we study) subsidizes internet access and related equipment for eligible \acp{hcp}.

The programs are funded and administered under the \ac{usf}. 
The \ac{fcc} has appointed \ac{usac}, a private non-profit entity, to administer \ac{usf}, while the \ac{fcc} retains the right to oversight and assessment. 
\ac{usf} is not dependent on the federal budget. 
Instead, funds are raised by taxing telecommunication companies through the ``contribution factor.''
The contribution factor is the percentage of interstate end-user revenues that telecommunication companies must pay to \ac{usac} to fund \ac{usf}. A complex system of accounting determines the contribution factor that each telecommunication company must pay.%
\footnote{Some explanations may be found on the \href{https://www.fcc.gov/general/contribution-factor-quarterly-filings-universal-service-fund-usf-management-support}{FCC website}.} 
\ac{usac} sets the contribution factor on a quarterly basis to ensure that program costs remain covered. 
Since 1997, \ac{usac} has set the contribution factor over 100 times.%
\footnote{\href{https://www.fcc.gov/general/contribution-factor-quarterly-filings-universal-service-fund-usf-management-support}{Contribution Factor \& Quarterly Filings - USF Management Support}.}
While the \ac{fcc} has the authority to appeal or adjust the contribution factor, it has never done so \citep{dunstan2023fcc}.
Therefore, \ac{usac} is the de facto setter and enforcer of the contribution factor.

In response to the ever-increasing program expenditure, \ac{usac} has been steadily raising the contribution factor.
Telecommunication companies are allowed to shift the contribution factor to consumers, itemized separately under ``Universal Service Charge'', ``USF fee'', or similar names.%
\footnote{\href{https://www.govinfo.gov/content/pkg/CFR-2020-title47-vol3/pdf/CFR-2020-title47-vol3-sec54-712.pdf}{47 CFR Ch. I § 54.712}, \href{https://www.usac.org/about/universal-service/}{USAC: Universal Service}}
Therefore, the contribution factor is a tax on ordinary internet users.
The Government Accountability Office has \Mquote{raised concerns about the cost burden the fund imposes on consumers} \citep{gao_2012_high_cost_program}. 

Figure \ref{fig:be_contribution_factor} shows the trend of the contribution factor since the program's inception. 
It began at 3.14\% in 1998 \citep{FCC1998FR_DA98_1130}. 
Within two years, it rose to 5.7\% \citep{FCC2000DA00517}. 
When the contribution factor reached 16.1\% in 2014, an \ac{fcc} commissioner called it a \Mquote{clearly disturbing and unsustainable} path.%
\footnote{\href{https://www.fcc.gov/news-events/blog/2014/09/11/usf-contribution-factor-over-time}{USF Contribution Factor Over Time}.}
This concern was echoed in the Senate in 2023 when the contribution factor reached 29\%, where Senator John Thune emphasized the need to \Mquote{address the inefficiencies within} the program, and raised the concern that the program is riddled with waste, fraud, and abuse \citep[p.3-4]{senate_universal_service1}.
As of the first quarter of 2026, the contribution factor is 37.6\%.%
\footnote{\href{https://docs.fcc.gov/public/attachments/DA-25-223A1.pdf}{FCC:DA-25-223A1}.}

\begin{figure}[htbp]
    \centering
    \begin{minipage}{\linewidth}
    \centering
    \includegraphics[width=.8\linewidth]{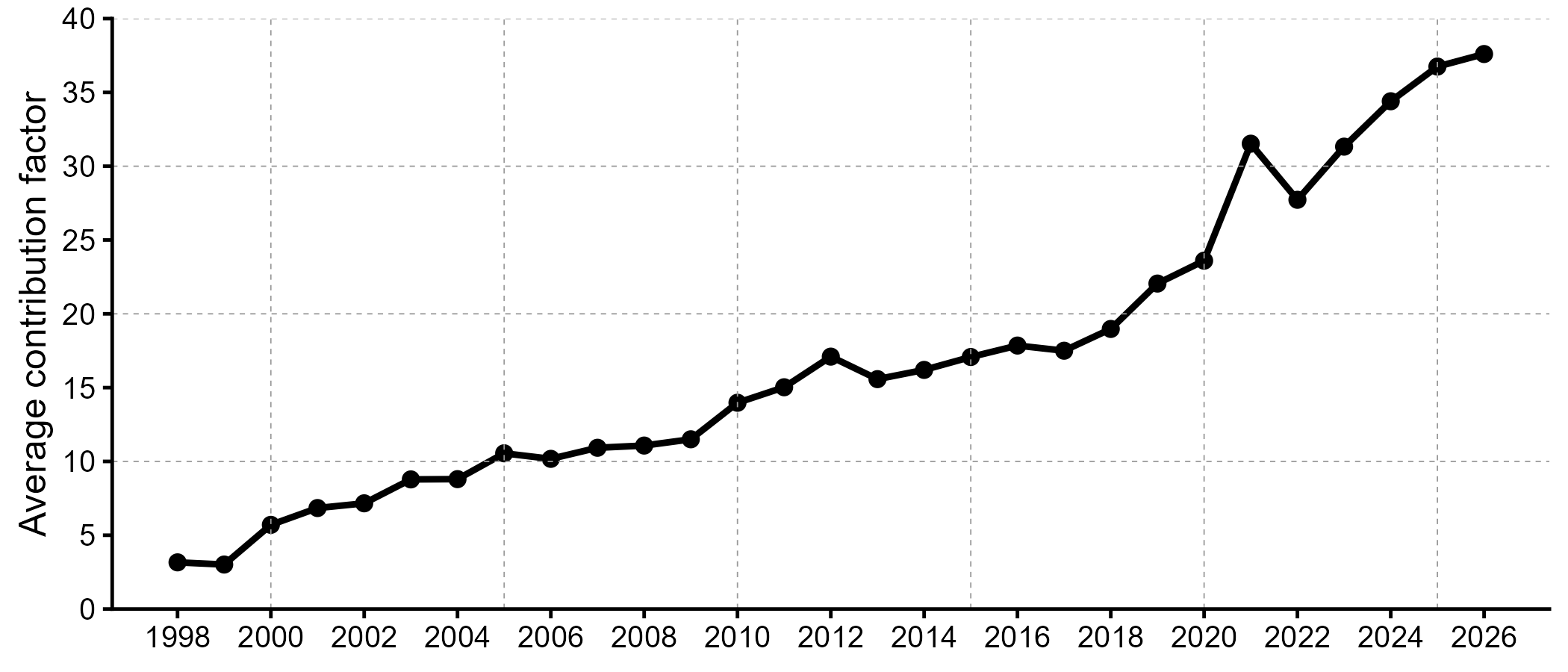}
    \caption{The evolution of the contribution factor.}
    \label{fig:be_contribution_factor}
    \vspace{2pt}
    \parbox{\linewidth}{\scriptsize\justifying\textit{Notes:} The figure illustrates the evolution of the contribution factor over its lifetime. The contribution factor is set on a quarterly basis. The figure reports the average of the quarterly values to report a single observation per year. The contribution factor in percentages is reported on the vertical axis.}
    \end{minipage}
\end{figure}

Below is a back-of-the-envelope calculation of the effect of the program on internet prices in the \ac{us}. 
The total \ac{isp} revenues in the \ac{us} in 2025 are estimated to be \href{https://www.cognitivemarketresearch.com/internet-service-market-report}{\$127.5 billion}. 
The \ac{usf} budget has been \$8.6 billion in 2024 \citep{usf_annual_report_2024}.
Knowing that internet prices must be raised to fund \ac{usf}, internet prices are likely inflated by $8.6/(127.5 - 8.6)\times100=7.2\%$ via the contribution factor.
The price elasticity of demand for residential internet is -0.18 to -0.47 \citep{czajkowski2024assessing}. 
Using this, the contribution factor likely priced out 1.3\% to 3.4\% of users, who could become reliant on subsidies to stay connected.
This creates a positive feedback loop: rising program costs increase the contribution factor, which raises internet prices, exacerbates affordability, increases internet users' subsidy reliance, and drives up subsidy expenditure \citep{hazlett2019educational}. 

It is understood in public choice theory that bureaucratic entities may seek to maximize the budget under their management as a means of political influence and financial power \citep{niskanen1968peculiar}. Besides, \ac{usac}'s own expenses are paid using the contribution factor \citep{gao_usac_admin_costs}, which amounted to \$324 million in 2024 \citep{usf_annual_report_2024}. Therefore, the relationship between the \ac{fcc} and \ac{usac} may create a classical principal-agent problem, where cost-containment is not aligned with, or potentially contradicts, \ac{usac}'s priorities.
In a review of the High Cost program's 2008 spending, the \ac{fcc} Office of Inspector General estimated that 23.3\% of the subsidies were paid out ``erroneously'' \citep[p.9]{wallsten2011universal}, shedding light on the possible size of waste, fraud, and abuse. 
Despite the \ac{fcc}'s authority and central role, it has admitted to falling short in auditing cost information of entities under its oversight \citep{zolnierek2008loop}.

\subsection{The Rural Health Care program}
\label{sec:rhc_program}

\ac{rhc}'s original goal in 1997 was to ensure that rural \acp{hcp} have access to \enquote{telecommunication services necessary for the provision of health care services} \citep[p.4]{gilroy2008universal}, and that rural \acp{hcp} pay no more than their urban counterparts on internet bills.%
\footnote{\href{https://www.fcc.gov/guides/universal-service-program-rural-health-care-providers}{The FCC’s Universal Service Rural Health Care Programs}}
This led to the creation of the first subsidy mechanism under the \ac{rhc}, called the Telecommunications Program (\p1)\ in 1997. 
\p1\ uses the following flat-rate subsidy reimbursement scheme: 
If an eligible rural \ac{hcp} who pays internet price $p$ could show that a similar plan costs $p_u$ in a nearby urban location, the \ac{hcp} would pay $p_u$ out of its own pocket (\ac{hcp} net cost), and the subsidy would cover $p-p_u$. 

This design is economically problematic.
If an \ac{hcp} chooses an unjustifiably expensive internet plan, \p1\ fully compensates the \ac{hcp} for the excessive payment. In contrast, if an \ac{hcp} switches to a more cost-effective internet plan and minimizes waste, \p1\ fully captures the cost savings.
This full cost insulation creates a perfectly inelastic demand for internet services, eliminating any cost-saving incentives \citep{rabbani_bb2}.
 
The \acf{hcf} was introduced in 2012, piloted in 2013, and fully implemented in 2014 to run alongside \p1. 
Instead of benchmarking on urban rates, \ac{hcf} reimburses \acp{hcp} for 65\% of their rural internet bills while requiring the \acp{hcp} to pay the remaining 35\%. 
This proportional cost-sharing incentivizes waste reduction because it rewards \acp{hcp} \$0.35 for every \$1.00 they save on their internet bills, and penalizes them \$0.35 for every \$1.00 they waste. 
Although none of the official \ac{fcc} documents discloses the reasons for introducing \ac{hcf} to run alongside \p1, the \ac{fcc} stated that they \enquote{expect} a gradual migration of \p1\ users to \ac{hcf}%
\footnote{\href{https://rb.gy/5fkul}{Healthcare Connect Fund - Frequently Asked Questions}}, hinting that the introduction of \ac{hcf} was an attempt to phase out \p1's suboptimal incentive design. 

\ac{hcf} allows \acp{hcp} to either individually request subsidies (\p2) or to create a consortium of \acp{hcp} that requests subsidies on behalf of its members (\p2c). 
\p2c\ extends subsidy coverage to eligible urban \acp{hcp} that belong to a majority-rural consortium. 
A consortium is deemed majority rural if more than half of its members are rural. 
It is based on a simple count of \acp{hcp} and not size or market capitalization.
For example, if a large urban general hospital and two small rural clinics form a consortium, it is considered majority rural because two of the three members are rural, and all three members would be subsidized.

\p1, \p2, and \p2c\ use a common trio of eligibility criteria:
Subsidy recipients must 
(1) be a non-profit or public entity, 
(2) be in a rural area or belong to a majority-rural consortium, and 
(3) belong to one of the following categories: 
non-profit hospitals; 
skilled nursing facilities; 
community mental health centers; 
rural health clinics; 
dedicated emergency departments of rural for-profit hospitals; 
community health centers or health centers providing health care to migrants; 
local health departments or agencies; 
post-secondary educational institutions offering health care instruction such as teaching hospitals or medical schools; 
part-time eligible entities located in a facility that is ineligible; 
or consortia of health care providers consisting of one or more entities listed above.%
\footnote{\url{https://www.usac.org/rural-health-care/telecommunications-program/step-1-determine-eligibility-of-your-site/}}
Once eligibility is established, \acp{hcp} can choose either program as their subsidy channel, except for eligible urban \acp{hcp} who must remain in \p2c. 
\acp{hcp} can switch programs at the time of subsidy renewal.

Besides eligible rural and eligible urban, \acp{hcp} could be ineligible rural or ineligible urban.
Ineligible \acp{hcp} do not count in the calculation that determines the majority-rural status of a consortium.
When a consortium submits a subsidy request, it enters a competitive bidding process, in which \acp{isp} compete by placing bids for providing all the internet services listed by the consortium for all its members.
As explained in Section \ref{sec:intro}, we hypothesize that this could give rise to an unintended cross-subsidization scheme, where eligible \acp{hcp} covertly extend subsidy coverage to ineligible members in the same consortium. 
This could be done by artificially lowering the rates for ineligible members while raising the rates for eligible members in a manner that leaves the \ac{isp} revenue intact. 
For every \$1 that is cross-subsidized, the consortium would collect an additional \$0.65 in subsidies.

\begin{table}[htbp]
\centering
\caption{\ac{hcp} entity type composition by program and year, 2013--2014.}
\label{tab:bc_HCP_type}
\begin{minipage}{\columnwidth}
\centering
\small
\begin{tabular*}{\linewidth}{@{\extracolsep{\fill}} l r rrr}
\Xhline{2pt}
 & \multicolumn{1}{c}{2013} & \multicolumn{3}{c}{2014} \\
\cmidrule(lr){2-2} \cmidrule(lr){3-5}
HCP type & \p1{} & \p1{} & \p2{} & \p2c{} \\ \hline
Rural health clinic & 262 & 136 & 96 & 21 \\
Non-profit hospital & 398 & 144 & 150 & 7 \\
Mental health center & 84 & 49 & 18 & 0 \\
Community health center & 96 & 67 & 24 & 1 \\
Local health department & 115 & 107 & 2 & 0 \\
Skilled nursing facility & 0 & 0 & 0 & 0 \\
Medical school & 3 & 1 & 2 & 0 \\
Emergency room & 3 & 0 & 2 & 0 \\
Not available & 2 & 2 & 0 & 0 \\
\Xhline{2pt}
\end{tabular*}
\vspace{2pt}
\parbox{\linewidth}{\scriptsize\justifying\textit{Notes:} The table reports the number of \ac{hcp}s of each entity type participating in each program in 2013 and 2014. \p1{} is the price-cap program, \p2{} is the individual ad valorem program, and \p2c{} is the consortium ad valorem program. The sample is restricted to the 970 baseline \acp{hcp}. \ac{hcp}s with mixed program participation within a year are excluded.}
\end{minipage}
\end{table}

Table \ref{tab:bc_HCP_type} shows the composition of \acp{hcp} that received subsidies under each program during 2013--2014.
Non-profit hospitals make up the largest group, followed by rural health clinics, local health departments, community health centers, and mental health centers.

\subsection{Institutional underpinnings of cross-subsidization}
\label{sec:cross_sub}

This section describes the institutional settings that may enable cross-subsidization.
Any group of \acp{hcp} could consolidate all its subsidy requests into a single consortium application. 
The application must enter what resembles an auction process that allows \acp{isp} to place bids to provide services to the consortium. 
The process begins with \acp{hcp} filing Form 465 to \ac{usac} to officially declare and detail their internet connectivity needs (bandwidth, locations, technology, etc.) in a \ac{rfp}. 
\ac{usac} puts \acp{rfp} on its website for 28 days. 
This is called the \enquote{competitive bidding} process, which is the \acp{isp}' opportunity to place bids. 
After this window ends, the consortium can review the bids and select the winner. 
The consortium may sign a contract with the winner as early as on day 29, without \ac{usac}'s approval, so long as the HCP can show that it has chosen the \enquote{most cost-effective} bid and document its evaluation criteria \citep{usac_guide_461_462}.

\begin{table}[htbp]
\centering
\caption{\ac{rfp} criteria examples.}
\label{tab:rfp_weights}
\begin{minipage}{\columnwidth}
\centering
\small
\begin{tabular*}{\linewidth}{@{\extracolsep{\fill}} llc}
\Xhline{2pt}
                              & Criteria                                          & Weight (\%) \\ \hline
\multirow{8}{*}{Consortium 1} & Cost                                              & 16          \\
                              & Overall network solution                           & 14          \\
                              & 24 hour support team                              & 13          \\
                              & Prior experience including past performance       & 12          \\
                              & Prices for ineligible services, products and fees & 14          \\
                              & Provisioning and installation time            & 11          \\
                              & Leveraging existing technology                    & 10          \\
                              & Reliability of equipment (new v. refurbished)     & 10          \\ \hline
\multirow{2}{*}{Consortium 2} & Cost                                              & 50          \\
                              & Prior experience including past performance       & 50          \\ \hline
\multirow{5}{*}{Consortium 3} & Prior experience, including past performance      & 20          \\
                              & Reliability of Service                            & 30          \\
                              & Other (solicitation compliance)                   & 10          \\
                              & Cost                                              & 30          \\
                              & Project mnagement plan                           & 10          \\ \hline
\multirow{6}{*}{Consortium 4} & Cost                                              & 20          \\
                              & Reliability of service                            & 20          \\
                              & Other (consideration of early termination fees)   & 15          \\
                              & Leverage existing resources                       & 15          \\
                              & Prior experience, including past performance      & 15          \\
                              & Technical support                                 & 15          \\ \Xhline{2pt}
\end{tabular*}
\vspace{2pt}
\parbox{\linewidth}{\scriptsize\justifying\textit{Notes:} The table shows actual criteria that four different consortia have reported to use in selecting the winning \ac{isp} in the competitive bidding process. The table reports the criteria and weights. While the consortium names have been masked, all four consortia are subsidy recipients.}
\end{minipage}
\end{table}

Notably, the auction winner is not the \ac{isp} that places the lowest bid. 
Instead, the consortium is at liberty to assign any weights to any objective or subjective criteria that it deems to be important \citep[p.3]{fcc_form_461}. 
Table \ref{tab:rfp_weights} shows four actual instances of criteria and weights that were used in \acp{rfp}. 
This allows consortia to engineer the criteria and weights to make a particular \ac{isp} win, even if the \ac{isp}'s bid is objectively uncompetitive. 

Even the legal language that is meant to prevent abuse seems vague and unenforceable: 
The \ac{rfp} submission system requires the consortium representative to certify under penalty of perjury that they have selected the 
\Mquote{method that costs the least after consideration of the features, quality of transmission, reliability, and other factors that the health care provider deems relevant to choosing a method of providing the required health care services.} \citep[p.2]{fcc_form_462}.

This process does not have the economic building blocks of an auction. 
An auction has quantitative metrics in place that are public information before the auction begins, so that a given set of bids results in a deterministically or probabilistically predictable outcome, a deviation from which could be appealed. 
In the case of \ac{rhc}, the process is based on subjective and unquantifiable measures that the consortium sets and enforces, and its subjectivity makes it unobjectionable. 

Instead, this process resembles a procurement mechanism with pure buyer discretion at the cost of a third party.
This unusual combination (the consortium sets the auction criteria, selects the winner, and makes taxpayers pay the price) gives the consortium the utmost negotiating leverage. 
Therefore, the \ac{isp}'s best response may not entail competitive bidding. 
Instead, an \ac{isp}'s winning strategy may be to surpass other bidders in bending the terms of the trade in favor of the consortium. 
Of course, there is a limit to how much the terms of the trade could be distorted before it becomes unattractive to the bidding \ac{isp}. 

We hypothesize that this process allows for devising a scheme in which a surplus is transferred from taxpayers to the consortium without cannibalizing the \ac{isp}'s revenue.
Each consortium may have members (\acp{hcp}) that are eligible to receive subsidies, and members that are ineligible.
Eligibility is determined before competitive bidding begins \citep[p.5]{fcc_form_460}.
So, when bids are being placed and negotiations are ongoing, both parties (consortium and \ac{isp}) know which lines of internet connection belong to eligible members and which lines belong to ineligible members. 
Both parties know that every \$1 of cross-subsidization transfers \$0.65 from taxpayers to the consortium in a way that is revenue-neutral to the \ac{isp}. 

Given that each \ac{rfp} has many details that may need clarification, consortia and \acp{isp} are allowed to engage in one-on-one discussions to clarify needs and expectations before submitting bids. 
Moreover, \ac{usac} allows \ac{isp} who bid in some auctions to act as a consortium representative or consultant in other auctions, provided the \ac{isp} maintains sufficient organizational separation between staff involved in bidding and those engaged in consortium consulting activities \citep[p.10]{fcc_form_460}. 
This may facilitate collaboration between consortia and \acp{isp}, weaken the competitive nature of the bidding process, and induce cartel-like behaviors. 

If a consortium and \ac{isp} deem the bidding process as one stage of a repeated game, a cooperative equilibrium may be maintained in the long run.
In this setting, collusion is lucrative and likely. 
In 2023, an \ac{isp} named GCI Communication Corp agreed to pay \$42.1 million to settle a fraud allegation case where it was accused of violating the competitive bidding regulations and colluding with an \ac{hcp} to inflate prices and knowingly receive excessive subsidy payments during 2015-2018 (FCC order \href{https://docs.fcc.gov/public/attachments/DA-23-380A1.pdf}{DA 23-380}). 
The data shows that, during 2012-2023, GCI was the service provider to 4275 lines of subsidy, and it has received \$1,237,016,437 in subsidy payments.
So, the \$42.1 million settlement claws back only 3.4\% of the total subsidies GCI collected since 2012. 
This creates a clear economic incentive to engage in fraudulent behavior: the potential gains from inflating claims far exceed the financial risk posed by penalty and enforcement. 
In effect, the penalty functions as a minor operational expense rather than a meaningful deterrent, rendering the expected payoff of collusion overwhelmingly positive.

\section{Model proofs and derivations}
\label{sec:appx_model}

\subsection{Proof of Proposition~\ref{prop:cap}}
\label{sec:appx_cap}

When the cap binds ($p\ge\bar{p}$), the consumer price is pinned at $\bar{p}$ and demand is $D(\bar{p})$. The firm's objective~\eqref{eq:cap_problem} is
\[
\pi(p)=(p-c)\,D(\bar{p})-\alpha\,\Phi(p-\bar{p}).
\]
The first-order condition with respect to $p$ is
\[
D(\bar{p})-\alpha\,\Phi'(p-\bar{p})=0.
\]
Since $\Phi''>0$ and $\Phi'(0)=0$, the marginal penalty $\Phi'$ is strictly increasing from zero. Hence the equation $\Phi'(p-\bar{p})=D(\bar{p})/\alpha$ has a unique solution
\[
p^{cap}-\bar{p}=(\Phi')^{-1}\!\!\left(\frac{D(\bar{p})}{\alpha}\right),
\]
establishing~\eqref{eq:pcap}. Because $(\Phi')^{-1}$ is increasing, $p^{cap}$ is increasing in $D(\bar{p})$ and decreasing in $\alpha$. The second-order condition $-\alpha\,\Phi''(p-\bar{p})<0$ holds by $\Phi''>0$, confirming a maximum.

Under the quadratic penalty $\Phi(\delta)=\tfrac{\gamma}{2}\delta^2$, $\Phi'(\delta)=\gamma\delta$, so $p^{cap}-\bar{p}=D(\bar{p})/(\alpha\gamma)$.
Government outlays are
\[
G^{cap}=(p^{cap}-\bar{p})\,D(\bar{p})=\frac{D(\bar{p})}{\alpha\gamma}\cdot D(\bar{p})=\frac{D(\bar{p})^2}{\alpha\gamma}.
\tag*{$\square$}
\]

\subsection{Proof of Proposition~\ref{prop:cap_to_av}}
\label{sec:appx_av}

\paragraph{Part~(i).}
By the effective-marginal-cost representation~\eqref{eq:pcav_emc},
$p_c^{adv}(\tau)=p^{no}\!\bigl(c(1-\tau)\bigr)$.
Since $dp^{no}/d\tilde{c}>0$, the map
$\tau\mapsto p_c^{adv}(\tau)$ is strictly decreasing.
The critical rate $\tau^*$ is defined by
$p_c^{adv}(\tau^*)=\bar{p}$~\eqref{eq:cond_barp}.
For $\tau\ge\tau^*$, monotonicity gives
$p_c^{adv}(\tau)\le p_c^{adv}(\tau^*)=\bar{p}$, with strict
inequality when $\tau>\tau^*$.
Since the cap binds, $p_c^{cap}=\bar{p}$.

\paragraph{Part~(ii).}
Since $D'<0$ and $p_c^{adv}(\tau)\le\bar{p}=p_c^{cap}$ by Part~(i),
\[
Q^{adv}(\tau)=D\!\bigl(p_c^{adv}(\tau)\bigr)\ge D(\bar{p})=Q^{cap}.
\]

\paragraph{Part~(iii).}
Define $E(p)\equiv p\,D(p)$.
Then
\begin{equation}
E'(p)=D(p)+p\,D'(p)=D(p)\bigl(1-\varepsilon_D(p)\bigr),
\label{eq:appx_Eprime}
\end{equation}
where $\varepsilon_D(p)\equiv -pD'(p)/D(p)$ is the price elasticity of demand.
If $\varepsilon_D(p)\le 1$ on $[p_c^{adv}(\tau),\,\bar{p}]$, then $E'(p)\ge 0$ on that interval, so $E$ is nondecreasing.
Since $p_c^{adv}(\tau)\le\bar{p}$,
\[
E^{adv}(\tau)=E\!\bigl(p_c^{adv}(\tau)\bigr)\le E(\bar{p})=E^{cap}.
\]

\paragraph{Part~(iv).}
Under the quadratic penalty $\Phi(\delta)=\tfrac{\gamma}{2}\delta^2$,
Proposition~\ref{prop:cap} gives
$G^{cap}=D(\bar{p})^2/(\alpha\gamma)$.
For any fixed $\tau\in(0,1)$, the ad valorem equilibrium price
$p^{adv}(\tau)$ is finite, since it solves the smooth first-order
condition~\eqref{eq:foc_av} with bounded primitives. Therefore,
$G^{adv}(\tau)=\tau\,p^{adv}(\tau)\,Q^{adv}(\tau)$
is finite and independent of $(\alpha,\gamma)$.
As $\alpha\gamma\to 0$, $G^{cap}\to\infty$ while $G^{adv}(\tau)$
remains bounded.
Hence $G^{adv}(\tau)<G^{cap}$ for sufficiently small $\alpha\gamma$.
\hfill$\square$

\subsection{Proof of Proposition~\ref{prop:hump}}
\label{sec:appx_hump}

From~\eqref{eq:CostMini2}, the consortium's reduced-form objective is
\[
\Psi(\kappa)=B\,\kappa-\alpha\,\Phi\!\bigl(B(\kappa-1);\,R\bigr),
\qquad \kappa\in[1,\,1+R].
\]

\paragraph{Part~(i).}
Differentiating,
$\Psi'(\kappa)=B-\alpha\,B\,\Phi'\!\bigl(B(\kappa-1);\,R\bigr)
=B\bigl[1-\alpha\,\Phi'\!\bigl(B(\kappa-1);\,R\bigr)\bigr]$.
Evaluating at $\kappa=1$:
\[
\Psi'(1)=B\bigl[1-\alpha\,\Phi'(0;\,R)\bigr]=B>0,
\]
where the last equality uses condition~(ii), $\Phi'(0;\,R)=0$.
Since $\Psi'(1)>0$, a small increase in $\kappa$ above~$1$ raises the consortium's payoff, so $\kappa^*>1$ and $\tilde{p}_E>p_E^{adv}$.

\paragraph{Part~(ii).}
The second derivative is
\[
\Psi''(\kappa)=-\alpha\,B^2\,\Phi''\!\bigl(B(\kappa-1);\,R\bigr)<0
\qquad\text{for all }\kappa\ge 1,
\]
by condition~(iv), $\Phi''>0$.
Hence $\Psi$ is strictly concave on $[1,\,1+R]$, and the maximizer $\kappa^*$ is unique.

\paragraph{Part~(iii).}
By part~(i), $\Psi'(1)>0$, so $\kappa^*>1$.
At the upper boundary,
$\Psi'(1+R)=B\bigl[1-\alpha\,\Phi'(BR;\,R)\bigr]$.
If $\alpha\,\Phi'(BR;\,R)>1$, then $\Psi'(1+R)<0$.
Since $\Psi'$ is continuous, strictly decreasing (by $\Psi''<0$), positive at $\kappa=1$, and negative at $\kappa=1+R$, the intermediate value theorem yields a unique interior $\kappa^*\in(1,\,1+R)$ satisfying $\Psi'(\kappa^*)=0$, which gives~\eqref{eq:foc_kappa}.
If $\alpha\,\Phi'(BR;\,R)\le 1$, then $\Psi'(\kappa)\ge 0$ for all $\kappa\in[1,\,1+R]$, so $\Psi$ is nondecreasing on the entire domain and $\kappa^*=1+R$.

\paragraph{Part~(iv).}
Under $\Phi(\delta;\,R)=\tfrac{\gamma R}{2}\,\delta^2$, verify:
(i)~$\Phi(0;\,R)=0$;
(ii)~$\Phi'(0;\,R)=\gamma R\cdot 0=0$;
(iii)~$\Phi'(\delta;\,R)=\gamma R\,\delta>0$ for $\delta>0$;
(iv)~$\Phi''(\delta;\,R)=\gamma R>0$;
(v)~$\partial^2\Phi/\partial\delta\,\partial R=\gamma>0$.
The first-order condition~\eqref{eq:foc_kappa} becomes
$\alpha\,\gamma R\,B(\kappa^*-1)=1$, yielding
$\kappa^*-1=1/(\alpha\gamma B\,R)$.
This is feasible ($\kappa^*-1\le R$) if and only if
$1/(\alpha\gamma B\,R)\le R$, i.e., $R\ge R^*\equiv 1/\sqrt{\alpha\gamma B}$.
When $R<R^*$, the interior solution exceeds the upper bound, so the constraint binds and $\kappa^*=1+R$.
Therefore
\[
\kappa^*(R)=\min\!\left\{1+R,\;1+\frac{1}{\alpha\gamma B\,R}\right\}.
\]
For $R\le R^*$: $\kappa^*(R)=1+R$, strictly increasing.
For $R\ge R^*$: $\kappa^*(R)=1+1/(\alpha\gamma B\,R)$, strictly decreasing.
Both branches meet at $R=R^*$ with $\kappa^*(R^*)=1+R^*$.
As $R\to 0$, $\kappa^*\to 1$; as $R\to\infty$, $\kappa^*\to 1$.
Hence $\kappa^*$ is hump-shaped, attaining its unique maximum at $R^*$.
\hfill$\square$

\subsection{Closed-form solutions under linear demand}
\label{sec:appx_linear}

Throughout this subsection, specialize to linear demand $D(p_c)=a-bp_c$ with $a,b>0$ and $a>bc$ (to ensure positive equilibrium quantities), and use the quadratic penalty $\Phi(\delta)=\tfrac{\gamma}{2}\delta^2$ ($\gamma>0$).

\subsubsection*{No-subsidy monopoly ($M0$)}

The monopolist solves $\max_{p\ge 0}\,(p-c)(a-bp)$.
The first-order condition $a-bp-(p-c)b=0$ yields
\begin{equation}
p^{no}=\frac{1}{2}\!\left(\frac{a}{b}+c\right),
\qquad
Q^{no}=\frac{1}{2}(a-bc).
\label{eq:appx_no}
\end{equation}
Profits, consumer expenditures, and government outlays are
\[
\pi^{no}=\frac{(a-bc)^2}{4b},
\qquad
E^{no}=p^{no}\,Q^{no}=\frac{(a+bc)(a-bc)}{4b},
\qquad
G^{no}=0.
\]

\subsubsection*{Price-cap regime}

Under a binding cap ($p\ge\bar{p}$), demand is $D(\bar{p})=a-b\bar{p}$.
Proposition~\ref{prop:cap} with $\Phi'(\delta)=\gamma\delta$ gives
\begin{equation}
p^{cap}=\bar{p}+\frac{a-b\bar{p}}{\alpha\gamma},
\qquad
G^{cap}=\frac{(a-b\bar{p})^2}{\alpha\gamma}.
\label{eq:appx_cap}
\end{equation}
The cap binds whenever the unconstrained monopoly price exceeds $\bar{p}$, i.e., $p^{no}>\bar{p}$, which requires $\bar{p}<\tfrac{1}{2}(a/b+c)$.

\subsubsection*{Ad valorem regime}

The firm maximizes $(p-c)(a-b(1-\tau)p)$.
The first-order condition yields
\begin{equation}
p^{adv}(\tau)=\frac{1}{2}\!\left(\frac{a}{b(1-\tau)}+c\right).
\label{eq:appx_pAV}
\end{equation}
The consumer price, quantity, and government outlays are
\begin{align}
p_c^{adv}(\tau) &=(1-\tau)\,p^{adv}(\tau)
=\frac{1}{2}\!\left(\frac{a}{b}+c(1-\tau)\right),
\label{eq:appx_pcAV}\\[4pt]
Q^{adv}(\tau) &=a-b\,p_c^{adv}(\tau)
=\frac{1}{2}\bigl(a-bc(1-\tau)\bigr),
\label{eq:appx_QAV}\\[4pt]
G^{adv}(\tau) &=\tau\,p^{adv}(\tau)\,Q^{adv}(\tau).
\label{eq:appx_GAV}
\end{align}

\subsubsection*{Sufficient condition for ad valorem dominance}

Condition~\eqref{eq:cond_barp} requires $p_c^{adv}(\tau)\le\bar{p}$.
Substituting~\eqref{eq:appx_pcAV}:
\begin{equation}
\frac{1}{2}\!\left(\frac{a}{b}+c(1-\tau)\right)\le\bar{p}
\quad\Longleftrightarrow\quad
\tau\ge 1-\frac{2\bar{p}-a/b}{c},
\label{eq:appx_tau_condition}
\end{equation}
provided the right-hand side lies in $(0,1)$.

\subsubsection*{Fiscal comparison under linear demand}

From~\eqref{eq:appx_cap}, $G^{cap}=(a-b\bar{p})^2/(\alpha\gamma)$.
As $\alpha\gamma\to 0$, $G^{cap}\to\infty$.
For any fixed $\tau\in(0,1)$, $G^{adv}(\tau)$ is bounded, confirming Proposition~\ref{prop:cap_to_av}(iv).

To illustrate the comparison at finite $\alpha\gamma$, note that $G^{adv}(\tau)<G^{cap}$ whenever
\[
\tau\,p^{adv}(\tau)\,Q^{adv}(\tau)<\frac{(a-b\bar{p})^2}{\alpha\gamma}.
\]
Using~\eqref{eq:appx_pAV}--\eqref{eq:appx_QAV}, the left-hand side is a function of $(\tau,a,b,c)$ that does not depend on $\alpha\gamma$, while the right-hand side is unbounded as $\alpha\gamma\to 0$.
Hence for any parameter configuration satisfying $a>bc$ and condition~\eqref{eq:appx_tau_condition}, there exists a threshold $\overline{\alpha\gamma}>0$ below which $G^{adv}(\tau)<G^{cap}$.

\subsubsection*{Consortium under linear demand}

With linear demand $D_j(p)=a_j-b_j p$ for each member $j\in\{E,I\}$, the monopoly equilibrium prices are
\[
p_E^{adv}=\frac{1}{2}\!\left(\frac{a_E}{b_E(1-\tau)}+c\right),
\qquad
p_I^{adv}=\frac{1}{2}\!\left(\frac{a_I}{b_I}+c\right).
\]
The reduced-form objects defined in~\eqref{eq:consortium_notation} become
\[
B=\tau\,p_E^{adv}\bigl(a_E-b_E(1-\tau)p_E^{adv}\bigr),
\qquad
R=\frac{p_I^{adv}(a_I-b_I p_I^{adv})}{p_E^{adv}(a_E-b_E(1-\tau)p_E^{adv})}.
\]
Proposition~\ref{prop:hump} then delivers the optimal distortion
\[
\kappa^*(R)=\min\!\left\{1+R,\;1+\frac{1}{\alpha\gamma B\,R}\right\},
\]
with peak distortion at $R^*=1/\sqrt{\alpha\gamma B}$.

For the symmetric-demand case $D_E=D_I=D$ with $a_E=a_I=a$, $b_E=b_I=b$:
\[
p_E^{adv}=\frac{1}{2}\!\left(\frac{a}{b(1-\tau)}+c\right),
\qquad
p_I^{adv}=\frac{1}{2}\!\left(\frac{a}{b}+c\right)=p^{no}.
\]
The revenue ratio simplifies to
\[
R=\frac{p^{no}\,D(p^{no})}{p_E^{adv}\,D\!\bigl((1-\tau)p_E^{adv}\bigr)},
\]
which is a function of $\tau$ alone (for given $a,b,c$).
As $\tau\to 0$, $p_E^{adv}\to p^{no}$ and $R\to 1$, while the subsidy base $B\to 0$, so $R^*\to\infty$ and $\kappa^*\to 1+R$: the constraint binds but the subsidy base is negligible.
As $\tau\to 1$, $p_E^{adv}\to\infty$ and $B$ grows, so $R^*\to 0$: enforcement dominates and $\kappa^*\to 1$.

\clearpage 
\section{Double/Debiased machine learning implementation}
\label{sec:machine_learning}

This section details the implementation of the machine learning method. 
We estimate the following \ac{plr} model:
\begin{equation}
\label{eq:PLR}
\bar{Y}_{jt} = \bm{S}_{jt}'\,\bm{\theta}_0 + g_0(\bm{X}_{jt}) + \varepsilon_{jt}, \qquad \mathbb{E}[\varepsilon_{jt} \mid \bm{X}_{jt}, \bm{S}_{jt}] = 0,
\end{equation}

\noindent where \(\bm{\theta}_0\) is the parameter vector of interest, \(g_0(\bm{X}_{jt})\) is an unknown confounding function, and \(\bm{X}_{jt} \in \mathbb{R}^{p}\) is a vector of covariates. The nuisance functions are:
\begin{equation}
\label{eq:nuisances}
\ell_0(\bm{x}) = \mathbb{E}[\bar{Y}_{jt} \mid \bm{X}_{jt} = \bm{x}], \qquad \bm{m}_0(\bm{x}) = \mathbb{E}[\bm{S}_{jt} \mid \bm{X}_{jt} = \bm{x}] \in \mathbb{R}^{N_d},
\end{equation}
where \(\bm{X}_{jt}\) denotes the conditioning variables.
The \ac{dml} estimator leverages the Neyman-orthogonal score for the \ac{plr} model:
$\psi(\bm{W}_{jt}; \bm{\theta}, \eta) = \big(\bm{S}_{jt} - \bm{m}(\bm{X}_{jt})\big) \Big(\bar{Y}_{jt} - \ell(\bm{X}_{jt}) - \big(\bm{S}_{jt} - \bm{m}(\bm{X}_{jt})\big)'\bm{\theta}\Big),$
where \(\bm{W}_{jt} = (\bar{Y}_{jt}, \bm{S}_{jt}, \bm{X}_{jt})\), \(\eta = (\ell, \bm{m})\), and \(\mathbb{E}[\psi(\bm{W}_{jt}; \bm{\theta}_0, \eta_0)] = \bm{0}\).
The score satisfies Neyman orthogonality: $\left.\frac{\partial}{\partial t} \mathbb{E}\left[\psi(\bm{W}_{jt}; \bm{\theta}_0, \ell_0 + t h_\ell, \bm{m}_0)\right]\right|_{t=0} = 0,$ 
and $\left.\frac{\partial}{\partial t} \mathbb{E}\left[\psi(\bm{W}_{jt}; \bm{\theta}_0, \ell_0, \bm{m}_0 + t h_m)\right]\right|_{t=0} = 0$ for regular perturbations \(h_\ell, h_m\), ensuring robustness to nuisance estimation errors.
The nuance parameters are estimated via plug-in method that employs machine learning techniques. In our case, we  employ an ensemble learning method that combines multiple decision trees to produce robust and flexible predictions for regression tasks. Each tree is trained on a bootstrap sample of the data, and predictions are aggregated to reduce variance and improve generalization.

Relative to \ac{ols}, \ac{dml} allows $\bm{X}_{jt}$ to enter through a general, unknown function $g_0(\cdot).$
This method removes first-order sensitivity of the target to nuisance estimation error via orthogonalization.
We mitigate the problem of overfitting through sample splitting and out-of-fold prediction.
Specifically, our estimator is based on a \(K\)-fold cross-fitting procedure following \citet{chernozhukov2018double}:

\begin{enumerate}
\item \textbf{Split.} Randomly partition the sample \(\{1, \dots, n\}\) into \(K\) disjoint folds \(I_1, ..., I_K\), each containing approximately \(n/K\) observations.
\item \textbf{Train nuisances off-fold.} For each fold \(k\), fit \(\hat{\ell}^{(-k)}(\bm{x})\) and \(\hat{\bm{m}}^{(-k)}(\bm{x})\) on the complement sample \(I_k^c\) using plug-in machine learning methods.
\item \textbf{Predict on held-out fold.} For each \((j,t) \in I_k\), compute the cross-fit residuals using the out-of-fold nuance estimators:
\[
\tilde{Y}_{jt} = \bar{Y}_{jt} - \hat{\ell}^{(-k)}(\bm{X}_{jt}), \quad \tilde{\bm{S}}_{jt} = \bm{S}_{jt} - \hat{\bm{m}}^{(-k)}(\bm{X}_{jt}),
\]
where \(\hat{\ell}^{(-k)}\) and \(\hat{\bm{m}}^{(-k)}\) are machine learning estimates of \(\ell_0\) and \(\bm{m}_0\) trained on the complement of fold \(k\).
\item \textbf{Aggregate and estimate.} Stack residuals across all folds and compute \(\hat{\bm{\theta}}\) using the following equation:
\begin{equation*}
\hat{\bm{\theta}} = \left( \frac{1}{n} \sum_{j,t} \tilde{\bm{S}}_{jt} \tilde{\bm{S}}_{jt}' \right)^{-1} \left( \frac{1}{n} \sum_{j,t} \tilde{\bm{S}}_{jt} \tilde{Y}_{jt} \right).
\end{equation*}
\item \textbf{Variance.} The asymptotic variance of our estimators is evaluated using the sandwich form:
\begin{align*}
\widehat{\mathrm{Var}}(\hat{\bm{\theta}}) &= \hat{J}^{-1} \hat{\Sigma} \hat{J}^{-1} / n,
\end{align*}
where $\hat{J} = -\frac{1}{n} \sum_{j,t} \tilde{\bm{S}}_{jt} \tilde{\bm{S}}_{jt}',$ the estimated covariance matrix is  
$\hat{\Sigma} = \frac{1}{n} \sum_{j,t} \psi_{jt}(\hat{\bm{\theta}}) \psi_{jt}(\hat{\bm{\theta}})',$ where 
$\psi_{jt}(\hat{\bm{\theta}}) = \tilde{\bm{S}}_{jt} \big( \tilde{Y}_{jt} - \tilde{\bm{S}}_{jt}' \hat{\bm{\theta}} \big).$
\end{enumerate}
\clearpage 
\section{Data clean-up and exclusion process}
\label{sec:apx_data_cleanup}

This appendix explains the steps taken to turn the original data into the net sample used in the baseline analysis.
The descriptions below closely follow Table~\ref{tab:ab_exclusion}.
The raw 2013--2014 sample consists of 36,576 request-level observations.
Subsidy requests could be for one year, a fraction of a year, or multiple years. 
To mitigate the potential biases that this variability could induce, we discarded non-annual contract. 
Table \ref{tab:robustness_combined} column~r3 confirms that this exclusion step is inconsequential. 

\begin{table}[htbp]
\centering
\caption{Data clean-up steps.}
\label{tab:ab_exclusion}
\begin{minipage}{\columnwidth}
\centering
\small
\begin{tabular*}{\linewidth}{@{\extracolsep{\fill}} lc}
\Xhline{2pt}
Data clean-up step & Sample size \\ \hline
Original sample & 36,576 \\
Limit to annual contracts & 13,370 \\
Limit category: Data; Leased/Tariffed Facilities or Services & 13,350 \\
Drop grandfathered: P1 obs with missing price info & 12,080 \\
Removed requests that were submitted or withdrawn & 12,038 \\
Remove cases for which we could not recover price & 11,798 \\
Drop Alaska & 11,308 \\
Drop if speed is missing & 11,290 \\
Limit to valid speed units & 11,245 \\
Keep speed above 50Kbps & 11,234 \\
Limit to speeds between 1-100 Mbps & 8,050 \\
Drop subsidy = 0 & 7,780 \\
One observation per HCP-year (TWFE aggregation) & 4,231 \\
Limit to HCPs that appear in both years & 1,940 \\ \Xhline{2pt}
\end{tabular*}
\vspace{2pt}
\parbox{\linewidth}{\scriptsize\justifying\textit{Notes:} Each row reports the number of remaining observations after the stated exclusion step. Counts are request-level through the speed and subsidy filters and become HCP-year-level after aggregation. All exclusion criteria are applied to every year in the data, but sample sizes shown here are restricted to the 2013--2014 baseline year pair used in the TWFE analysis.}
\end{minipage}
\end{table}

While \p2/\p2c\ subsidizes only internet plans, \p1\ may also cover other categories such as equipment, construction, infrastructure, management, and maintenance.
To avoid contamination, we limit the \p1\ sample to two categories of expense that refer to the cost of internet plans: \enquote{Data} or \enquote{Leased/Tariffed Facilities or Services}.

Since the inception of the program in 1997, the \ac{fcc} has reclassified rurality boundaries several times. 
Each reclassification causes two events: 
some formerly urban \acp{hcp} would become rural and eligible for subsidies, whereas some formerly rural \acp{hcp} would become urban.
\acp{hcp} in the latter group may be grandfathered to remain eligible. 
If the latter group chooses to remain in \p1, they are no longer obligated to report a comparison of rural and urban prices because they are now located in an urban area. 
Given that in \p1, the difference between rural and urban prices is key to recovering the subsidy amount, and it is unobservable for grandfathered \acp{hcp}, we discarded grandfathered observations in \p1.

There are multiple values for request status in the data.
Requests are initially ``submitted''.
A request may be a ``duplicate'' of another request.
\acp{hcp} may choose to ``withdraw'' a request.
Otherwise, a request undergoes processing to be eventually ``approved'', ``partially approved'', or ``denied''.
A ``committed'' status means that the request has been approved and the fund disbursement is done or underway.
That said, we limited the sample to ``committed'' requests.

Next, we limit the sample to observations for which we were able to recover both subsidy and price. 
The recovery process is different for \p1\ vs \p2/\p2c. 
Under \p1, a subsidy request must mention the rural and urban price. 
We set the price equal to the rural price, and we set the subsidy equal to the difference between the rural and urban prices. 
Under \p2/\p2c, the request specifies the subsidy amount. 
In this case, we take the subsidy amount as given, and we set the price equal to the subsidy divided by 0.65. 

Next, we drop \acp{hcp} that are located in Alaska. 
Alaskan observations include many outliers featuring extremely high prices and very low speeds. 
This is an artifact of Alaska's unique demography and geography. 
Alaska has the 4th lowest population among \ac{us} states, but it is the largest state by area (2.47 times the size of the second largest state, Texas).%
\footnote{\href{https://en.wikipedia.org/wiki/List_of_U.S._states_and_territories_by_area}{List of U.S. states and territories by area}, \href{https://en.wikipedia.org/wiki/List_of_U.S._states_and_territories_by_population}{List of U.S. states and territories by population}} 
While the average population density in the \ac{us} is 37 people per square kilometer, the population density in Alaska is 0.50.%
\footnote{\href{https://en.wikipedia.org/wiki/List_of_U.S._states_and_territories_by_population_density}{List of U.S. states and territories by population density}}
In per-capita terms, it is prohibitively costly to provide landline internet connection to such a sparsely populated state across vast geographies and rough terrain. 
So, expensive technologies such as satellite and \ac{mpls}, that make up only 5\% of users in other states, make up 81\% of the users in Alaska. 
As a result, Alaskan internet is 20 times costlier than the national average and it ranks last in terms of speed. 
That said, following \citep{rabbani_bb1, rabbani_bb2}, we exclude Alaska from our baseline analysis to ensure that the results are not influenced by the anomalies of Alaskan internet.
Nevertheless, the exclusion of Alaska is inconsequential (Table \ref{tab:robustness_combined}, column~r5).

Throughout this paper, we take download speed as the measure of internet speed, i.e., the measure of the quantity of consumption of internet. 
This assumption was made to overcome a data limitation. 
The data reports internet speed differently under each program. 
Under \p2/\p2c, both download speed and upload speed are reported. 
Under \p1, the data reports bandwidth as a single measure of speed. 
We believe this corresponds to download speed because it is usually the case that, for a given internet plan, download speed is greater than or equal to the upload speed.%
\footnote{Based on the \href{https://itif.org/publications/2021/05/12/broadband-myth-series-do-we-need-symmetrical-upload-and-download-speeds/}{Information Technology and Innovation Foundation}}
We verified that this condition holds in 98.66\% of the observable sample.
That said, if all programs report download speed, and some programs (\p1) do not report upload speed, then the only available measure of the quantity of consumption would be download speed.

However, recovering download speed was a separate challenge. 
The data reports bandwidth and download speed as a free-form user input that combines speed and speed unit. 
This enables non-standard inputs (reporting in gigabits per second, megabits per second, etc.) and erroneous data entry such as invalid speed units. 
To address this, we limited the sample to observations for which we could identify one of the following speed units: 
``GB'', ``Gbps'', ``KB'', ``Kbps'', ``MB'', ``Mbps'', or ``Gigabyte per second''. 
Then we converted all speed units to the \ac{mbps} equivalent. 
We discarded observations for which speed was missing or the speed unit was not recoverable.
To ensure that this data exclusion step is not distorting the data, Table \ref{tab:robustness_combined}, column~(r2) limits the sample to observations that verbatim specified speed in ``\ac{mbps}''. The results confirm the baseline findings.

Since the program still subsidizes legacy technologies, very low speeds, such as dial-up connections, are commonly found in the data.
On the other hand, some \acp{hcp} have established data centers with extremely high speeds.
We limit the sample to speeds that are in 1--100\ac{mbps}.
The robustness check in Table \ref{tab:robustness_combined}, column~(r4) shows that the results hold if all speed levels are included.

Next, we drop observations for which the subsidy amount is zero, as these do not represent subsidized internet access.
The next step is the process of aggregating request-level data to \ac{hcp}-level. 
A given \ac{hcp} is allowed to request as many lines of internet subsidy as they need. Many \acp{hcp} have tens or hundreds of concurrent lines of subsidy. 
The data reports year and \ac{hcp} ID. 
Using these two measures, we aggregate all requests of the same \ac{hcp} in the same year into one observation. 
That is, the \ac{hcp}-level internet speed is the sum of internet speeds that share the same \ac{hcp} ID and year. 
A similar approach defines \ac{hcp}-level measures of price and subsidy. 

After aggregating to the \ac{hcp}-year level, the 2013--2014 sample contains 4,231 observations.
But not all \acp{hcp} that appear in 2013 appear in 2014, and vice versa.
To create a balanced panel, we limit the sample to \acp{hcp} that appear in both years, yielding a net sample of 1,940 observations.

\clearpage 
\section{Robustness checks}
\label{sec:robust}

This appendix discusses the robustness of the baseline findings to alterations to the sample and specification.
Table~\ref{tab:robustness_combined} reports the results for nine variants, each modifying one aspect of the baseline.
The columns are discussed in order below.

\begin{table}[htbp]
\centering
\caption{Robustness checks: continuous treatment, 2014.}
\label{tab:robustness_combined}
\small
\adjustbox{max width=\linewidth}{%
\begin{tabular}{lccccccccc}
\Xhline{2pt}
 & \shortstack[c]{\textbf{Excluding}\\\textbf{medical}\\\textbf{schools}\\{\scriptsize (r1)}} & \shortstack[c]{\textbf{Mbps}\\\textbf{speeds}\\\textbf{only}\\{\scriptsize (r2)}} & \shortstack[c]{\textbf{All}\\\textbf{contract}\\\textbf{durations}\\{\scriptsize (r3)}} & \shortstack[c]{\textbf{All}\\\textbf{speed}\\\textbf{ranges}\\{\scriptsize (r4)}} & \shortstack[c]{\textbf{Including}\\\textbf{Alaska}\\{\scriptsize (r5)}} & \shortstack[c]{\textbf{Rural}\\\textbf{health}\\\textbf{clinics}\\{\scriptsize (r6.1)}} & \shortstack[c]{\textbf{Non-}\\\textbf{profit}\\\textbf{hospitals}\\{\scriptsize (r6.2)}} & \shortstack[c]{\textbf{Community}\\\textbf{health}\\\textbf{centers}\\{\scriptsize (r6.3)}} & \shortstack[c]{\textbf{Ethernet}\\\textbf{only}\\{\scriptsize (r7)}} \\
\hline
\multicolumn{10}{c}{\textbf{Panel~A: \boldmath$\ln(\text{price})$}} \\
\hline
$\tau_{01}$ & -1.252*** & -1.821*** & -0.878*** & -1.205*** & -1.274*** & -0.692*** & -0.982*** & -4.693*** & -1.895*** \\
 & (0.086) & (0.066) & (0.062) & (0.071) & (0.078) & (0.134) & (0.100) & (0.439) & (0.368) \\
 & [0.000] & [0.000] & [0.000] & [0.000] & [0.000] & [0.000] & [0.000] & [0.000] & [0.000] \\
$\tau_{02}$ & 0.555*** & 0.446*** & -0.223** & 0.497*** & 0.548*** & 0.518*** & 0.735*** & -0.908 &  \\
 & (0.136) & (0.071) & (0.091) & (0.142) & (0.128) & (0.108) & (0.208) & (0.743) &  \\
 & [0.000] & [0.000] & [0.014] & [0.000] & [0.000] & [0.000] & [0.000] & [0.225] &  \\
$\tau_{02} - \tau_{01}$ & 1.806*** & 2.267*** & 0.655*** & 1.701*** & 1.822*** & 1.210*** & 1.717*** & 3.785*** &  \\
 & (0.147) & (0.092) & (0.098) & (0.149) & (0.139) & (0.153) & (0.208) & (0.923) &  \\
 & [0.000] & [0.000] & [0.000] & [0.000] & [0.000] & [0.000] & [0.000] & [0.000] &  \\
$R^2$ & 0.385 & 0.686 & 0.467 & 0.495 & 0.388 & 0.423 & 0.526 & 0.848 & 0.460 \\
\hline
\multicolumn{10}{c}{\textbf{Panel~B: \boldmath$\ln(\text{subsidy})$}} \\
\hline
$\tau_{01}$ & -1.259*** & -1.778*** & -0.355*** & -1.279*** & -1.277*** & -0.800*** & -0.883*** & -5.183*** & -2.093*** \\
 & (0.103) & (0.083) & (0.041) & (0.084) & (0.094) & (0.164) & (0.129) & (0.477) & (0.416) \\
 & [0.000] & [0.000] & [0.000] & [0.000] & [0.000] & [0.000] & [0.000] & [0.000] & [0.000] \\
$\tau_{02}$ & 0.283* & 0.167* & 0.146** & 0.252 & 0.287* & 0.192 & 0.558** & -1.237 &  \\
 & (0.164) & (0.090) & (0.060) & (0.169) & (0.154) & (0.132) & (0.268) & (0.807) &  \\
 & [0.085] & [0.063] & [0.015] & [0.135] & [0.063] & [0.146] & [0.038] & [0.129] &  \\
$\tau_{02} - \tau_{01}$ & 1.542*** & 1.945*** & 0.501*** & 1.531*** & 1.564*** & 0.993*** & 1.441*** & 3.946*** &  \\
 & (0.178) & (0.116) & (0.065) & (0.176) & (0.167) & (0.187) & (0.269) & (1.002) &  \\
 & [0.000] & [0.000] & [0.000] & [0.000] & [0.000] & [0.000] & [0.000] & [0.000] &  \\
$R^2$ & 0.310 & 0.564 & 0.457 & 0.464 & 0.313 & 0.273 & 0.435 & 0.850 & 0.420 \\
\hline
\multicolumn{10}{c}{\textbf{Panel~C: \boldmath$\ln(\text{HCP net cost})$}} \\
\hline
$\tau_{01}$ & -0.720*** & -1.541*** & -1.019*** & -0.615*** & -0.752*** & -0.293 & -0.777*** & -2.263*** & -1.136*** \\
 & (0.075) & (0.092) & (0.104) & (0.064) & (0.069) & (0.203) & (0.102) & (0.343) & (0.390) \\
 & [0.000] & [0.000] & [0.000] & [0.000] & [0.000] & [0.149] & [0.000] & [0.000] & [0.004] \\
$\tau_{02}$ & 1.739*** & 1.674*** & 0.078 & 1.631*** & 1.699*** & 1.868*** & 1.634*** & 0.586 &  \\
 & (0.119) & (0.100) & (0.151) & (0.127) & (0.113) & (0.163) & (0.213) & (0.580) &  \\
 & [0.000] & [0.000] & [0.604] & [0.000] & [0.000] & [0.000] & [0.000] & [0.315] &  \\
$\tau_{02} - \tau_{01}$ & 2.459*** & 3.214*** & 1.098*** & 2.247*** & 2.452*** & 2.161*** & 2.411*** & 2.850*** &  \\
 & (0.128) & (0.129) & (0.163) & (0.133) & (0.123) & (0.231) & (0.213) & (0.721) &  \\
 & [0.000] & [0.000] & [0.000] & [0.000] & [0.000] & [0.000] & [0.000] & [0.000] &  \\
$R^2$ & 0.433 & 0.589 & 0.406 & 0.506 & 0.424 & 0.496 & 0.494 & 0.710 & 0.384 \\
\hline
$N$ & 1,934 & 1,432 & 4,882 & 2,506 & 2,238 & 530 & 804 & 192 & 218 \\
\Xhline{2pt}
\end{tabular}}
\vspace{2pt}
\parbox{\linewidth}{\scriptsize\justifying\textit{Notes:} Each column reports TWFE continuous treatment estimates for 2014 under a different sample restriction. (r1)~excludes medical schools; (r2)~restricts to Mbps-unit speeds; (r3)~includes all contract durations; (r4)~removes speed range restrictions; (r5)~includes Alaska; (r6.1)--(r6.3)~restrict to a single HCP entity type (rural health clinics, non-profit hospitals, community health centers); (r7)~restricts to Ethernet service type. $\tau_{01}$ is the treatment effect of switching from \p1{} to \p2{}, and $\tau_{02}$ is the effect of switching from \p1{} to \p2c{}. $\tau_{02} - \tau_{01}$ is the estimated difference. The number of observations is identical across all three panels. Standard errors in parentheses; $p$-values in brackets. Significance: $^{***}\, p < 0.01$, $^{**}\, p < 0.05$, $^{*}\, p < 0.1$.}
\end{table}

\textit{Excluding medical schools (r1).}
The baseline sample includes a small number of post-secondary educational institutions (medical schools).
Because these entities differ from typical \acp{hcp} in size and internet needs, column~(r1) excludes them.
The results are virtually identical to the baseline.

\textit{Mbps speeds only (r2).}
As discussed in Section~\ref{sec:apx_data_cleanup}, converting free-form speed entries to a common unit required assumptions about non-standard speed units (e.g., GB, MB, KB).
If certain errors are clustered in one program, they could introduce bias.
Column~(r2) restricts the sample to observations that verbatim reported speed in \ac{mbps}, eliminating any unit-conversion assumptions.
The results support the baseline findings.

\textit{All contract durations (r3).}
The baseline excludes non-annual subsidy contracts.
For example, if an \ac{hcp} began receiving subsidies mid-year, the baseline discards the partially funded year.
Column~(r3) retains all contract durations.
The sample size substantially rises, and the results remain qualitatively unchanged.

\textit{All speed ranges (r4).}
The baseline limits speeds to 1--100\ac{mbps} to mitigate the effect of outliers at both ends: legacy technologies such as dial-up at the low end, and data-center-grade connections at the high end.
Column~(r4) removes the speed restriction entirely.
If anything, the results become larger and more significant, strongly supporting the baseline findings.

\textit{Including Alaska (r5).}
We excluded Alaska from the baseline for the reasons explained in Section~\ref{sec:apx_data_cleanup}.
Column~(r5) retains Alaskan observations.
The results remain qualitatively intact.

\textit{Single \ac{hcp} types (r6.1--r6.3).}
Different \ac{hcp} types may have different internet needs, and a concentration of specific types in one program could bias the estimates.
To investigate, we estimate the effects within three \ac{hcp} types that have sufficient sample sizes across programs: rural health clinics (r6.1), non-profit hospitals (r6.2), and community health centers (r6.3).
The results corroborate the baseline findings.

\textit{Ethernet only (r7).}
The baseline keeps all internet service types, listed in Table~\ref{tab:bb_summary_pooledOLS}.
Some technologies may be inherently more expensive for reasons such as reliability or responsiveness, and clustering of expensive technologies in one program could induce bias.
Column~(r7) limits the data to Ethernet, the most common service type found across all three programs.
The much smaller sample reduces statistical power, and several coefficients cannot be estimated.
Yet, to the extent that estimates could be produced, the results confirm the baseline findings.

\clearpage
\section{Supporting tables and figures}
\label{sec:tab_fig}

\begin{table}[htbp]
\centering
\caption{TWFE continuous regression results by year pair.}
\label{tab:bj_twfe_allyears}
\small
\adjustbox{max width=\linewidth}{%
\begin{tabular}{lcccccccc}
\Xhline{2pt}
 & 2014 & 2015 & 2016 & 2017 & 2018 & 2019 & 2020 & 2021 \\ \hline
\multicolumn{9}{c}{\textbf{Panel A: $\ln(\text{price})$}} \\
\hline
$\tau_{01}$ & -1.249*** & -1.275*** & -1.044*** & -0.705*** & -0.499*** & -0.869*** & -0.949*** & -0.595*** \\
 & (0.085) & (0.086) & (0.101) & (0.083) & (0.062) & (0.133) & (0.179) & (0.170) \\
 & [0.000] & [0.000] & [0.000] & [0.000] & [0.000] & [0.000] & [0.000] & [0.000] \\
$\tau_{02}$ & 0.556*** & 0.008 & -0.896*** & 0.514*** & -0.328*** & -0.344 & -0.873*** & -0.594** \\
 & (0.136) & (0.143) & (0.295) & (0.162) & (0.073) & (0.311) & (0.214) & (0.260) \\
 & [0.000] & [0.953] & [0.002] & [0.002] & [0.000] & [0.269] & [0.000] & [0.022] \\
$\tau_{02} - \tau_{01}$ & 1.804*** & 1.284*** & 0.149 & 1.219*** & 0.171*** & 0.525* & 0.076 & 0.001 \\
 & (0.147) & (0.162) & (0.305) & (0.169) & (0.069) & (0.331) & (0.267) & (0.292) \\
 & [0.000] & [0.000] & [0.313] & [0.000] & [0.006] & [0.056] & [0.389] & [0.498] \\
$N$ & 1,940 & 3,658 & 3,950 & 4,388 & 5,832 & 7,004 & 7,290 & 7,812 \\
$R^2$ & 0.387 & 0.410 & 0.409 & 0.366 & 0.362 & 0.271 & 0.270 & 0.278 \\
\hline
\multicolumn{9}{c}{\textbf{Panel B: $\ln(\text{subsidy})$}} \\
\hline
$\tau_{01}$ & -1.262*** & -1.398*** & -1.118*** & -0.734*** & -0.563*** & -0.951*** & -0.735*** & -0.757*** \\
 & (0.102) & (0.093) & (0.105) & (0.086) & (0.063) & (0.134) & (0.181) & (0.170) \\
 & [0.000] & [0.000] & [0.000] & [0.000] & [0.000] & [0.000] & [0.000] & [0.000] \\
$\tau_{02}$ & 0.284* & 0.201 & -1.090*** & 0.474*** & -0.507*** & -0.623** & -0.850*** & -0.693*** \\
 & (0.164) & (0.153) & (0.306) & (0.168) & (0.074) & (0.314) & (0.217) & (0.261) \\
 & [0.084] & [0.191] & [0.000] & [0.005] & [0.000] & [0.047] & [0.000] & [0.008] \\
$\tau_{02} - \tau_{01}$ & 1.546*** & 1.599*** & 0.029 & 1.208*** & 0.056 & 0.328 & -0.115 & 0.065 \\
 & (0.177) & (0.174) & (0.317) & (0.175) & (0.070) & (0.334) & (0.271) & (0.293) \\
 & [0.000] & [0.000] & [0.464] & [0.000] & [0.210] & [0.163] & [0.665] & [0.413] \\
$N$ & 1,940 & 3,658 & 3,950 & 4,388 & 5,832 & 7,004 & 7,290 & 7,812 \\
$R^2$ & 0.312 & 0.402 & 0.403 & 0.356 & 0.387 & 0.269 & 0.266 & 0.279 \\
\hline
\multicolumn{9}{c}{\textbf{Panel C: $\ln(\text{HCP net cost})$}} \\
\hline
$\tau_{01}$ & -0.704*** & -0.390*** & -0.578*** & -0.285*** & -0.084 & -0.524*** & -0.424** & 0.294* \\
 & (0.074) & (0.088) & (0.098) & (0.091) & (0.071) & (0.138) & (0.180) & (0.170) \\
 & [0.000] & [0.000] & [0.000] & [0.002] & [0.233] & [0.000] & [0.019] & [0.084] \\
$\tau_{02}$ & 1.740*** & 0.100 & 0.263 & 0.848*** & 0.247*** & 0.958*** & -0.370* & 0.053 \\
 & (0.119) & (0.145) & (0.288) & (0.177) & (0.082) & (0.324) & (0.216) & (0.261) \\
 & [0.000] & [0.489] & [0.363] & [0.000] & [0.003] & [0.003] & [0.086] & [0.839] \\
$\tau_{02} - \tau_{01}$ & 2.445*** & 0.490*** & 0.841*** & 1.133*** & 0.331*** & 1.482*** & 0.054 & -0.241 \\
 & (0.128) & (0.165) & (0.298) & (0.185) & (0.078) & (0.345) & (0.270) & (0.292) \\
 & [0.000] & [0.001] & [0.002] & [0.000] & [0.000] & [0.000] & [0.420] & [0.796] \\
$N$ & 1,940 & 3,658 & 3,950 & 4,388 & 5,832 & 7,004 & 7,290 & 7,812 \\
$R^2$ & 0.431 & 0.321 & 0.375 & 0.313 & 0.305 & 0.262 & 0.255 & 0.273 \\
\Xhline{2pt}
\end{tabular}}
\vspace{2pt}
\parbox{\linewidth}{\scriptsize\justifying\textit{Notes:} Each column reports TWFE estimates with continuous treatment shares for the indicated year pair (base year $t-1$, treatment year $t$). $\tau_{01}$ is the estimated effect of switching from \p1{} to \p2{}, and $\tau_{02}$ is the effect of switching from \p1{} to \p2c{}. The row $\tau_{02} - \tau_{01}$ reports the estimated difference. Standard errors in parentheses; $p$-values in brackets. Significance: $^{***}\, p < 0.01$, $^{**}\, p < 0.05$, $^{*}\, p < 0.1$.}
\end{table}

\begin{table}[htbp]
\centering
\caption{Goodness of fit: price--speed functional forms, 2014.}
\label{tab:bk_model_comparison}
\begin{minipage}{\columnwidth}
\centering
\fontsize{10}{12}\selectfont
\resizebox{\ifdim\width>\linewidth\linewidth\else\width\fi}{!}{%
\begin{tabular}{rlccccccc}
\Xhline{2pt}
 & & & \multicolumn{2}{c}{\p1} & \multicolumn{2}{c}{\p2} & \multicolumn{2}{c}{\p2c} \\
\cline{4-5}\cline{6-7}\cline{8-9}
Rank & Model & $k$ & Adj.~$R^2$ & RMSE & Adj.~$R^2$ & RMSE & Adj.~$R^2$ & RMSE \\ \hline
\multicolumn{9}{c}{\textbf{Panel A: HCP level} ($N$: \p1=643, \p2=419, \p2c=43)} \\ \hline
1 & $P = a + bS + cS^2$ & 3 & 0.273 & \$22,930 & 0.009 & \$5,499 & 0.424 & \$21,209 \\
2 & $P = a + b \cdot \ln(S)$ & 2 & 0.242 & \$23,434 & 0.008 & \$5,508 & 0.470 & \$20,592 \\
3 & $P = a + b \cdot \sqrt{S}$ & 2 & 0.201 & \$24,058 & 0.007 & \$5,510 & 0.466 & \$20,669 \\
4 & $\ln(P) = a + b\ln(S) + c[\ln(S)]^2$ & 3 & 0.168 & \$24,538 & -0.121 & \$5,848 & 0.391 & \$21,814 \\
5 & $\ln(P) = a + b \cdot \ln(S)$ & 2 & 0.160 & \$24,680 & -0.121 & \$5,856 & 0.424 & \$21,472 \\
6 & $P = F + c \cdot S$ & 2 & 0.138 & \$24,998 & 0.005 & \$5,516 & 0.430 & \$21,355 \\
7 & $\ln(P) = a + b \cdot \sqrt{S}$ & 2 & 0.098 & \$25,571 & -0.120 & \$5,852 & 0.362 & \$22,607 \\
8 & $\ln(P) = a + b \cdot S$ & 2 & 0.033 & \$26,470 & -0.123 & \$5,860 & 0.318 & \$23,368 \\
\hline
\multicolumn{9}{c}{\textbf{Panel B: Request level} ($N$: \p1=1,305, \p2=973, \p2c=56)} \\ \hline
1 & $P = a + bS + cS^2$ & 3 & 0.318 & \$19,516 & 0.019 & \$6,457 & 0.558 & \$20,644 \\
2 & $P = a + b \cdot \ln(S)$ & 2 & 0.271 & \$20,180 & -0.001 & \$6,524 & 0.576 & \$20,419 \\
3 & $P = a + b \cdot \sqrt{S}$ & 2 & 0.218 & \$20,903 & -0.000 & \$6,522 & 0.544 & \$21,170 \\
4 & $\ln(P) = a + b\ln(S) + c[\ln(S)]^2$ & 3 & 0.210 & \$20,999 & -0.135 & \$6,943 & 0.521 & \$21,496 \\
5 & $\ln(P) = a + b \cdot \ln(S)$ & 2 & 0.178 & \$21,425 & -0.147 & \$6,985 & 0.495 & \$22,290 \\
6 & $P = F + c \cdot S$ & 2 & 0.150 & \$21,784 & -0.001 & \$6,525 & 0.488 & \$22,426 \\
7 & $\ln(P) = a + b \cdot \sqrt{S}$ & 2 & 0.105 & \$22,362 & -0.143 & \$6,973 & 0.418 & \$23,912 \\
8 & $\ln(P) = a + b \cdot S$ & 2 & 0.043 & \$23,124 & -0.145 & \$6,978 & 0.389 & \$24,508 \\
\Xhline{2pt}
\end{tabular}}
\vspace{2pt}
\parbox{\linewidth}{\scriptsize\justifying\textit{Notes:} The table compares eight parametric specifications for the price--speed relationship, fitted separately for each program using 2014 data for the 970 baseline HCPs. Panel~A aggregates to one observation per HCP--program (mean price and mean speed across requests). Panel~B uses request-level observations. $P$ denotes the annual price (\$) and $S$ the download speed (Mbps). $k$ is the number of estimated parameters. Adjusted~$R^2$ and RMSE are computed on the original price scale for all models (log-space predictions are back-transformed via exponentiation). Because the back-transformed predictions from log-space models are not the OLS solution on the price scale, their price-scale $R^2$ can be negative; a negative value indicates that the back-transformed model predicts prices less accurately than the sample mean. Models are ranked by adjusted~$R^2$ for \p1{} in Panel~A.}
\end{minipage}
\end{table}

\begin{table}[h!]
\centering
\caption{Box-Cox test for the price--speed functional form, 2014.}
\label{tab:bk_boxcox}
\begin{minipage}{\columnwidth}
\centering
\small
\begin{tabular*}{\linewidth}{@{\extracolsep{\fill}} lccccccc}
\Xhline{2pt}
 & & & \multicolumn{2}{c}{$H_0\colon \lambda = 0$} & \multicolumn{2}{c}{$H_0\colon \lambda = 1$} \\
\cmidrule(lr){4-5} \cmidrule(lr){6-7}
Program & $N$ & $\hat{\lambda}$ & LR stat. & $p$-value & LR stat. & $p$-value \\ \hline
\multicolumn{7}{c}{\textbf{Panel A: HCP level}} \\ \hline
\P1 & 643 & -0.170 & 22.60 & $<$0.001 & 1217.22 & $<$0.001 \\
\P2 & 419 & 0.200 & 21.32 & $<$0.001 & 321.35 & $<$0.001 \\
\P2c & 43 & 0.190 & 1.59 & 0.208 & 29.08 & $<$0.001 \\
\hline
\multicolumn{7}{c}{\textbf{Panel B: Request level}} \\ \hline
\P1 & 1,305 & 0.070 & 12.43 & $<$0.001 & 2323.49 & $<$0.001 \\
\P2 & 973 & 0.140 & 28.02 & $<$0.001 & 1054.35 & $<$0.001 \\
\P2c & 56 & 0.240 & 4.22 & 0.040 & 39.52 & $<$0.001 \\
\Xhline{2pt}
\end{tabular*}
\vspace{2pt}
\parbox{\linewidth}{\scriptsize\justifying\textit{Notes:} This table reports the Box-Cox test for the transformation of the dependent variable in the price--speed hedonic regression $P^{(\lambda)} = a + b\,\ln(S) + \varepsilon$, where $P^{(\lambda)} = (P^{\lambda} - 1)/\lambda$ \citep{boxcox1964}. The transformation nests the log-log model ($\lambda = 0$, implying $\ln P = a + b\,\ln S$) and the lin-log model ($\lambda = 1$, implying $P = a + b\,\ln S$) as special cases. $\hat{\lambda}$ is the maximum likelihood estimate. The LR statistic is $2(\ell(\hat{\lambda}) - \ell(\lambda_0))$, distributed $\chi^2(1)$ under the null. Rejection of $\lambda = 0$ favors the lin-log over the log-log; failure to reject $\lambda = 1$ is consistent with the lin-log. Data are 2014 observations for the 970 baseline HCPs, matching Table~\ref{tab:bk_model_comparison}. Panel~A aggregates to one observation per HCP--program; Panel~B uses request-level observations.}
\end{minipage}
\end{table}

\begin{table}[htbp]
\centering
\caption{Combined regression results (levels), 2014.}
\label{tab:bj_main_reg_levels}
\fontsize{10}{12}\selectfont
\begin{tabular*}{\linewidth}{@{\extracolsep{\fill}} ll ccc}
\Xhline{2pt}
 & & \textbf{Price} & \textbf{Subsidy} & \textbf{HCP net cost} \\ \hline
$\tau_{01}$ & Coeff. & -28.855*** & -22.581*** & -6.274*** \\
 & SE & (3.795) & (3.433) & (0.829) \\
 & $p$-value & 0.000 & 0.000 & 0.000 \\
 & Elasticity & -0.933 & -0.956 & -0.861 \\
$\tau_{02}$ & Coeff. & 20.171*** & 7.741 & 12.431*** \\
 & SE & (6.088) & (5.508) & (1.329) \\
 & $p$-value & 0.001 & 0.160 & 0.000 \\
 & Elasticity & 0.652 & 0.328 & 1.706 \\
$\tau_{02} - \tau_{01}$ & Coeff. & 49.027*** & 30.322*** & 18.705*** \\
 & SE & (6.572) & (5.945) & (1.435) \\
 & $p$-value & 0.000 & 0.000 & 0.000 \\
 & Elasticity & 1.586 & 1.283 & 2.568 \\
$N$ & & 1,940 & 1,940 & 1,940 \\
$R^2$ & & 0.192 & 0.135 & 0.290 \\
\Xhline{2pt}
\end{tabular*}
\vspace{2pt}
\parbox{\linewidth}{\scriptsize\justifying\textit{Notes:} All columns report two-way fixed effects estimates with continuous treatment shares. The dependent variables are price, subsidy, and \ac{hcp} net cost. Coefficients and standard errors are reported in thousands of dollars. Elasticity is the implied semi-elasticity, computed as the coefficient divided by the sample mean of the dependent variable ($\hat{\beta}/\bar{y}$). $\tau_{01}$ is the estimated treatment effect of switching from \p1{} to \p2{}, and $\tau_{02}$ is the effect of switching from \p1{} to \p2c{}. The row $\tau_{02} - \tau_{01}$ reports the estimated difference. A log-linear specification of this table (with $\ln(\text{price})$, $\ln(\text{subsidy})$, and $\ln(\text{HCP net cost})$ as dependent variables) is reported in Table~\ref{tab:bj_main_reg_combined} in the main text. Significance: $^{***}\, p < 0.01$, $^{**}\, p < 0.05$, $^{*}\, p < 0.1$.}
\end{table}

\begin{table}[h!]
\centering
\caption{Semi-elasticity comparison: log-linear vs.\ linear models (Panel~A, 2014).}
\label{tab:bj_elasticity_comparison}
\small
\begin{tabular*}{\linewidth}{@{\extracolsep{\fill}} l cc cc cc}
\Xhline{2pt}
 & \multicolumn{2}{c}{\textbf{Price}} & \multicolumn{2}{c}{\textbf{Subsidy}} & \multicolumn{2}{c}{\textbf{\ac{hcp} net cost}} \\
\cmidrule(lr){2-3} \cmidrule(lr){4-5} \cmidrule(lr){6-7}
 & Log & Level$/\bar{y}$ & Log & Level$/\bar{y}$ & Log & Level$/\bar{y}$ \\ \hline
$\tau_{01}$ & -1.249 & -0.933 & -1.262 & -0.956 & -0.704 & -0.861 \\
$\tau_{02}$ & 0.556 & 0.652 & 0.284 & 0.328 & 1.740 & 1.706 \\
$\tau_{02} - \tau_{01}$ & 1.804 & 1.586 & 1.546 & 1.283 & 2.445 & 2.568 \\
\addlinespace[4pt]
$\bar{y}$ &  & 30,915 &  & 23,630 &  & 7,285 \\
\Xhline{2pt}
\end{tabular*}
\vspace{2pt}
\parbox{\linewidth}{\scriptsize\justifying\textit{Notes:} This table compares semi-elasticities from two functional forms, both estimated with TWFE and continuous treatment shares (Panel~A). The ``Log'' column reports the coefficient from the log-linear model (Table~\ref{tab:bj_main_reg_combined} in the appendix), which equals the semi-elasticity $\partial \ln y / \partial x$ directly. The ``Level$/\bar{y}$'' column divides the coefficient from the linear model (Table~\ref{tab:bj_main_reg_levels}) by the sample mean~$\bar{y}$, giving the implied semi-elasticity $\hat{\beta}/\bar{y}$ evaluated at the mean. $\bar{y}$ is the arithmetic mean of the dependent variable across both pre- and post-treatment periods for 2014.}
\end{table}

\clearpage
\section{Influence diagnostics}
\label{sec:apx_influence}

This appendix applies Cook's Distance to the baseline 2013--2014 \ac{twfe} specification to assess whether the estimates are driven by a small number of influential observations.
For a fixed-effects panel model, Cook's Distance is computed on the within-transformed (\ac{hcp}-demeaned) data.
Let $D_i$ denote the Cook's Distance for observation $i$.
Following the standard threshold, observations with $D_i > 4/N$ are flagged as influential.
Because influence depends on the outcome variable, we compute Cook's Distance separately for each outcome (ln price, ln subsidy, ln HCP net cost) and take the union of flagged observations across all three equations.
To maintain the balanced panel structure, if either panel row (pre- or post-treatment) of an \ac{hcp} is flagged, we drop the entire \ac{hcp}.

Table~\ref{tab:ap_influence} reports the baseline and trimmed estimates.
Of the 970 baseline HCPs, 143 (14.7\%) have at least one flagged observation.
The flagging rate varies sharply by switching behavior: only 21 of 511 stayers (4.1\%) are flagged, compared to 95 of 416 $\mathcal{P}_1 \to \mathcal{P}_2$ switchers (22.8\%), 26 of 40 $\mathcal{P}_1 \to \mathcal{P}_2^c$ switchers (65.0\%), and 1 of 3 mixed switchers (33.3\%).
This pattern is expected: switchers experience the largest within-\ac{hcp} changes in treatment status and thus in prices, making them high-leverage observations by construction.

\begin{table}[h!]
\centering
\caption{Influence diagnostics: baseline vs.\ trimmed estimates.}
\label{tab:ap_influence}
\begin{minipage}{\columnwidth}
\centering
\small
\begin{tabular*}{\linewidth}{@{\extracolsep{\fill}} lcccccc}
\Xhline{2pt}
 & \multicolumn{2}{c}{\textbf{ln(price)}} & \multicolumn{2}{c}{\textbf{ln(subsidy)}} & \multicolumn{2}{c}{\textbf{ln(HCP net cost)}} \\
\cmidrule(lr){2-3} \cmidrule(lr){4-5} \cmidrule(lr){6-7}
 & Baseline & Trimmed & Baseline & Trimmed & Baseline & Trimmed \\ \hline
$\tau_{01}$ & -1.249*** & -0.830*** & -1.262*** & -0.857*** & -0.704*** & -0.454*** \\
 & (0.085) & (0.050) & (0.102) & (0.069) & (0.074) & (0.067) \\
$\Delta\%$ &  & +33.5\% &  & +32.1\% &  & +35.5\% \\[4pt]
$\tau_{02}$ & 0.556*** & 0.503*** & 0.284* & 0.176 & 1.740*** & 1.723*** \\
 & (0.136) & (0.103) & (0.164) & (0.142) & (0.119) & (0.137) \\
$\Delta\%$ &  & -9.6\% &  & -38.1\% &  & -1.0\% \\
\hline
$R^2$ & 0.387 & 0.597 & 0.312 & 0.474 & 0.431 & 0.481 \\
$N$ & 1,940 & 1,654 & 1,940 & 1,654 & 1,940 & 1,654 \\
\Xhline{2pt}
\end{tabular*}
\vspace{2pt}
\parbox{\linewidth}{\scriptsize\justifying\textit{Notes:} This table reports influence diagnostics for the baseline 2013--2014 TWFE specification. Cook's Distance is computed on the within-transformed (HCP-demeaned) model for each outcome variable. Observations with Cook's $D > 4/N$ are flagged as influential. The ``Trimmed'' columns re-estimate the model after dropping all HCPs with at least one flagged panel-row in any outcome equation. $\Delta\%$ reports the percentage change in the trimmed coefficient relative to baseline. Standard errors in parentheses. *** $p<0.01$, ** $p<0.05$, * $p<0.1$.}
\end{minipage}
\end{table}

After dropping the 143 flagged HCPs, $\tau_{01}$ decreases in magnitude by roughly one-third across all outcomes, but remains significant.
The trimmed estimate of $-0.830$ for ln(price) implies that switching to $\mathcal{P}_2$ reduces prices by approximately 56\% (versus 71\% at baseline).
The $\tau_{02}$ coefficient is more stable.
For ln(\ac{hcp} net cost), the trimmed estimate (1.723) is within 1\% of the baseline (1.740).
This is striking given that 65\% of $\mathcal{P}_2^c$ HCPs were flagged.
The within $R^2$ rises after trimming, confirming that the flagged observations were outliers in the residual space rather than in the treatment-effect space.
In sum, the direction, significance, and economic interpretation of all treatment effects survive the removal of influential observations.

\clearpage
\section{Common-support restriction on speed}
\label{sec:apx_common_support}

The GAM price--speed curves and kernel density plots (Figure~\ref{fig:bd_density_4x4}) show that $\mathcal{P}_2^c$ participants operate in a narrower speed range than those in $\mathcal{P}_1$ or $\mathcal{P}_2$.
If the baseline treatment effects are driven by price variation at extreme speeds where $\mathcal{P}_2^c$ has no representation, the comparison across programs may not be on common support.
This appendix addresses this concern by restricting all programs to the speed domain of $\mathcal{P}_2^c$.
We identify the range of pre-treatment (2013) download speeds among HCPs that participate in $\mathcal{P}_2^c$: [1.5, 139.6] Mbps.
We then drop all HCPs whose 2013 speed falls outside this range, regardless of their program assignment.
This removes 21 of 970 baseline HCPs (2.2\%), all from $\mathcal{P}_1$ and $\mathcal{P}_2$, leaving a sample of 949 HCPs observed on common support. 
Table~\ref{tab:aq_common_support} compares the baseline and restricted estimates.
The coefficients are virtually unchanged and retain their statistical significance.

\begin{table}[h!]
\centering
\caption{Common-support restriction.}
\label{tab:aq_common_support}
\begin{minipage}{\columnwidth}
\centering
\small
\begin{tabular*}{\linewidth}{@{\extracolsep{\fill}} lcccccc}
\Xhline{2pt}
 & \multicolumn{2}{c}{\textbf{ln(price)}} & \multicolumn{2}{c}{\textbf{ln(subsidy)}} & \multicolumn{2}{c}{\textbf{ln(HCP net cost)}} \\
\cmidrule(lr){2-3} \cmidrule(lr){4-5} \cmidrule(lr){6-7}
 & Baseline & Restricted & Baseline & Restricted & Baseline & Restricted \\ \hline
$\tau_{01}$ & -1.249*** & -1.224*** & -1.262*** & -1.250*** & -0.704*** & -0.673*** \\
 & (0.085) & (0.089) & (0.102) & (0.107) & (0.074) & (0.076) \\
$\Delta\%$ &  & +2.0\% &  & +0.9\% &  & +4.5\% \\[4pt]
$\tau_{02}$ & 0.556*** & 0.556*** & 0.284* & 0.288* & 1.740*** & 1.727*** \\
 & (0.136) & (0.137) & (0.164) & (0.166) & (0.119) & (0.117) \\
$\Delta\%$ &  & +0.1\% &  & +1.6\% &  & -0.8\% \\
\hline
$R^2$ & 0.387 & 0.366 & 0.312 & 0.298 & 0.431 & 0.415 \\
$N$ & 1,940 & 1,898 & 1,940 & 1,898 & 1,940 & 1,898 \\
\Xhline{2pt}
\end{tabular*}
\vspace{2pt}
\parbox{\linewidth}{\scriptsize\justifying\textit{Notes:} The ``Restricted'' columns limit the sample to HCPs whose pre-treatment (2013) speed falls within the $\mathcal{P}_2^c$ range: [1.5, 139.6] Mbps. $\Delta\%$ reports the percentage change relative to baseline. Standard errors in parentheses. *** $p<0.01$, ** $p<0.05$, * $p<0.1$.}
\end{minipage}
\end{table}

\clearpage
\section{Coefficient stability: Oster bounds}
\label{sec:apx_oster}

This appendix applies the coefficient stability framework of \citet{oster2019unobservable} to assess whether unobserved time-varying confounders could explain away the estimated treatment effects.
Consider two versions of the baseline \ac{twfe} model.
The \textit{short} model includes only the treatment structure and HCP fixed effects but excludes time-varying controls.
The \textit{long} model adds log speed and log number of requests.
Let $\tilde{\beta}$ and $\hat{\beta}$ denote the treatment coefficient from the short and long models, respectively, and let $\tilde{R}^2$ and $R^2$ denote the corresponding within $R^2$ values.
Under proportional selection, the degree to which selection on unobservables would need to exceed selection on observables to drive the treatment effect to zero is:
\[
\delta = \frac{\hat{\beta} \cdot (R^2 - \tilde{R}^2)}{(\tilde{\beta} - \hat{\beta}) \cdot (R^2_{\max} - R^2)}
\]
where $R^2_{\max} = \min(1.3 \times R^2,\, 1)$ following \citet{oster2019unobservable}. The bias-adjusted coefficient assuming equal proportional selection ($\delta = 1$) is:
\[
\beta^*(\delta = 1) = \hat{\beta} - (\tilde{\beta} - \hat{\beta}) \cdot \frac{R^2_{\max} - R^2}{R^2 - \tilde{R}^2}
\]

Table~\ref{tab:ao_oster_bounds} reports the results ($\tau_{12}$ is omitted because $\mathcal{P}_2$ was introduced in 2014, so no \ac{hcp} could have been on $\mathcal{P}_2$ in 2013).
$\delta$ is negative for every coefficient and outcome. For $\tau_{01}$, the uncontrolled $\tilde{\beta}$ is near zero or positive, but the controlled $\hat{\beta}$ is large and negative, indicating that adding controls \textit{reveals} the treatment effect rather than inflating it.
For unobservables to drive $\hat{\beta}$ to zero, they would have to work in the opposite direction of the observables, which is implausible under standard assumptions \citep{altonji2005selection}.
$\tau_{02}$ is highly stable: $\delta$ ranges from $-54$ to $-435$ across outcomes, and $\beta^*(\delta = 1)$ barely differs from $\hat{\beta}$.
These results complement the logit analysis in Table~\ref{tab:zz4_logit_switching}, which shows that observables contribute virtually no explanatory power to the switching decision beyond the mechanical cost ratio (pseudo $R^2$: 0.031 to 0.037).

\begin{table}[h!]
\centering
\caption{Coefficient stability: Oster (2019) bounds.}
\label{tab:ao_oster_bounds}
\small
\begin{tabular*}{\linewidth}{@{\extracolsep{\fill}} lcccccc}
\Xhline{2pt}
 & \multicolumn{2}{c}{\textbf{Panel A: ln(price)}} & \multicolumn{2}{c}{\textbf{Panel B: ln(subsidy)}} & \multicolumn{2}{c}{\textbf{Panel C: ln(HCP net cost)}} \\
\cmidrule(lr){2-3} \cmidrule(lr){4-5} \cmidrule(lr){6-7}
 & $\tau_{01}$ & $\tau_{02}$ & $\tau_{01}$ & $\tau_{02}$ & $\tau_{01}$ & $\tau_{02}$ \\ \hline
$\tilde{\beta}$ (uncontrolled) & -0.261 & 0.543 & -0.226 & 0.267 & 0.132 & 1.731 \\
$\hat{\beta}$ (controlled) & -1.249*** & 0.556*** & -1.262*** & 0.284* & -0.704*** & 1.740*** \\
 & (0.085) & (0.136) & (0.102) & (0.164) & (0.074) & (0.119) \\
$\tilde{R}^2$ & 0.043 & 0.043 & 0.015 & 0.015 & 0.135 & 0.135 \\
$R^2$ & 0.387 & 0.387 & 0.312 & 0.312 & 0.431 & 0.431 \\
$R^2_{\max}$ & 0.502 & 0.502 & 0.405 & 0.405 & 0.561 & 0.561 \\
$\delta$ & -3.75 & -126.81 & -3.86 & -53.67 & -1.93 & -434.58 \\
$\beta^*(\delta\!=\!1)$ & -1.582 & 0.560 & -1.589 & 0.289 & -1.070 & 1.744 \\
$N$ & 1,940 & 1,940 & 1,940 & 1,940 & 1,940 & 1,940 \\
\Xhline{2pt}
\end{tabular*}
\vspace{2pt}
\parbox{\linewidth}{\scriptsize\justifying\textit{Notes:} Baseline 2013--2014 TWFE specification. $\tilde{\beta}$: coefficient from the short model (no time-varying controls); $\hat{\beta}$: coefficient from the long model (adds log speed and log requests). $R^2_{\max} = \min(1.3 \times R^2, 1)$. $\delta$: proportional selection ratio needed to drive the effect to zero. $\beta^*(\delta\!=\!1)$: bias-adjusted coefficient assuming equal selection. $\tau_{12}$ is omitted because $\mathcal{P}_2$ did not exist in 2013. Standard errors of $\hat{\beta}$ in parentheses. *** $p<0.01$, ** $p<0.05$, * $p<0.1$.}
\end{table}

\end{document}